\documentclass[11pt]{article}

\usepackage[margin=1in]{geometry}
\usepackage{setspace}
\usepackage{amsmath,amssymb,amsthm,mathtools}
\usepackage{mathrsfs}
\usepackage{xcolor}
\usepackage{graphicx}
\newif\ifnativealgorithms
\IfFileExists{algpseudocode.sty}{\nativealgorithmstrue}{\nativealgorithmsfalse}
\ifnativealgorithms\usepackage{algorithm}\usepackage{algpseudocode}\else
  \usepackage{float}
  \floatstyle{ruled}\newfloat{algorithm}{tbp}{loa}\floatname{algorithm}{Algorithm}
  \newcounter{algline}\newcount\algdepth
  \newenvironment{algorithmic}[1][]{\begin{list}{}{\setlength{\leftmargin}{2.4em}\setlength{\itemsep}{2pt}\setlength{\parsep}{0pt}}\setcounter{algline}{0}\algdepth=0}{\end{list}}
  \newcommand{\State}{\stepcounter{algline}\item[\scriptsize\arabic{algline}:]\hspace*{\the\algdepth em}}
  \newcommand{\For}[1]{\State\textbf{for} #1 \textbf{do}\advance\algdepth by 1}
  \newcommand{\EndFor}{\advance\algdepth by -1\State\textbf{end for}}\fi
\usepackage{microtype}
\usepackage{enumitem}

\usepackage[round,authoryear]{natbib}
\usepackage[hidelinks]{hyperref}
\IfFileExists{xurl.sty}{\usepackage{xurl}}{}
\hypersetup{pdftitle={Equilibrium in Multi-Agent Reinforcement Learning},pdfauthor={Maurizio D'Andrea and Bar Light}}

\usepackage{authblk}

\newtheorem{theorem}{Theorem}[section]
\newtheorem{proposition}[theorem]{Proposition}
\newtheorem{lemma}[theorem]{Lemma}
\newtheorem{corollary}[theorem]{Corollary}
\newtheorem{assumption}[theorem]{Assumption}
\theoremstyle{definition}
\newtheorem{definition}[theorem]{Definition}
\newtheorem{example}[theorem]{Example}
\theoremstyle{remark}
\newtheorem{remark}[theorem]{Remark}

\newcommand{\E}{\mathbb{E}}
\newcommand{\one}{\mathbf{1}}
\newcommand{\R}{\mathbb{R}}
\newcommand{\cA}{\mathcal{A}}
\newcommand{\cS}{\mathcal{S}}
\newcommand{\bbP}{\mathbb{P}}
\newcommand{\TV}{\mathrm{TV}}
\newcommand{\proj}{\mathrm{Proj}}
\newcommand{\eps}{\varepsilon}
\newcommand{\GapMCCE}{\operatorname{Gap}_{\mathrm{MCCE}}}
\newcommand{\GapMPE}{\operatorname{Gap}_{\mathrm{MPE}}}

\title{Equilibrium in Multi-Agent Reinforcement Learning}

\author[1]{Maurizio D'Andrea}
\affil[1]{Department of Analytics and Operations\\ National University of Singapore's Business School\\
\texttt{dmaurizio\_1@nus.edu.sg}}

\author[2]{Bar Light}
\affil[2]{Business School and Institute of Operations Research and Analytics, National University of Singapore\\
\texttt{barlight@nus.edu.sg}}

\begin{document}

\maketitle

\begin{abstract}
Standard solution concepts for stochastic games, such as Markov perfect equilibrium and stationary Markov coarse correlated equilibrium, are computationally difficult, and thus, standard decentralized reinforcement-learning algorithms should not generally be expected to converge to them. In this paper, we study the equilibrium generated by such algorithms. In particular, we introduce a new solution concept for stochastic games, Markov Bayes coarse correlated equilibrium (MBCCE), defined as a distribution over states and stationary policy profiles such that, after observing the state but before observing her recommended action, no player can gain by choosing a different current action, with the sampled policy profile governing play thereafter.  We discuss the parallels between MBCCE and coarse correlated equilibrium (CCE) in finite normal-form games and show that MBCCE retains several of its key properties. We then introduce a corresponding regret notion, adaptive Markov coarse regret (AMCR), and show that vanishing AMCR implies that every accumulation point of the empirical distribution of realized states and policy profiles is an MBCCE. Crucially, we show that achieving AMCR reduces to two standard learning tasks: minimizing external regret at each state and accurately evaluating the current joint policy. We then prove that under mild conditions two natural RL algorithmic designs achieve vanishing AMCR: a decentralized asynchronous actor--critic algorithm through a new two-timescale stochastic-approximation analysis, and a standard episodic multi-agent projected policy-gradient method. Hence, both algorithms approach the MBCCE set, and we establish explicit finite-time convergence rates for both. Finally, we relate MBCCE to stationary MCCE and use this connection together with our algorithmic guarantees to identify a natural class of stochastic games in which these algorithms achieve approximate MCCE. Taken together, our results identify MBCCE as a new equilibrium notion for the long-run empirical behavior of standard reinforcement-learning methods in general finite discounted stochastic games.  
\end{abstract}

\section{Introduction}
\label{sec:intro}

A key challenge in multi-agent reinforcement learning is to identify a solution concept that is both strategically meaningful and compatible with decentralized learning.  A Markov perfect equilibrium (MPE) is the  Markovian refinement of Nash equilibrium for stochastic games: it evaluates one stationary product-policy profile and requires optimal behavior at every state. While MPE therefore has a clear strategic interpretation, it is generally not compatible with decentralized learning. In normal-form games, coarse correlated equilibrium (CCE) provides a natural alternative to Nash that is closely tied to learning: no-regret dynamics generate CCE, providing a learning foundation for the concept, and CCE itself is computationally tractable. This suggests seeking an analogous relaxation of MPE in stochastic games. A stationary Markov coarse correlated equilibrium (MCCE), which evaluates one stationary joint-action law, is a natural first candidate for extending the same connection. The analogy, however, generally breaks down. \citet{DaskalakisGolowichZhang2023} show that computing a constant-approximate stationary MCCE is PPAD-hard even in two-player turn-based discounted stochastic games.\footnote{\citet{DaskalakisGolowichZhang2023} also construct an efficient nonstationary Markov CCE. Their policy specifies the future sequence of Markov decision rules in advance. This differs from the learning environments studied here, which are standard in RL, in which future policies are revised after new data are observed and are therefore not known at the initial date.}

This computational difficulty has a natural counterpart from the learning perspective. The no-regret route to stationary MCCE would naturally compare realized play with a fixed Markov-policy deviation in hindsight. In the standard reinforcement-learning setting we study, at each date the players observe the current state, choose actions, receive payoffs, and observe the next state, and use the resulting data to revise their policies. The transition law and payoff functions may be unknown, each player may have limited information about the other players, and because all players learn simultaneously, the policy profile generating the data can continue to change. In such a setting, changing a player's Markov policy generally changes the subsequent state path and therefore the observations used by all players to update their policies. The realized history therefore does not determine the payoff of such a deviation without specifying how the learners would have behaved along the counterfactual history. To the best of our knowledge, whether decentralized learners can achieve sublinear regret against this pathwise Markov-policy regret criterion in general stochastic games in the continuing, partial-feedback setting studied here remains open. Our main question is therefore: in a general stochastic game, what equilibrium property can be guaranteed for the adaptive policy path generated by decentralized tabular learning algorithms, even when the policies do not converge?

To answer this question, we introduce a new distributional equilibrium notion, \emph{Markov Bayes coarse correlated equilibrium} (MBCCE), which extends CCE to stochastic games differently from MCCE, as we discuss in detail in Section~\ref{sec:markov-bcce}. An MBCCE is a joint distribution over a public state and a stationary product-policy profile. The profile is drawn first, and the state is then drawn according to that profile's long-run state frequencies in continuing play or its normalized discounted visits from a fixed initial distribution in episodic play. The players observe the state but not the sampled policy profile, so the state induces a common posterior over continuation policies. A mediator who observes the sampled profile recommends a current action drawn
according to that profile at the observed state. After observing the state but
before seeing her recommendation, a player may instead choose a fixed action at that state. From the next state onward, play follows the sampled stationary
profile. MBCCE requires that, conditional on every state that occurs with positive
probability, no player gains by choosing such a fixed action rather than following the mediator's recommendation.\footnote{The term ``Bayes'' refers to the posterior over continuation policies induced by the observed public state; it does not require private types. The timing is the usual Bayes coarse-obedience timing adapted to a publicly observed state; see \citet{HartlineSyrgkanisTardos2015}.}

The key feature of MBCCE is that each sampled stationary policy profile determines both current play and the continuation value used to evaluate a deviation. Thus, MBCCE preserves the coarse-obedience logic of CCE while accounting for the effect of current actions on future states. 
The statewise statements below concern states that occur with positive probability.
In a one-state stochastic game, MBCCE reduces to ordinary CCE. More generally,
when a unilateral current action does not affect the next-state distribution,
 the MBCCE condition at each state becomes
the CCE condition of the corresponding stage game. In addition, if the MBCCE distribution
is concentrated on one stationary product-policy profile, and every state has positive probability, it reduces to MPE. MBCCE also retains further properties of CCE. First, Section~\ref{sec:properties} introduces a Markov smoothness condition and shows that the standard smoothness argument in \cite{Roughgarden2009} yields statewise welfare guarantees for MBCCE. Second, MBCCE preserves the linear structure of CCE: for any finite collection of stationary policy profiles, its equilibrium conditions can be formulated as a linear program, as discussed in Section~\ref{sec:properties}.

To connect decentralized learning to MBCCE, we introduce \emph{adaptive Markov
coarse regret} (AMCR). AMCR evaluates a sequence of stationary product-policy
profiles along the realized state path. At every visit to a state, it compares
the player's current mixed action with each fixed action available at that state.
For such a comparison, the player changes only her current action: the opponents
use their contemporaneous policies at the current state, and the contemporaneous
joint policy profile is used to evaluate continuation from the next state onward.
Across the realized visits to each state, AMCR then asks, for each player, which
single fixed action would have produced the largest cumulative gain in hindsight.
Importantly, AMCR accounts for the effect of the current action on future states while requiring only evaluation of the joint policy at that date. 
It therefore does not require specifying how the other learners would have behaved following a counterfactual deviation, making it compatible with adaptive algorithms that learn from data. 
For the continuing and episodic protocols studied here, we show that a vanishing AMCR gap implies that every accumulation point of the empirical distribution of realized states and contemporaneous policy profiles is an MBCCE. The state frequencies follow from the actor's decreasing updates in continuing play and from holding the profile fixed through each independently terminated episode in episodic play.

Our first result identifies a simple principle for obtaining AMCR: combine
no regret at each state with accurate evaluation of the current joint policy
profile. Joint-policy evaluation assigns each current action a score that
accounts for both its immediate payoff and its effect on future states, assuming
that the current joint policy is followed thereafter. Theorem~\ref{thm:AMCR-regret-tracking}
bounds the AMCR gap by the aggregate statewise external regret under these
scores and the average joint-policy-evaluation error. With exact policy
evaluation, AMCR becomes exactly statewise
external regret. Thus, obtaining
AMCR reduces to two familiar tasks: evaluating the current joint policy and
using no-regret learning among the actions available at each visited state.

Guided by this principle, we consider two natural decentralized reinforcement learning algorithms that yield AMCR. The first is an asynchronous actor--critic algorithm that learns continuously along a single trajectory. Each player observes the current state, the realized joint action, her own payoff, and the next state, but not the opponents' policies, payoff functions, or learning rules, and the transition kernel is unknown. We use a standard temporal-difference update to estimate action values from realized state--action transitions, while the player's own policy and the realized actions of the opponents are used to estimate the continuation value of the current joint policy. We use these values to provide statewise feedback to a no-regret actor and we allow each player--state pair to use any actor from a broad family that has sublinear external regret and changes its policy sufficiently slowly; this family includes Hedge, Euclidean dual averaging, projected gradient ascent, and their optimistic variants, and different players and states may use different actors. To generate the coverage needed for learning in such a setting, the implemented policies include uniform exploration that vanishes over time, and we impose only that the state chain generated by uniform joint actions is irreducible. The main analytical difficulty is that the policy profile being evaluated continues to change while the value updates are asynchronous and occur on endogenous local clocks, namely the visit counts generated by play that index the updates for each state or state--action pair. We develop a new two-timescale analysis showing that the value estimates track this moving joint-policy evaluation; Section~\ref{sec:algorithm} gives the technical overview. Combining this tracking result with statewise no regret, Theorem~\ref{thm:algorithm-AMCR} shows that the behavior-policy path generated by this algorithm achieves vanishing AMCR almost surely. Theorem~\ref{thm:finite-time-AMCR} further provides finite-time rates for such a setting. The expected AMCR gap decreases at a polynomial rate whose exponent depends on how many exploratory choices may be needed to reach a state and action. By the connection between AMCR and MBCCE, every accumulation point of the empirical distribution of states and behavior profiles is an MBCCE.

Our second algorithm is an episodic projected policy-gradient method, similar to the REINFORCE procedure studied by \cite{LeonardosEtAl2022} for potential Markov games. We show that the same type of decentralized procedure can obtain AMCR in any finite stochastic game. In this setting, the joint policy is chosen before the initial state is drawn and held fixed within each episode and each player observes only the states, her own actions, and her own realized payoffs; episode returns provide estimates of current-action gains, while projected policy-gradient updates provide statewise no regret. With  suitable step sizes, Theorem~\ref{thm:reinforce-AMCR} shows that the resulting behavior-policy path generated by this algorithm achieves exact AMCR asymptotically in every finite discounted stochastic game. The expected AMCR gap over all stages played decreases at least as fast as the inverse cube root of the number of completed episodes.  Consequently, every accumulation point of the empirical distribution of realized states and contemporaneous policies is an MBCCE. Therefore, these two algorithms, despite both relying on standard features of reinforcement learning, differ substantially and still provide complementary implementations of the same regret--evaluation principle.

We also prove a converse. For every MBCCE under either sampling rule, Theorem~\ref{thm:mbcce-converse} constructs algorithms whose empirical distribution converges almost surely to that equilibrium, with sublinear aggregate statewise external regret and vanishing average policy-evaluation error on the states that play can reach.  Thus, MBCCE is the smallest set that contains all empirical limits allowed by the two natural learning tasks of statewise regret minimization and current-policy evaluation.

We next connect MBCCE to stationary MCCE. Given an MBCCE in which every state has positive probability, we form a stationary
joint-action law by averaging, at each state, the action distributions prescribed
by its policy profiles. We show that the MCCE gap of this averaged law is controlled by how replacing the sampled continuation policy with the averaged law changes the payoff gain from a unilateral action. In particular, the averaged law is an approximate MCCE whenever this change in deviation payoffs is small. We use this connection to derive
MCCE guarantees for transition-aggregative stochastic games, in which the
next-state distribution depends Lipschitz continuously on a weighted aggregate
of players' actions, while stage payoffs may have unrestricted strategic
interactions. When players have equal weight and the bounds on action effects and transition sensitivity do not depend on the number of players, each player's influence on the
next-state distribution is bounded by a constant divided by the number of players. At a fixed discount factor, we show that the empirical joint-action law
generated by either algorithm is then an approximate MCCE with an error of the same order, provided the critic tracks the current policy in continuing play and every state can be drawn initially in episodic play. Dynamic routing and common-pool-resource games provide canonical
examples of such settings, and hence, our algorithms achieve approximate MCCE for such games when the number of players is large. We also show in this setting we can provide welfare guarantees for MBCCE.

The rest of the paper is organized as follows. Section~\ref{sec:model} introduces the stochastic game and the feedback structure. Section~\ref{sec:AMCR} defines AMCR and establishes the regret--joint-policy-evaluation principle. Section~\ref{sec:markov-bcce} introduces MBCCE and establishes its connection to empirical learning paths. Section~\ref{sec:algorithm} studies the continuing actor--critic algorithm, and Section~\ref{sec:reinforce-AMCR} studies episodic projected REINFORCE. Section~\ref{sec:properties} develops further properties of MBCCE, including its linearity of structure and welfare guarantees under smoothness. Section~\ref{sec:converse} proves the converse for continuing and episodic play. Section~\ref{sec:mcce} relates MBCCE to stationary MCCE and applies the connection to transition-aggregative stochastic games.

\subsection{Motivating example}
\label{sec:motivating-example}

Before introducing our equilibrium notions, we compare a normal-form game with a
stochastic game that has a similar best-response pattern. The example provides a stylized setting in which stationary MCCE does not
enlarge MPE, while MBCCE produces an enlargement analogous to that of CCE over
Nash equilibrium---the type of enlargement one expects from adaptive learning
in such settings. 

\begin{example}
\label{ex:stochastic-investment-MBCCE}
Consider first the common-interest anti-coordination game
\[
\begin{array}{c|cc}
 & 0 & 1\\ \hline
0 & (0,0) & (1,1)\\
1 & (1,1) & (0,0).
\end{array}
\]
The Nash marginal profiles are $(1,0)$,
$(0,1)$, and $(1/2,1/2)$. The possible action marginals under CCE form the
shaded quadrilateral in
Figure~\ref{fig:anti-coordination-investment-MBCCE1}(a). As usual in such games, the CCE set
generated by no-regret learning is larger than the Nash equilibrium set, since
such dynamics generally do not converge to an exact best response.

Now consider a two-player stochastic investment game with states
$\cS=\{0,1\}$ and actions $A_i=\{0,1\}$, where action $1$ means investment. The  following particular primitives
simplify the calculations, but the economic structure is a standard
public-investment problem with strategic substitutes.
\[
u_i(s,a)=s-\frac12a_i,
\qquad
P(1\mid0,a)=\frac{a_1+a_2}{3},
\qquad
P(1\mid1,a)=\frac34,
\qquad
\gamma=\frac45.
\]
The complementary transition probability leads to state $0$. State $1$ yields
positive payoff but investment is costly. Aggregate investment at state $0$
raises the probability of reaching state $1$. For MBCCE in this example, we use episodic sampling with an initial distribution that assigns positive probability to both states; Section~\ref{sec:markov-bcce} gives the formal definition.
At state $1$, investment lowers the current payoff by $1/2$ without affecting
transitions, so action $0$ is strictly optimal. Under all three equilibrium
notions considered here (MPE, MCCE, and MBCCE whenever state $1$ occurs with positive probability), players therefore choose action
$0$ at state $1$. Given this behavior, a direct calculation provided in the
appendix shows that at state $0$ a player strictly prefers to invest exactly when the
other player invests with probability below $1/2$. Intuitively, the other
player's investment already raises the probability of reaching state $1$,
reducing the benefit from investing oneself. Thus, the best-response
correspondence at state $0$ is exactly the same as in the normal-form
anti-coordination game. The MPE investment probabilities at state $0$ are
consequently the same three points shown in
Figure~\ref{fig:anti-coordination-investment-MBCCE1}(b).

Stationary MCCE does not enlarge this set. In this example, under any stationary
correlated action law, the value
of every stationary deviation depends only on the players' statewise investment
probabilities. Hence, replacing the correlated law by independent Markov
policies with the same marginals does not change the equilibrium values, so a
stationary joint-action law is an MCCE exactly when the corresponding product
policy is an MPE.\footnote{MCCE can still correlate the players' actions. In
particular, at the symmetric marginal profile $(1/2,1/2)$, MCCE allows
correlated joint-action laws with these marginals, whereas MPE uses independent
mixing, but this additional correlation does not affect payoffs, transitions,
or deviation values.} The MCCE investment probabilities at state $0$ are
therefore again $(1,0)$, $(0,1)$, and $(1/2,1/2)$.

MBCCE produces the analogue of the CCE enlargement discussed above. As shown
in Figure~\ref{fig:anti-coordination-investment-MBCCE1}(b), MBCCE permits a
two-dimensional set of investment probabilities at state $0$, whereas MPE and
MCCE are restricted to the three exact stationary best-response profiles
described above.\footnote{To obtain the shaded region, we allow the distribution over state-$0$ continuation policies to include mixed profiles. For example, the state-$0$ posterior that assigns equal probability to the profiles $(2/3,0)$ and $(0,2/3)$ satisfies the MBCCE inequalities and gives the marginal pair $(1/3,1/3)$. Action $0$ at state $1$ is strictly optimal and, under the episodic sampling convention used here, is chosen under every MBCCE in this game. The formal construction and a characterization of the full shaded region are given in the Appendix.}

In multi-agent reinforcement learning, one would not generally expect the policy path to converge to one of these stationary best-response profiles, much as no-regret dynamics in normal-form games need not converge to a Nash equilibrium. In this sense, MBCCE extends CCE from normal-form games as the distributional equilibrium generated by the empirical state--policy distributions of reinforcement learning algorithms, at least as we show for the actor--critic and policy-gradient algorithms studied in this paper.
\end{example}

\begin{figure}[t]
\centering
\includegraphics[width=0.98\textwidth]
{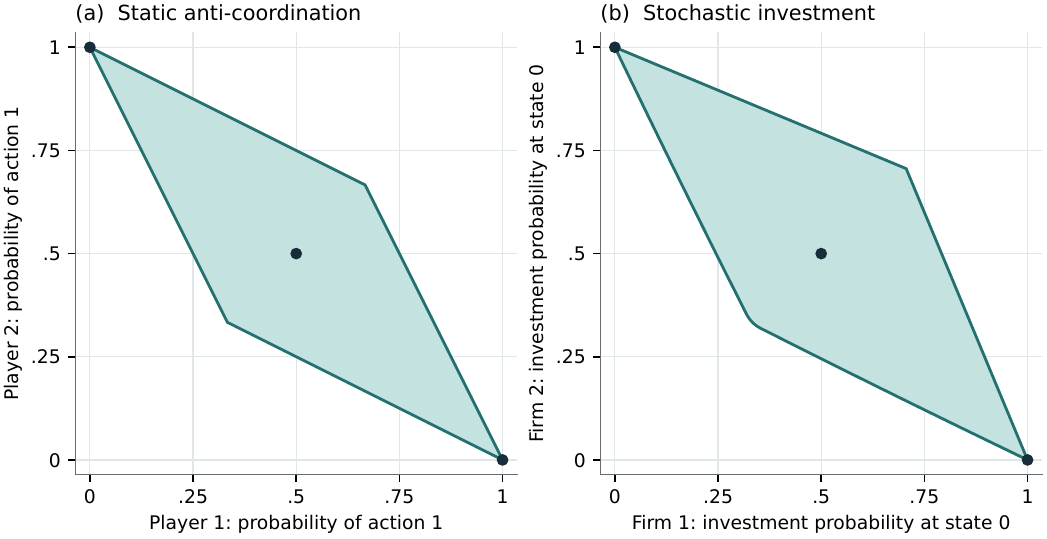}
\caption{Each axis gives a player's probability of action $1$; in panel~(b),
action $1$ is investment at state $0$. The points show Nash equilibrium
marginals in panel~(a) and the common MPE--MCCE marginals in panel~(b). The
shaded regions show CCE marginals in panel~(a) and MBCCE marginals in
panel~(b). The MBCCE region uses episodic sampling (see Section \ref{sec:markov-bcce}).}
\label{fig:anti-coordination-investment-MBCCE1}
\end{figure}

\subsection{Related Literature}

Our paper is related to two strands of literature.

\textbf{Learning in stochastic games.}
A growing literature develops algorithms that produce equilibrium policies in general-sum Markov games. \citet{JinLiuWangYu2024} introduce V-learning for episodic Markov games, obtaining correlated or coarse correlated policies in multiplayer general-sum games without dependence on the size of the joint action space. In finite-horizon models, \citet{CaiLuoWeiZheng2024}, \citet{MaoEtAl2024} develop policy-optimization procedures built from statewise no-regret updates. \citet{DaskalakisGolowichZhang2023} establish the computational hardness of stationary MCCE and construct an efficiently learnable nonstationary Markov CCE. Recently, \citet{FarinaKontogiannisPanageasPollatos2026} extend the hardness result to single-controller stochastic games. A broader literature studies learning under additional stochastic-game structure. In two-player zero-sum stochastic games, \citet{DaskalakisFosterGolowich2020} study independent policy-gradient methods, while \citet{XieChenWangYang2023} and \citet{YanLiChenFan2024} study learning with function approximation or offline data. Global policy-gradient guarantees are also available in Markov potential games \citep{LeonardosEtAl2022,DingWeiZhangJovanovic2022}. Mean-field equilibrium and learning, as well as networked multiagent reinforcement learning, are studied in  \citep{AdlakhaJohari2013,QuWiermanLi2022,GuoHuXuZhang2023,GuGuoWeiXu2025,light2025computing}. In general stochastic games, \citet{GiannouLotidisMertikopoulosVlatakisGkaragkounis2022} establish local convergence around second-order stationary Nash policies, while \citet{MaheshwariWuSastry2024} characterize the asymptotic behavior of decentralized actor--critic dynamics using Markov near-potential functions. These contributions focus on computing an equilibrium policy or a prescribed policy sequence, or on convergence of learning dynamics, often under additional game structure. We instead characterize the equilibrium content of the adaptive path produced by decentralized tabular learning, and our AMCR--MBCCE framework is developed for general finite discounted stochastic games.

A related literature studies regret against policy deviations in Markov games.
\citet{FosterGolowichKakade2023} establish computational and statistical
hardness results for independent no-regret learning in finite-horizon episodic
Markov games. \citet{LiuWangJin2022} and \citet{ZhanLeeYang2023} obtain no-regret
guarantees against fixed policy comparators in episodic policy-revealing
settings, where an opponent's past policies are observed;
\citet{ZhanLeeYang2023} also allows general function approximation. \cite{WangLiuBaiJin2023} develop decentralized
policy-space mirror descent and all-policy evaluation in episodic Markov
games under function approximation.
\citet{ErezEtAl2023} study regret against Markov-policy deviations in
finite-horizon common-algorithm self-play and obtain sublinear swap regret
under an additional endogenous condition on the resulting policy iterates.
\citet{NguyenTangArora2024} study one learner against an adaptive opponent,
obtaining positive results under bounded-memory, stationarity, and consistency
conditions together with hardness results outside such structure.
These results use episodic protocols or rely on substantially stronger information or structural
assumptions. They do not cover the continuing, single-trajectory partial-feedback
setting of our actor--critic result.

\textbf{Asynchronous two-timescale stochastic approximation.}
This literature is related specifically to our continuing actor--critic algorithm. The differential-inclusion framework developed in \citet{BenaimHofbauerSorin2005} and its asynchronous extension in \citet{PerkinsLeslie2013}, which also treats coupled two-timescale recursions, provide general asymptotic tools for stochastic-approximation dynamics. More recently, \citet{Borkar2025GeneralCase} allows either timescale to remain nonconvergent and characterizes the possible asymptotic behavior. These frameworks characterize asymptotic behavior and do not provide the quantitative moving-target policy-evaluation bound used in our finite-time AMCR result. In addition, asynchronous analyses typically assume sufficient state--action visitation, whereas in our setting we derive this coverage from a primitive communication condition on the transition dynamics together with the exploration built into the algorithm. \citet{HakamiDehghan2016} analyze a specific three-timescale no-regret $Q$-learning scheme for constrained general-sum stochastic games and establish convergence to the set of stationary correlated equilibria. In contrast, our critic tracks any no-regret actor satisfying mild regret and movement guarantees, even when the actor path does not converge.

Recent finite-time analyses study more structured two-timescale fixed-point recursions. \citet{ChandakHaqueBambos2025} assume that both the fast map and the reduced slow map are contractive. \citet{Chandak2026NonExpansive} permits a nonexpansive reduced slow map and obtains a residual-rate exponent arbitrarily close to one quarter. In these papers, the fast recursion tracks the fixed point associated with a prescribed slow fixed-point iteration whose reduced map is contractive or nonexpansive, and the analysis targets convergence or residual of the slow recursion. Our objective is different: without requiring the slow policy update to be a fixed-point iteration or its reduced map to be contractive or nonexpansive, we establish quantitative policy-evaluation tracking along a possibly nonconvergent policy path and use this tracking to obtain finite-time AMCR guarantees. 

\section{Model}
\label{sec:model}

\subsection{Stochastic game}

We consider an $N$-player discounted stochastic game with finite state space
$\cS$, finite action sets $A_i$, and joint action space
$\cA:=\prod_{i=1}^N A_i$. Write $a=(a_1,\ldots,a_N)$ and
$A_{-i}:=\prod_{j\ne i}A_j$. At stage $t$, the state is $s_t$, each player
$i$ chooses an action $a_{i,t}$, and the resulting joint action is
$a_t=(a_{1,t},\ldots,a_{N,t})$. Player $i$ then receives
$u_i(s_t,a_t)\in[0,1]$, and the next state is drawn from
$P(\cdot\mid s_t,a_t)$. The discount factor is $\gamma\in(0,1)$, and
$H:=(1-\gamma)^{-1}$.

 A stationary Markov policy of player $i$ is a map
$\pi_i:\cS\to\Delta(A_i)$, where $\pi_i(a_i\mid s)$ is the probability of
action $a_i$ at state $s$. A stationary product-policy profile is denoted by
$\Pi=(\pi_i)_{i=1}^N$, with joint action law
$\Pi(\cdot\mid s):=\bigotimes_{i=1}^N\pi_i(\cdot\mid s)$. We write
$\Pi_{-i}(\cdot\mid s):=\bigotimes_{j\ne i}\pi_j(\cdot\mid s)$ for the
marginal of all players other than $i$.  

\subsection{Information and feedback}

At date $t$, the learning rules maintain a target profile $\Pi_t\in\mathcal C$
and select a behavior profile $p_t\in\mathcal C$ where 
\[
\mathcal C:=\prod_{s\in\cS}\prod_{i=1}^N\Delta(A_i)
\]
is the compact set of stationary product-policy profiles. The target profile is the
policy being updated and evaluated, whereas the behavior profile generates the
actions used to collect data. The behavior profile introduces the additional
randomization required for learning.

All random variables are defined on a probability space
$(\Omega,\mathcal F,\bbP)$. Let $\mathcal H_t$ be the $\sigma$-field
generated by all data available in the analysis immediately before the
date-$t$ action.\footnote{Formally, let $\mathcal H_1$ be the $\sigma$-field
generated by the initial state and the initial variables of the learning
rules. Recursively, $\Pi_t$, $p_t$, and all variables maintained before the
date-$t$ action are $\mathcal H_t$-measurable. Conditional on
$\mathcal H_t$ and the realized action, the state moves according to the
transition kernel:
$
\bbP(s_{t+1}=s'\mid\mathcal H_t\vee\sigma(a_t))
=P(s'\mid s_t,a_t)
$.
The next history includes the current history, actions and rewards, the new state, and any internal randomization used to form the next profiles. In Section~\ref{sec:reinforce-AMCR} the transition law applies within an episode. At a reset, the new episode's profile is chosen before its initial state is drawn.}
Conditional on
$\mathcal H_t$, the players draw their actions independently according to the
behavior profile:
\begin{equation}
\label{eq:behavior-sampling}
\bbP(a_t=a\mid\mathcal H_t)
=p_t(a\mid s_t)
:=\prod_{i=1}^N p_{i,t}(a_i\mid s_t),
\qquad a\in\cA.
\end{equation}

We assume that player $i$ observes the current state, her own action and payoff, and the next state; she does not observe
the other players' policies, payoff functions, or learning rules, and the
transition kernel is unknown. In particular, players do not have all the information contained in $\mathcal H_t$. 
 The actor--critic
algorithm developed in Section~\ref{sec:algorithm} also assumes that player $i$ observes  the
realized joint action. In both learning algorithms, player $i$ does not observe the policies or counterfactual actions of the other players. The learners constructed for Theorem~\ref{thm:mbcce-converse} can agree in advance on prescribed profiles, states to reach, and joint actions to use when collecting observations. Each player still observes only the public states, her own actions, and her own payoffs during play.

\subsection{Notation}

For finite-dimensional vectors $x$ and $y$, write
$\langle x,y\rangle:=\sum_jx_jy_j$,
$\|x\|_1:=\sum_j|x_j|$, and $\|x\|_\infty:=\max_j|x_j|$.  The notation
$\one\{E\}$ denotes the indicator of an event $E$. For a scalar $x$, let
$(x)_+:=\max\{x,0\}$. For probability vectors $\alpha$ and $\beta$, let
$\TV(\alpha,\beta):=\tfrac12\|\alpha-\beta\|_1$. For stationary product
profiles, use
\[
\|\Pi-\Pi'\|_1
:=\sum_{s\in\cS}\sum_{i=1}^N
\|\pi_i(\cdot\mid s)-\pi_i'(\cdot\mid s)\|_1.
\]

For a horizon $T$, define
$ 
T_s(T):=\{t\le T:s_t=s\}
$, and 
$N_s(T):=|T_s(T)|$. 
When no ambiguity arises, write $T_s$ and $N_s$. Let
$t_1(s)<t_2(s)<\cdots$ denote the successive visits to state $s$ (when they occur), and define
the local state and state--action clocks
$
n_t(s):=\sum_{u=1}^t\one\{s_u=s\}$, and $
n_t(s,a):=\sum_{u=1}^t\one\{s_u=s,a_u=a\}$. 
Thus $n_{t_q(s)}(s)=q$ and $\sum_sN_s(T)=T$. These clocks index the
asynchronous updates below: a state block changes only when that state is
visited, and a state--action coordinate changes only when the corresponding
joint action is observed.

\section{Adaptive Markov coarse regret}
\label{sec:AMCR}

In this section we 
introduce Adaptive Markov coarse regret (AMCR). For each action at a state, AMCR evaluates the gain from using that action at every
realized visit, with the profile of all players at that date used for continuation. For each player--state pair, AMCR selects the best fixed action in
hindsight and requires its long-run time-average gain to vanish. 
Intuitively, AMCR requires that no fixed action
yield a persistent average improvement relative to the contemporaneous
profiles. This notion avoids specifying the unobserved future learning
path under a counterfactual deviation which is a key challenge in a general stochastic game setting as we discussed in the introduction.

To define AMCR we first define the value
of a current action under a stationary continuation law. Given a stationary
statewise joint-action law
$\rho=(\rho^s)_{s\in\cS}\in\prod_s\Delta(\cA)$, we call $(Q^\rho,V^\rho)$ the canonical Bellman pair if it satisfies
\begin{equation}
\label{eq:canonical-bellman-system}
\begin{aligned}
Q_i^\rho(s,a)
&=u_i(s,a)+\gamma\sum_{s'}P(s'\mid s,a)V_i^\rho(s'),\\
V_i^\rho(s)
&=\sum_a\rho^s(a)Q_i^\rho(s,a).
\end{aligned}
\end{equation}
The following Lemma follows from standard dynamic programming arguments. All proofs in this paper are in the appendix. 

\begin{lemma}[Canonical Bellman pair]
\label{lem:Bellman-tying-unique}
For every $\rho\in\prod_s\Delta(\cA)$,
system~\eqref{eq:canonical-bellman-system} has a unique solution, and
$Q^\rho,V^\rho\in[0,H]$ componentwise.
\end{lemma}

For a product profile $\Pi$, we denote by $(Q^\Pi,V^\Pi)$  the canonical pair
generated by $\Pi(\cdot\mid s)$ at every state. Intuitively, at date $t$, this pair treats
$\Pi_t$ as a stationary continuation policy for evaluating the current
decision. With this notation, for a product profile $\Pi$,  we define the deviation advantage
\begin{equation}
\label{eq:true_adv}
\mathcal A_i(\Pi;s,a_i')
:=\sum_{a_{-i}}\Pi_{-i}(a_{-i}\mid s)
Q_i^\Pi(s,(a_i',a_{-i}))-V_i^\Pi(s).
\end{equation}
as the expected gain
from choosing $a_i'$ at the current visit to $s$ compared to drawing an action
from $\pi_i(\cdot\mid s)$, when the other players draw from
$\Pi_{-i}(\cdot\mid s)$ and all players follow $\Pi$ from the next state
onward. 

\begin{definition}[Adaptive Markov coarse regret]
For a path $\Gamma_{1:T}=(\Gamma_t)_{t=1}^T$ of stationary product profiles,
define\footnote{Throughout the paper, all AMCR gaps are evaluated along the realized state path
$(s_t)_{t\ge 1}$ generated by the behavior policies. To make this dependence
explicit, one may write
\[
\mathcal G_T^{\mathrm{AMCR}}
(\Gamma_{1:T};s_{1:T})
:=
\frac1T\sum_{s\in\mathcal S}\sum_{i=1}^N
\left(
\max_{a_i'\in A_i}
\sum_{\substack{t\le T\\ s_t=s}}
\mathcal A_i(\Gamma_t;s,a_i')
\right)_+.
\]
We suppress $s_{1:T}$ from the notation. 
}
\begin{equation}
\label{eq:maf_cce_gap}
\mathcal G_T^{\mathrm{AMCR}}(\Gamma_{1:T})
:=\frac1T\sum_{s\in\cS}\sum_{i=1}^N
\left(
\max_{a_i'\in A_i}\sum_{t\in T_s(T)}
\mathcal A_i(\Gamma_t;s,a_i')
\right)_+.
\end{equation}
We say that the path is a $\delta$-AMCR at horizon $T$ when
$
\mathcal G_T^{\mathrm{AMCR}}(\Gamma_{1:T})\le\delta$, 
and it is asymptotically $\delta$-AMCR when
\begin{equation*}
\limsup_{T\to\infty}
\mathcal G_T^{\mathrm{AMCR}}(\Gamma_{1:T})\le\delta.
\end{equation*}
For the behavior-policy path actually played we write
$\mathcal G_T^{\mathrm{AMCR}}
:=\mathcal G_T^{\mathrm{AMCR}}(p_{1:T})$.
\end{definition}

The maximization in \eqref{eq:maf_cce_gap} selects the best fixed action in
hindsight for each player--state pair and evaluates it across all realized
visits to that state. A key feature is the time average: the
profile and its Bellman values may change from one visit to the next, but a
small gap means that no fixed statewise action would have produced a systematic
continuation-adjusted improvement along the realized path. The division by $T$ weights each state by its observed frequency; a state with a vanishing share of dates therefore contributes vanishingly little to AMCR even if its conditional mean deviation gain stays positive. The deviation is
coarse because it does not condition on a recommendation, and adaptive because
each date is evaluated under its own current profile. As explained above, AMCR
requires neither a forecast of future learning nor a model of how the other
players would react to an alternative history.

\begin{remark}[AMCR and MPE]
\label{rem:AMCR-rationality}
The one-deviation principle characterizes a Markov perfect equilibrium: at
every state, no player can gain by changing only her current action and then
returning to the equilibrium profile.
Formally, we have the following definition: 
\begin{definition}[Stationary Markov perfect equilibrium]
\label{def:MPE}
For a stationary product profile $\Pi$, define
\[
\GapMPE(\Pi)
:=
(1-\gamma)
\max_{\substack{i,\,s\in\cS\\ \beta_i:\cS\to A_i}}
\left(
V_i^{\beta_i,\Pi_{-i}}(s)-V_i^\Pi(s)
\right)_+.
\]
The profile $\Pi$ is a $\delta$-MPE if $\GapMPE(\Pi)\le\delta$.
\end{definition}

AMCR on the other hand aggregates deviation advantages along the realized
visits to each state and requires that no fixed action yield a positive
long-run average gain. If the sequence of policies converges, the continuation values stabilize, and under the state-space coverage condition below, AMCR reduces to MPE.
\begin{proposition}
\label{prop:AMCR-mpe}
Suppose $\Gamma_t\to\Pi^\star$,
$\mathcal G_T^{\mathrm{AMCR}}(\Gamma_{1:T})\to0$, and
$\liminf_{T\to\infty}N_s(T)/T>0$ for every $s\in\cS$.
Then $\GapMPE(\Pi^\star)=0$.
\end{proposition}
\end{remark}

We next show in Theorem~\ref{thm:AMCR-regret-tracking} that two properties
naturally desired of a learning algorithm imply AMCR: the actor should have
small statewise regret relative to the scores it uses, and those scores should
evaluate the true current joint profile accurately. We first define the
statewise regret induced by a given sequence of scores.
Suppose an algorithm maintains value tables $(Q_t,V_t)$. At a visit $t$ to
state $s_t$, define the target-law gain vector
\begin{equation}
\label{eq:target-law-gain}
g_{i,t}^{\mathrm{tar}}(a_i)
:=\sum_{a_{-i}}\Pi_{-i,t}(a_{-i}\mid s_t)
Q_{i,t}(s_t,(a_i,a_{-i}))
\end{equation}
and the regret under the target policy over the first $n$ visits to state $s$,
\begin{equation}
\label{eq:target-law-regret}
R_{i,s}^{\mathrm{tar}}(n)
:=\left(
\max_{a_i'\in A_i}\sum_{q=1}^n
\left[
 g_{i,t_q(s)}^{\mathrm{tar}}(a_i')
 -\left\langle
 \pi_{i,t_q(s)}(\cdot\mid s),
 g_{i,t_q(s)}^{\mathrm{tar}}
 \right\rangle
\right]
\right)_+.
\end{equation}
Set $R_{i,s}^{\mathrm{tar}}(0):=0$. This is the standard external-regret
criterion applied statewise: it compares the cumulative score obtained by the
actor at visits to state $s$ with the cumulative score of the best fixed action
in hindsight.

\begin{theorem}[Regret--joint-policy-evaluation principle]
\label{thm:AMCR-regret-tracking}
Let $e_t=\|(Q_t,V_t)-(Q^{\Pi_t},V^{\Pi_t})\|_\infty$. For every sample path and every $T$,
\begin{equation}
\label{eq:AMCR-regret-tracking}
\left|\mathcal G_T^{\mathrm{AMCR}}(\Pi_{1:T})-
 \frac1T\sum_{s,i}R_{i,s}^{\mathrm{tar}}(N_s(T))\right|
 \le\frac{2N}{T}\sum_{t=1}^T e_t.
\end{equation}
\end{theorem}

Theorem~\ref{thm:AMCR-regret-tracking} separates statewise regret from joint-policy-evaluation error. The regret term in \eqref{eq:AMCR-regret-tracking} measures how well the actor responds to its estimated statewise scores. The right-hand side measures the error in estimating the canonical Bellman pair of the current joint profile. Because the profile $\Pi_t$ may
change with $t$, $e_t$ is an online joint-policy-evaluation error: at each date,
the maintained tables by the learning algorithm are compared with the Bellman values of the
contemporaneous profile.

Consequently, if the average joint-policy-evaluation error vanishes, then
$
\mathcal G_T^{\mathrm{AMCR}}(\Pi_{1:T})\to 0
$ if and only if we have sub-linear regret, i.e., $
\sum_{s\in\cS}\sum_{i=1}^N
R_{i,s}^{\mathrm{tar}}(N_s(T))=o(T).
$ 
In particular, with exact policy evaluation, namely $e_t=0$ for every $t$, the AMCR gap equals
 aggregate statewise regret divided by $T$ at every horizon. This leads to a natural question of how 
exactly AMCR extends the no-regret interpretation of equilibrium in
normal-form games to stochastic games. We discuss this connection in detail in the next section. 

\section{Markov Bayes coarse correlated equilibrium}
\label{sec:markov-bcce}

We next relate AMCR to a static equilibrium notion that plays for Markov
perfect equilibrium a role analogous to that of coarse correlated equilibrium
for Nash equilibrium. The relevant static equilibrium notion preserves 
the timing of a stochastic game: the current state is publicly observed before
actions are chosen. Thus, if one considers a static game preceded by a public
state, the natural coarse equilibrium notion is a Bayes coarse correlated
equilibrium (BCCE), in which deviations may depend on the observed state but
must be chosen before observing the assigned action; see
\cite{HartlineSyrgkanisTardos2015}.

First draw a stationary product-policy profile $\Pi$ from a lottery $\nu$ on $\mathcal C$. Then draw the public state according to its visit frequencies under $\Pi$. Write $d_\Pi(s)$ for the probability of state $s$ at that decision when play follows $\Pi$. These probabilities come from how the decision date is sampled, before the current actions are drawn. The two draws give the joint law
\begin{equation}
 \mu(\{s\},d\Pi)=d_\Pi(s)\,\nu(d\Pi),\qquad s\in\cS.
 \label{eq:mbcce-state-law}
\end{equation}
Here $d_\Pi$ is a probability vector measurable in $\Pi$. Equation~\eqref{eq:mbcce-state-law} describes a decision at a single date without specifying how any player learns. We consider the standard continuing and episodic settings studied in reinforcement learning. Write $P_\Pi(s'\mid s)=\sum_a\Pi(a\mid s)P(s'\mid s,a)$ for the transition probabilities under $\Pi$.

In continuing play, consider a date sampled according to the long-run frequencies of play under the selected profile. The probability assigned to each state must equal the total probability of a transition into that state. Thus $d_\Pi$ satisfies
\begin{equation}
 d_\Pi=d_\Pi P_\Pi.
 \label{eq:mbcce-continuing-states}
\end{equation}
If $P_\Pi$ has a unique invariant distribution, it determines $d_\Pi$. If it has several closed classes, the frequencies also depend on the probabilities with which these classes are sampled. We allow any mixture of these classes and use the long-run frequencies within each sampled class. In the learning result, the joint frequencies observed along play determine the class weights in each empirical limit. Section~\ref{sec:algorithm} studies this setting for the actor--critic, which observes one continuing trajectory.

In episodic play, fix a distribution $\rho_0$ for the initial state. Hold the profile fixed through an episode, which continues after each stage with probability $\gamma$ independently of play. A date sampled from all played stages gives greater weight to longer episodes. Consequently the probability of a state under a given profile is its expected number of visits to that state divided by the expected episode length. The probability of reaching stage $k$, counting the first stage as zero, is $\gamma^k$. Hence
\begin{equation}
 d_\Pi=(1-\gamma)\sum_{k=0}^\infty\gamma^k\rho_0 P_\Pi^k
      =(1-\gamma)\rho_0(I-\gamma P_\Pi)^{-1}.
 \label{eq:mbcce-episodic-states}
\end{equation}
Section~\ref{sec:reinforce-AMCR} studies this setting for the episodic policy-gradient algorithm.

Let $\mu$ be the law in \eqref{eq:mbcce-state-law}. Ex ante,
its marginal on $\mathcal C$ is the common prior over stationary
product-policy profiles. Nature draws $(s,\Pi)$ according to $\mu$. Players
observe the public state $s$ but not $\Pi$, so whenever
$\mu(\{s\}\times\mathcal C)>0$ their common posterior over the continuation
policy is $\mu(d\Pi\mid s)$.
For the specified state frequencies, this posterior is
\[
 \mu(d\Pi\mid s)
   =\frac{d_\Pi(s)\,\nu(d\Pi)}{\int_{\mathcal C}d_{\Pi'}(s)\,\nu(d\Pi')}.
\]
After observing the state, a mediator who also observes $\Pi$ draws
$a\sim\Pi(\cdot\mid s)$. Before observing this assigned action, player $i$
may follow the mediator's recommendation or replace it by a fixed action $b_i\in A_i$.
Continuation play from the next state onward remains governed by $\Pi$.
Thus $Q_i^\Pi(s,a)$ is the payoff from following the  recommendation and
$Q_i^\Pi(s,(b_i,a_{-i}))$ is the payoff from the deviation. 

\begin{definition}[Markov Bayes coarse correlated equilibrium]
\label{def:markov-bcce}
For the specified way of sampling a decision date, a law $\mu$ of the form \eqref{eq:mbcce-state-law} is a \emph{Markov Bayes coarse correlated equilibrium} (MBCCE) if, for every
player $i$, every state $s$ with
$\mu(\{s\}\times\mathcal C)>0$, and every $b_i\in A_i$,
\begin{equation}
\label{eq:markov-bcce-obedience}
\E_{\substack{\Pi\sim\mu(\cdot\mid s)\\
a\sim\Pi(\cdot\mid s)}}
\left[Q_i^\Pi(s,a)\right]
\ge
\E_{\substack{\Pi\sim\mu(\cdot\mid s)\\
a\sim\Pi(\cdot\mid s)}}
\left[Q_i^\Pi(s,(b_i,a_{-i}))\right].
\end{equation}
\end{definition}

Equivalently, MBCCE requires that $\int_{\mathcal C}d_\Pi(s)\mathcal A_i(\Pi;s,b_i)\,\nu(d\Pi)\le0$ for every $i,s,b_i$. The factor $d_\Pi(s)$ weights the gain from a deviation by how often the sampled profile gives the player that decision.

The next result gives a stochastic game analogue of the familiar connection between
no regret and coarse correlated equilibrium in normal form games. Under the continuing and episodic protocols studied here, vanishing AMCR implies that every accumulation point of the states-policies empirical distribution is an MBCCE. The state frequencies follow from these protocols, as we now explain.

In continuing play, the decreasing updates of standard no-regret actors also give the state frequencies required by MBCCE. With bounded scores, Hedge, dual averaging, and projected gradient ascent change each policy row by an amount that vanishes with their step size. For the statewise updates studied in Section~\ref{sec:algorithm}, the policies actually played satisfy\footnote{On the $n$th visit to a state, the target rows change by at most a constant times $n^{-c}$, where $0<c<1$. Summing over the finitely many state clocks gives total movement of order $T^{1-c}$, so average movement is of order $T^{-c}$. Nonincreasing exploration adds only bounded total movement.}
\begin{equation}
 \frac1T\sum_{t=1}^{T-1}\|p_{t+1}-p_t\|_1\longrightarrow0.
 \label{eq:mbcce-average-increments}
\end{equation}
The policies may continue to change. Vanishing average movement ensures that assigning arrivals to their departure profiles or to their arrival profiles gives the same limiting state frequencies, as the proof shows.

In episodic play, the profile is chosen before the independent reset and held fixed through the episode. The expected number of visits under that profile gives exactly \eqref{eq:mbcce-episodic-states} after division by the mean episode length. Thus both protocols give the state frequencies required by the second draw. Under either protocol, obedience is the remaining requirement for equilibrium.

\begin{proposition}[AMCR and empirical MBCCE]
\label{prop:empirical-markov-bcce}
For a path $\Gamma_{1:T}$, let
$\widehat\mu_T:=T^{-1}\sum_{t=1}^T\delta_{(s_t,\Gamma_t)}$.
Then
\begin{align}
&\sum_{\substack{s\in\cS\\N_s(T)>0}}
\frac{N_s(T)}{T}
\sum_{i=1}^N
\left(
\max_{b_i\in A_i}
\E_{\substack{\Pi\sim\widehat\mu_T(\cdot\mid s)\\
a\sim\Pi(\cdot\mid s)}}
\left[
Q_i^\Pi(s,(b_i,a_{-i}))-Q_i^\Pi(s,a)
\right]
\right)_+
=
\mathcal G_T^{\mathrm{AMCR}}(\Gamma_{1:T}).
\label{eq:AMCR-empirical-mbcce}
\end{align}
For the behavior profiles $\Gamma_t=p_t$, suppose either that play is continuing and \eqref{eq:mbcce-average-increments} holds almost surely, or that play follows the geometric episode protocol above. Almost surely, AMCR vanishes if and only if every weak accumulation point of $\widehat\mu_T$ is an MBCCE under the corresponding sampling convention.
\end{proposition}

Under either protocol, vanishing AMCR thus gives convergence to the MBCCE set in the weak topology, even if the empirical laws have several limits.\footnote{Statewise regret and evaluation by themselves control only incentives. To see why the updating rule matters, let each next state be an independent fair draw from two states and give every player zero payoffs. A learner could select a different profile at each observed state. Its regret and exact evaluation error would both be zero. Each frozen profile would give equal frequencies to the two states, but the selected profile would identify the state on every date.}

MBCCE also relates to MPE in a similar way to the relation between CCE and Nash
equilibrium. If there is no uncertainty over the policy profile, so that the
policy marginal of $\mu$ is concentrated on some $\Pi$, and every state has
positive probability, then the posterior after every state is concentrated
on the same $\Pi$. Condition~\eqref{eq:markov-bcce-obedience} then reduces to
$\sum_{a_{-i}}\Pi_{-i}(a_{-i}\mid s)
Q_i^\Pi(s,(b_i,a_{-i}))\le V_i^\Pi(s)$ for every $i,s,b_i$.
These are exactly the statewise one-deviation conditions for a stationary
MPE (see Remark \ref{rem:AMCR-rationality}). Thus,  an  MBCCE concentrated on a single
stationary profile is an MPE in the stochastic game, paralleling the familiar
normal-form result that a CCE whose action distribution is a product
distribution is a Nash equilibrium.

More generally, we can think of MBCCE as another way to generalize CCE to stochastic games, as explained in the following remark.
\begin{remark}[Extensions of CCE to stochastic games]
\label{rem:two-CCE-extensions}
There are two equivalent representations of a CCE in a finite normal-form
game. First, a CCE is a joint-action distribution
$\lambda\in\Delta(\cA)$ satisfying the usual coarse-obedience inequalities.
Alternatively, let $\nu\in\Delta\left(\prod_i\Delta(A_i)\right)$
be a distribution over mixed-strategy profiles $x=(x_i)_{i=1}^N$,
first draw $x\sim\nu$, and then, conditional on $x$, draw each player's
action $a_i\sim x_i$ independently. The induced joint-action distribution is
$
\lambda_\nu(a)
:=
\E_{x\sim\nu}\left[\prod_{i=1}^N x_i(a_i)\right].
$
Then, for every player $i$ and fixed deviation $b_i\in A_i$,
\[
\E_{\substack{x\sim\nu\\a\sim\bigotimes_i x_i}}
\left[
u_i(a)-u_i(b_i,a_{-i})
\right]
=
\E_{a\sim\lambda_\nu}
\left[
u_i(a)-u_i(b_i,a_{-i})
\right].
\]
Hence, coarse obedience under $\nu$ is equivalent to the CCE inequalities
for $\lambda_\nu$. Conversely, every joint-action distribution $\lambda$
admits such a representation by assigning probability $\lambda(a)$ to the
pure mixed-strategy profile $(e_{a_i})_{i=1}^N$ for each $a\in\cA$.

These two equivalent representations lead to distinct extensions in a
stochastic game. Stationary MCCE extends the joint-action representation:
it specifies a stationary statewise joint-action law
$\sigma=(\sigma^s)_{s\in\cS}$ and evaluates continuation under $\sigma$.
MBCCE instead extends the distribution-over-mixed-strategy-profiles
representation: first draw $(s,\Pi)\sim\mu$, where $\mu$ is a
joint distribution over states and stationary product-policy profiles.
Conditional on $(s,\Pi)$, draw the players' current actions
independently according to $\Pi(\cdot\mid s)$. 
In a normal-form game, passing from $\nu$ to its induced joint-action
distribution $\lambda_\nu$ does not change any coarse-deviation payoff.
In a stochastic game, however, the analogous statewise averaging generally
changes the continuation values used to evaluate deviations. Thus, the two
representations of CCE that are equivalent in a normal-form game give rise
to distinct solution concepts in a stochastic game: stationary MCCE and
MBCCE. 
\end{remark}

\section{Actor--critic algorithm and AMCR guarantees}
\label{sec:algorithm}

Theorem~\ref{thm:AMCR-regret-tracking} reduces the target-path AMCR gap to
two terms: statewise regret under the target policy and online joint-policy-evaluation
error. We now develop an online actor--critic algorithm that controls both
terms with vanishing exploration. The actor's decreasing updates also give the state frequencies required by Proposition~\ref{prop:empirical-markov-bcce}, as verified below.

For each state $s$, let $\mathsf U_i(\cdot\mid s)$ be the uniform distribution
on $A_i$, so $\mathsf U_i(a_i\mid s)=1/|A_i|$, and let $\mathsf U$ denote
the corresponding product profile.

Choose a nonincreasing exploration schedule $\eps_t\in(0,1)$ tending to zero and define the behavior profile by
\begin{equation}
\label{eq:fixed-mixing}
p_{i,t}(\cdot\mid s)=(1-\eps_t)\pi_{i,t}(\cdot\mid s)+\eps_t\mathsf U_i(\cdot\mid s),
\qquad i=1,\ldots,N,\ s\in\cS.
\end{equation}

The target profile $\Pi_t$ is updated by the actors whereas
$p_t$ is the profile actually played.

Exploration in \eqref{eq:fixed-mixing} gives every joint action positive probability at each visited state. Importantly, the mild communication assumption below and the exploration schedules used in our results provide the state--action coverage needed by the critic. The analysis derives this coverage even when some state--action pairs occupy a vanishing fraction of dates.

\begin{assumption}
\label{ass:communication}
The finite-state chain with transition matrix under uniform joint actions $P^{\mathrm U}$ given by 
\begin{equation}
\label{eq:uniform-action-transition}
P^{\mathrm U}(s'\mid s)
:=\frac1{|\cA|}\sum_{a\in\cA}P(s'\mid s,a)
\end{equation}
is irreducible.
\end{assumption}

Assumption~\ref{ass:communication} means that from any state there is a finite path of feasible transitions to any other state. Exploration gives the algorithm a chance to follow such a path after every history. The proof uses this chance to derive visitation bounds for the learning path even when a limiting target policy has several recurrent classes.

In the actor--critic algorithm, each player maintains $Q$ and $V$ critic
tables. Under the model of Section~\ref{sec:model}, player $i$
observes $a_{-i,t}$ but not the opponents' policies. After this observation,
she forms the vector  $Q_{i,t}(s_t,(\cdot,a_{-i,t}))$.
Thus, one realization of the opponents' action supplies a score for every
action of player $i$ through her current $Q$ table. The $Q$ table update is a standard tabular one-step temporal-difference
update for the realized joint action. The $V$ update averages the current $Q$ values under player $i$'s target policy, while the opponents' realized action provides the corresponding averaging over their behavior without requiring their mixed policies to be observed.

The actor uses the same gain vector as feedback for a statewise no-regret
algorithm. We allow for general no-regret actors that belong to the family we now define. 

\begin{definition}[Admissible statewise actor]
\label{def:admissible-actor}
Fix a finite action set $A$ and a positive nonincreasing sequence
$(\beta_n)_{n\ge1}$. On round $n$, a statewise actor chooses
$\pi_n\in\Delta(A)$ before observing $g_n\in[0,H]^A$, using past gains and
any predictable internal state. It is \emph{admissible} if there are finite
nonnegative constants $D_{\mathrm A}$, $C_{\mathrm A}$, and
$L_{\mathrm A}$ such that, for every gain sequence and every $n\ge1$,
\begin{equation}
\label{eq:actor-regret-interface}
\max_{a\in A}\sum_{q=1}^n
\left[g_q(a)-\langle\pi_q,g_q\rangle\right]
\le
\frac{D_{\mathrm A}}{\beta_n}
+C_{\mathrm A}\sum_{q=1}^n\beta_q,
\end{equation}
and
\begin{equation}
\label{eq:actor-movement-interface}
\|\pi_{n+1}-\pi_n\|_1
\le L_{\mathrm A}\beta_n.
\end{equation}
\end{definition}

Condition~\eqref{eq:actor-regret-interface} is a pathwise external-regret
bound. Condition~\eqref{eq:actor-movement-interface} limits how rapidly the policy can change and is typically satisfied by standard no regret algorithms such as  follow the regularized leader (FTRL) and online mirror ascent as we now show. 

 Fix a norm $\|\cdot\|$ on $\R^A$ and let $\|\cdot\|_*$ denote its dual.
For the cumulative gain $S_{n-1}:=\sum_{q<n}g_q$, FTRL with
regularizer $h$ chooses
\[
\pi_n
=
\arg\max_{\pi\in\Delta(A)}
\left\{
\beta_n\langle S_{n-1},\pi\rangle-h(\pi)
\right\}.
\]

We next define online mirror ascent. Suppose that $h$ is differentiable, and
let
$
D_h(\pi,z)
:=
h(\pi)-h(z)-\langle\nabla h(z),\pi-z\rangle
$ 
be its Bregman divergence. Starting from $\pi_1\in\Delta(A)$, ordinary online
mirror ascent updates the policy after observing $g_n$ according to
\[
\pi_{n+1}
\in
\arg\max_{\pi\in\Delta(A)}
\left\{
\beta_n\langle\pi,g_n\rangle-D_h(\pi,\pi_n)
\right\}.
\]

\begin{proposition}[Standard admissible actors]
\label{prop:standard-admissible-actors}
Let $\beta_n=\beta_0n^{-c}$ with $c\in(0,1)$. FTRL is admissible
when $h$ is lower semicontinuous, strongly convex with respect to some norm,
and has finite oscillation
$
\sup_{\pi\in\Delta(A)}h(\pi)
-\inf_{\pi\in\Delta(A)}h(\pi)<\infty.
$ 
Online mirror ascent is admissible when $h$ is differentiable, strongly
convex, and
$
\max_{a\in A}\sup_{\pi\in\Delta(A)}D_h(e_a,\pi)<\infty,
$ 
where $e_a$ is the unit mass on $a$. The corresponding optimistic variants are also admissible when their predictable gain forecasts are uniformly bounded in the associated dual norm.
\end{proposition}

The proposition covers several familiar actors. For
$h(\pi)=\sum_a\pi(a)\log\pi(a)$, FTRL is Hedge. For
$h(\pi)=\|\pi\|_2^2/2$, FTRL is Euclidean dual averaging and
mirror ascent is projected gradient ascent. 
In addition, we show that the corresponding optimistic variants, including optimistic Hedge and the Euclidean optimistic mirror update, are also admissible when based on uniformly bounded predictable gain forecasts; their  updates and verification are given in the proof of Proposition~\ref{prop:standard-admissible-actors}.

\begin{remark}[Heterogeneous actor choices]
\label{rem:heterogeneous-actors}
Each player--state pair $(i,s)$ may use a different admissible actor, with
its own regularizer, forecast, and admissibility constants. The analysis
depends only on the resulting finite collection of constants in
\eqref{eq:actor-regret-interface}--\eqref{eq:actor-movement-interface}.
Adaptive or time-varying regularized learners are therefore also covered
whenever they satisfy these two inequalities.
\end{remark}

For each $(i,s)$, local round $n$ of actor $\mathfrak A_{i,s}$ occurs at date
$t_n(s)$. The policy $\pi_{i,t_n(s)}(\cdot\mid s)$ is the actor's $n$th
decision, and the gain supplied on that round is
$Q_{i,t_n(s)}(s,(\cdot,a_{-i,t_n(s)}))$.  After receiving this gain, the actor produces its
$(n+1)$st decision. 

The algorithm uses the local clocks defined in Section~\ref{sec:model}. It takes critic steps proportional to $n^{-b}$ and actor steps proportional to $n^{-c}$ with $1/2<b<c<1$, using the visit index of the coordinate being updated. At the $n$th
visit to state $s$, the coordinate $V_i(s)$ uses the state clock $n$,
whereas $Q_i(s,a)$ uses the number of observations of the particular
state--action pair $(s,a)$. Thus, unvisited coordinates do not change and
different coordinates evolve on endogenous clocks. The actor step decreases faster than the critic step at the same visit index because $c>b$. A state clock and a state--action clock can nevertheless advance at very different speeds. Proving accurate evaluation requires enough calendar time for even an infrequently updated critic coordinate to respond before the target policy moves far.

The actors update only the rows of the state just visited. On the $n$th visit to that state their total change is at most a constant times $n^{-c}$ by \eqref{eq:actor-movement-interface}. Summing over the local visit clocks gives
\begin{equation}
 \sum_{t<T}\|\Pi_{t+1}-\Pi_t\|_1
 \le C\sum_{s\in\cS}\sum_{n=1}^{N_s(T)}n^{-c}
 \le C' T^{1-c}.
 \label{eq:actor-total-variation}
\end{equation}
Mixing the targets with uniform actions as in \eqref{eq:fixed-mixing} adds at most a constant times $\sum_{t<T}|\eps_{t+1}-\eps_t|$, which is bounded because exploration is nonincreasing. Dividing by $T$ gives \eqref{eq:mbcce-average-increments} on every path of the algorithm. We now describe the updates and prove the required critic estimates.

\begin{algorithm}[t]
\caption{Asynchronous no-regret actor--critic with vanishing exploration}
\label{alg:bao-ac}
\begin{algorithmic}[1]
\State Fix $b\in(\tfrac12,1)$, $c\in(b,1)$,
$\alpha_0,\eta_0\in(0,1]$, $\beta_0>0$, and a nonincreasing exploration schedule $(\eps_t)$ tending to zero. For every
$s\in\cS$, $a\in\cA$, and $n\ge1$, set
\[
\eta_{s,n}=\eta_0n^{-b},\qquad
\alpha_{s,a,n}=\alpha_0n^{-b},\qquad
\beta_{s,n}=\beta_0n^{-c}.
\]
\State For every $i$, $s$, and $a$, initialize
$Q_{i,1}(s,a)\in[0,H]$. For every $i$ and $s$, initialize
$V_{i,1}(s)\in[0,H]$, choose an admissible actor
$\mathfrak A_{i,s}$ with schedule $(\beta_{s,n})_{n\ge1}$, and set
$\pi_{i,1}(\cdot\mid s)$ equal to the decision it produces before its first
local gain.
\For{$t=1,2,\ldots$}
  \State Observe $s=s_t$ and set $n=n_t(s)$.
  \State For every player $i$ and state $r$, set
  \[
  p_{i,t}(\cdot\mid r)
  =(1-\eps_t)\pi_{i,t}(\cdot\mid r)+\eps_t\mathsf U_i(\cdot\mid r).
  \]
  Independently across players, draw
  $a_{i,t}\sim p_{i,t}(\cdot\mid s)$.
  \State Each player $i$ observes $a_t$, $u_i(s,a_t)$, and $s_{t+1}$, and
  forms
  $Q_{i,t}(s,(\cdot,a_{-i,t}))$. 
  \State Let $k=n_t(s,a_t)$. For every player $i$, simultaneously update
  \[
  \begin{aligned}
  Q_{i,t+1}(s,a_t)
  &=(1-\alpha_{s,a_t,k})Q_{i,t}(s,a_t)
  +\alpha_{s,a_t,k}
  \bigl(u_i(s,a_t)+\gamma V_{i,t}(s_{t+1})\bigr),\\
  V_{i,t+1}(s)
  &=(1-\eta_{s,n})V_{i,t}(s)
  +\eta_{s,n}
  \left\langle
\pi_{i,t}(\cdot\mid s),
Q_{i,t}(s,(\cdot,a_{-i,t}))
\right\rangle.
  \end{aligned}
  \]
  All other critic coordinates are unchanged.
  \State For every player $i$, supply $Q_{i,t}(s,(\cdot,a_{-i,t}))$ as the $n$th gain to $\mathfrak A_{i,s}$, advance this actor to local round $n+1$,
  and set $\pi_{i,t+1}(\cdot\mid s)$ equal to the resulting decision.
  For every $r\ne s$, set
  $\pi_{i,t+1}(\cdot\mid r)=\pi_{i,t}(\cdot\mid r)$.
\EndFor
\end{algorithmic}
\end{algorithm}

We measure critic accuracy by the error $e_t$ defined in Theorem~\ref{thm:AMCR-regret-tracking}. Appendix~\ref{app:algorithm} shows that the conditional-mean updates contract toward the Bellman values of the current target profile, with an exploration error bounded by a constant times $\eps_t$.

We can now state the main learning guarantee.

\begin{theorem}[No-regret actor--critic AMCR guarantee]
\label{thm:algorithm-AMCR}
Under Assumption~\ref{ass:communication}, run Algorithm~\ref{alg:bao-ac} with any admissible actors and $\eps_t=1/\log(e+t)$. Almost surely,
\begin{align}
e_t&\longrightarrow0,\label{eq:main-critic-tracking}\\
R_{i,s}^{\mathrm{tar}}(n)/n&\longrightarrow0\qquad (i=1,\ldots,N,\ s\in\cS),\label{eq:main-target-regret}\\
\mathcal G_T^{\mathrm{AMCR}}(p_{1:T})&\longrightarrow0.\label{eq:algorithm-pathwise-AMCR-main}
\end{align}
\end{theorem}

The first conclusion is our main two-timescale result. To prove this, we show that the moving critic recursion is an
asymptotically vanishing perturbation of an asynchronous contraction toward
the Bellman values of $\Pi_t$.
The admissible-actor regret bound is then used to give \eqref{eq:main-target-regret};
combining this estimate with critic tracking in
Theorem~\ref{thm:AMCR-regret-tracking} yields the AMCR guarantee. Equation~\eqref{eq:main-target-regret} controls regret on the visit clock of each state, even if that state occupies a vanishing share of dates. By \eqref{eq:actor-total-variation} and Proposition~\ref{prop:empirical-markov-bcce}, every weak accumulation point of the empirical law of the states and behavior profiles is a continuing-play MBCCE almost surely.

\begin{remark}
\label{rem:target-behavior-transfer}
Recall that the behavior path $p_{1:T}$ describes the implemented play, whereas
$\Pi_{1:T}$ is the target-policy path used in the regret--evaluation
argument. We note that Lemma~\ref{lem:AMCR-advantage-lipschitz} implies
that every AMCR bound for either path transfers to the other at an additional cost bounded by a constant times $T^{-1}\sum_{t\le T}\eps_t$, which tends to zero. We therefore state the equilibrium guarantees for the
behavior path that is actually played by the algorithms.
\end{remark}

The same discrete-time analysis used to prove Theorem~\ref{thm:algorithm-AMCR} also yields a finite-horizon rate with a polynomial exploration schedule. Let $d$ be the largest shortest-path length between two states in the directed support graph of $P^{\mathrm U}$. A path to the desired state requires at most $d$ transitions, followed by one draw of the required joint action. The bound $d\le |\cS|-1$ also permits a schedule chosen using only the numbers of states and players.

\begin{theorem}[Finite-time AMCR rate]
\label{thm:finite-time-AMCR}
Under Assumption~\ref{ass:communication}, fix $0<r<1/[2N(d+1)+4]$ and choose any admissible actors. Run Algorithm~\ref{alg:bao-ac} with $\eps_t=(t+1)^{-r}$, $c=1-r$, and $b>1/2$ sufficiently close to $1/2$. For a constant $C_r$ independent of $T$,
\begin{equation}
\label{eq:finite-time-AMCR-behavior}
\E\big[\mathcal G_T^{\mathrm{AMCR}}(p_{1:T})\big]\le C_r T^{-r}.
\end{equation}
\end{theorem}

Under the same choices, the AMCR gap also vanishes almost surely, and every weak accumulation point of the empirical law of the states and behavior profiles is a continuing-play MBCCE. The expected evaluation error averaged through date $T$ is $O(T^{-r})$ throughout this range. If $r<1/\{4[N(d+1)+1]\}$ and $b$ is chosen sufficiently close to $1/2$, the critic error also tends to zero almost surely and has expectation $O(t^{-r})$ at date $t$. The supremum of the exponents supplied by this analysis is $1/[2N(d+1)+4]$; Appendix~\ref{app:finite-time-AMCR-rate} gives the calculation.

Theorem~\ref{thm:finite-time-AMCR} bounds deviation gains under the observed state frequencies. At a finite horizon those frequencies may differ from the invariant frequencies in \eqref{eq:mbcce-continuing-states}. To obtain a finite-time equilibrium guarantee, draw uniformly from the $T$ profiles actually played, then draw the state from the invariant distribution of the selected profile. Each $p_t$ puts positive probability on every joint action, so Assumption~\ref{ass:communication} gives $P_{p_t}$ a unique invariant distribution $d_{p_t}$. The resulting law is
\begin{equation}
 \bar\mu_T(\{s\},d\Pi)=\frac1T\sum_{t=1}^T d_{p_t}(s)\,\delta_{p_t}(d\Pi).
 \label{eq:actor-generated-law}
\end{equation}
Computing this law would require the transition probabilities. The learners need only their observations to carry out Algorithm~\ref{alg:bao-ac}.

\begin{corollary}[Incentives with invariant state frequencies]
\label{cor:actor-generated-law}
Under the conditions of Theorem~\ref{thm:finite-time-AMCR}, for a constant $C_r$ independent of $T$,
\begin{equation}
 \E\!\left[\sum_{i,s}\left(\max_{b_i\in A_i}\int_{\mathcal C}
 \mathcal A_i(\Pi;s,b_i)\,\bar\mu_T(\{s\},d\Pi)\right)_+\right]
 \le C_r T^{-r}.
 \label{eq:actor-generated-gap}
\end{equation}
\end{corollary}
Thus assigning the profiles their invariant state frequencies preserves the order of the incentive bound.

\textbf{Challenges and technical overview of the analysis.}
The proof of Theorems \ref{thm:algorithm-AMCR} and \ref{thm:finite-time-AMCR} combines three features that are usually treated separately in stochastic-approximation analyses if at all. The critic is asynchronous and uses endogenous local clocks: a value coordinate is updated only when its state is visited, and an action-value coordinate only when the corresponding state--action pair is observed. At the same time, its evaluation target moves continuously because the statewise actors keep updating, and the theorems only assume mild conditions on the actors in order to capture a wide range of no-regret algorithms. Finally, the observations are generated by an exploratory behavior policy that creates a systematic bias that vanishes with exploration. We are not aware of a standard two-timescale stochastic-approximation theorem that directly covers this combination while providing both the almost-sure moving-target tracking and the finite-time estimates.

The analysis proceeds through a direct blockwise tracking argument. The blocks are used only in the proof; the actors continue to update from their observations and do not freeze their policies at common dates. We first freeze the target-policy profile. For every fixed profile, the zero-exploration conditional-mean critic map is a contraction and has a unique fixed point. This fixed point is the canonical evaluation used in AMCR. We then show that it varies Lipschitz continuously with the target profile. The conditional mean of the observed update differs from this map by the current exploration probability times a bounded term.

The next issue is to ensure that every critic coordinate is updated sufficiently often. This cannot be inferred from action exploration alone because the state process is generated by the game and the evolving policies. The communication assumption and exploration together imply that, after every possible history, each state--action pair has a positive probability of being observed within a bounded number of stages. The lower bound on this probability decreases as exploration vanishes. A concentration argument converts this hitting property into visitation guarantees over moving time intervals. Consequently, on every sufficiently late block of the appropriate length, each critic coordinate receives a cumulative update bounded away from zero. The block length is chosen to match the critic timescale and increases as the guaranteed observation probability decreases.

Over one block, each critic coordinate receives sufficient cumulative update to ensure a uniform contraction toward the fixed point associated with the target policy at the beginning of the block. Frequently visited coordinates can receive many more updates in the same block. Because each update averages the previous value with the new target, these extra updates keep a bounded perturbation under control. With the inverse-log exploration schedule, the block remains short enough that the target policy changes only by a vanishing amount. Thus, over each block, the critic makes substantial progress toward the evaluation target while that target is nearly frozen.

The accumulated critic noise is also small. Its treatment requires some care: the realized action both determines which action-value coordinate is updated and generates the value-critic innovation, while the subsequent state transition generates the action-value innovation. We separate these two sources by splitting each stage into two half-steps and constructing a valid martingale decomposition that allows maximal concentration inequalities to control the noise accumulated over the block.

Combining these ingredients gives the central recursion: the tracking error at the end of a block is at most a fixed fraction of the error at its beginning, plus the movement of the target, the accumulated martingale noise, and the exploration error. Iterating this recursion with the inverse-log exploration schedule proves almost-sure tracking. To obtain the polynomial incentive bound, we also sum the contractions, weighting each block by its length. The overall actor movement is controlled on the local state clocks, and each squared critic step enters the sum over disjoint blocks only once. This gives the time-average evaluation estimate required by Theorem~\ref{thm:AMCR-regret-tracking}.

The actor argument is then carried out separately on the visit clock of each state and is more standard. Each observed score is decomposed into its conditional mean and a martingale difference. The pathwise regret guarantee of each admissible actor, together with a standard martingale analysis, yields sublinear target-policy regret with a vanishing exploration error. Once critic tracking and target-policy regret have been established, Theorem~\ref{thm:AMCR-regret-tracking} gives the AMCR guarantee.

\section{Policy gradient algorithm and AMCR guarantees}
\label{sec:reinforce-AMCR}

In this section we obtain exact AMCR with an
episodic projected stochastic policy gradient method. The learner chooses its policy before reset and keeps it fixed until termination. Conditional on that policy, the expected visit counts divided by the expected episode length give \eqref{eq:mbcce-episodic-states}.

We use the stochastic-game primitives from Section~\ref{sec:model} with an
episodic timing protocol. Fix an episode-start distribution
$\rho_0\in\Delta(\cS)$. At the beginning of each episode, the environment
draws the initial state from $\rho_0$. At every stage of the episode, actions
and payoffs are generated as in Section~\ref{sec:model}. After the payoff is
realized, the episode ends with probability $1-\gamma$, independently of the
past. If it continues, the next state is drawn from
$P(\cdot\mid s,a)$. The behavior profile remains fixed throughout the
episode. After termination, all players update, and the next episode begins
from a new draw of the initial state. Because continuation occurs
independently with probability $\gamma$ after each payoff, for every fixed
behavior profile and initial state, the expected undiscounted episode return
equals the corresponding $\gamma$-discounted value.
At each played stage, player $i$ observes the current state, her own action,
her own payoff, and the next state when the episode continues. Such episodic
policy-update protocols are standard in theoretical multi-agent
policy-gradient analyses; see
\cite{DaskalakisFosterGolowich2020,LeonardosEtAl2022,DingWeiZhangJovanovic2022}.

We note two main differences between Algorithm~\ref{alg:bao-ac} and
Algorithm~\ref{alg:reinforce-AMCR} in terms of feedback and updates. Algorithm~\ref{alg:reinforce-AMCR}
uses only each player's own payoff feedback, whereas
Algorithm~\ref{alg:bao-ac} also observes the realized joint action.
In addition, Algorithm~\ref{alg:reinforce-AMCR} updates policies only at
common episode boundaries, while Algorithm~\ref{alg:bao-ac} learns
continuously along a single trajectory. 

Despite their differences, Algorithm~\ref{alg:bao-ac} and
Algorithm~\ref{alg:reinforce-AMCR} obtain AMCR through the same underlying principle: they combine
evaluation of the current joint policy with statewise no-regret learning.
Algorithm~\ref{alg:bao-ac} performs the evaluation explicitly through its
critic and feeds the resulting scores to a no-regret actor. In contrast,
Algorithm~\ref{alg:reinforce-AMCR} evaluates the current joint policy
implicitly from frozen-policy episode returns and uses these returns to form
policy-gradient estimates; the projected-gradient update then provides the
corresponding statewise regret control.

Algorithm~\ref{alg:reinforce-AMCR} follows the standard REINFORCE approach to
estimating policy gradients and uses the same direct parameterization,
likelihood-ratio estimator, and exploratory policy considered in
\cite{LeonardosEtAl2022}.

\begin{algorithm}[t]
\caption{Projected REINFORCE under own-payoff feedback}
\label{alg:reinforce-AMCR}
\begin{algorithmic}[1]
\State Choose positive nonincreasing sequences
$(\varepsilon_e)_{e\ge1}\subset(0,1)$ and $(\eta_e)_{e\ge1}$.
For every player $i$, initialize
$\pi_{i,1}\in\prod_{s\in\cS}\Delta(A_i)$.
\For{$e=1,2,\ldots$}
    \State Every player $i$ sets
    \begin{equation}
    \label{eq:reinforce-behavior}
    p_{i,e}(\cdot\mid s)
    :=
    (1-\varepsilon_e)\pi_{i,e}(\cdot\mid s)
    +
    \varepsilon_e\mathsf U_i(\cdot\mid s),
    \qquad s\in\cS,
    \end{equation}
    and let $p_e=(p_{i,e})_{i=1}^N$.
    \State Independently of the past, the environment draws
    $s_{e,0}\sim\rho_0$ and, independently, an unobserved terminal index
    $L_e$ satisfying
    \begin{equation}
    \label{eq:reinforce-horizon}
    \bbP(L_e=\ell)=(1-\gamma)\gamma^\ell,
    \qquad \ell=0,1,\ldots.
    \end{equation}
    \For{$k=0,\ldots,L_e$}
        \State Every player independently draws
        $a_{i,e,k}\sim p_{i,e}(\cdot\mid s_{e,k})$.
        \State Each player $i$ observes $s_{e,k}$, $a_{i,e,k}$, and
        $u_i(s_{e,k},a_{e,k})$. If $k<L_e$, the environment draws
        $s_{e,k+1}\sim P(\cdot\mid s_{e,k},a_{e,k})$.
    \EndFor
    \State Every player $i$ forms
    \begin{equation}
    \label{eq:reinforce-estimator}
    \widehat G_{i,e}
    :=
    \left(
    \sum_{k=0}^{L_e}u_i(s_{e,k},a_{e,k})
    \right)
    \left(
    \sum_{k=0}^{L_e}
    \nabla_{\pi_i}
    \log p_{i,e}(a_{i,e,k}\mid s_{e,k})
    \right),
    \end{equation}
    where $\nabla_{\pi_i}$ denotes the ambient Euclidean gradient with
    respect to the coordinates
    $(\pi_i(a_i\mid s))_{s\in\cS,a_i\in A_i}$, evaluated at $\pi_{i,e}$.
    \State Every player updates before the next episode by Euclidean projection:
    \begin{equation}
    \label{eq:reinforce-update}
    \pi_{i,e+1}
    :=
    \proj_{\prod_{s\in\cS}\Delta(A_i)}
    \left(
    \pi_{i,e}+\eta_e\widehat G_{i,e}
    \right).
    \end{equation}
\EndFor
\end{algorithmic}
\end{algorithm}

Index all played stages consecutively across episodes by $t=1,2,\ldots$.
If stage $t$ belongs to episode $e$, set $p_t:=p_e$, where $p_e$ is the
behavior profile held fixed throughout that episode, and let $s_t$ and $a_t$
denote the corresponding state and joint action. Thus
$(p_t,s_t,a_t)_{t\ge1}$ describes the sequence of played stages across
episodes. Throughout this section,
$\mathcal G_T^{\mathrm{AMCR}}$ is evaluated on the first $T$ stages of this
sequence. 
Write $Y_e=L_e+1$ for the length of episode $e$ and $K_E=\sum_{e=1}^E Y_e$ for the number of stages in the first $E$ episodes. 
Set $A_{\max}:=\max_i|A_i|$.

For a finite-time equilibrium guarantee, use the policy lottery from the observed stages and compute its state probabilities by \eqref{eq:mbcce-episodic-states}. Set
\begin{equation}
 \bar\mu_{K_E}(\{s\},d\Pi)
 =\frac{1}{K_E}\sum_{e=1}^E Y_e\,d_{p_e}(s)\,\delta_{p_e}(d\Pi).
 \label{eq:reinforce-generated-law}
\end{equation}
An episode of length $Y_e$ contributes $Y_e$ stages to the empirical lottery, so its profile has weight $Y_e/K_E$ from that episode. For this profile, \eqref{eq:reinforce-generated-law} uses discounted visit frequencies in place of the observed state frequencies. The resulting lottery has the episodic state frequencies in \eqref{eq:mbcce-episodic-states} at this horizon. We use it to state a finite-time result. Players update using their own episode observations.

\begin{theorem}
\label{thm:reinforce-AMCR}
For any $\rho_0\in\Delta(\cS)$, run
Algorithm~\ref{alg:reinforce-AMCR} with
$\varepsilon_e=(e+1)^{-1/3}$ and
$\eta_e=\sqrt{|\cS|}(e+1)^{-2/3}/(A_{\max}H^2)$.
There is a universal constant $C<\infty$ such that, for every $E\ge1$,
\begin{equation}
\label{eq:reinforce-expected-rate}
\E\!\left[\mathcal G_{K_E}^{\mathrm{AMCR}}\right]
\le
C N H A_{\max}^{3/2}\sqrt{|\cS|}\,E^{-1/3}.
\end{equation}
The law in \eqref{eq:reinforce-generated-law} also satisfies
\begin{equation}
 \E\!\left[\sum_{i,s}\left(\max_{b_i\in A_i}\int
 \mathcal A_i(\Pi;s,b_i)\,\bar\mu_{K_E}(\{s\},d\Pi)\right)_+\right]
 \le C N H A_{\max}^{3/2}\sqrt{|\cS|}\,E^{-1/3},
 \label{eq:reinforce-episodic-gap}
\end{equation}
after increasing $C$ if needed.
Moreover, $\mathcal G_T^{\mathrm{AMCR}}\to0$ almost surely, and every empirical accumulation point is an episodic MBCCE.
\end{theorem}

Equation~\eqref{eq:reinforce-episodic-gap} concerns a law that already satisfies the required state frequencies, even at a finite horizon. Its incentive bound is the sum of the positive deviation gains weighted by the probability of the state at which the deviation occurs. For the realized state--policy pairs, holding each policy fixed through an independently terminated episode gives the same state frequencies in the limit by Proposition~\ref{prop:empirical-markov-bcce}.

\section {Further properties of MBCCE}
\label{sec:properties}

Beyond the comparison in Example~\ref{ex:stochastic-investment-MBCCE}, we next
discuss two further properties of MBCCE that illustrate its parallels with coarse
correlated equilibrium in finite normal-form games: its linear structure and its
welfare guarantees under smoothness.

Fix a finite list of candidate profiles $(\Pi^k)_k$, obtained, for example, from a discretization of the policy space or an observed policy path. All sums over $k$ below range over this list. For each $\Pi^k$, first evaluate the
policy by solving \eqref{eq:canonical-bellman-system}. When the decision process specifies $d_{\Pi^k}$, the joint weight on state $s$ and profile $\Pi^k$ must be $d_{\Pi^k}(s)$ times the profile's prior probability. For example, \eqref{eq:mbcce-episodic-states} gives these coefficients for episodic play. Once these evaluations and state frequencies are fixed, the probabilities assigned to the candidate profiles can be chosen by a linear program.

In continuing play, assigning a profile its prior probability can leave open how its visits divide across closed classes. We therefore use its joint weights $x_{s,k}$ with the states as the decision variables. When the class weights are to be determined, the weight assigned to a state must equal the total weight entering that state under the profile's transitions. When the process specifies $d_{\Pi^k}$, we use those probabilities directly. Let $z_{i,s}$ bound player $i$'s largest positive deviation gain at $s$, weighted by the probability of $s$. Using the appropriate state-frequency equation, solve
\begin{equation}
\label{eq:markov-bcce-finite-lp}
\begin{aligned}
\min_{x,z}\quad
&\sum_{s\in\cS}\sum_{i=1}^N z_{i,s}\\
\text{subject to}\quad
&\sum_k x_{s,k}\mathcal A_i(\Pi^k;s,b_i)\le z_{i,s}
&&\forall i,s,b_i,\\
&x_{s,k}=
\begin{cases}
\sum_r x_{r,k}P_{\Pi^k}(s\mid r) &\text{for unspecified class weights},\\
d_{\Pi^k}(s)\sum_r x_{r,k} &\text{for specified } d_{\Pi^k}
\end{cases}
&&\forall s,k,\\
&\sum_{s,k}x_{s,k}=1,\qquad x_{s,k}\ge0,\quad z_{i,s}\ge0.
\end{aligned}
\end{equation}
Here $\sum_s x_{s,k}$ is the probability assigned to $\Pi^k$, and $\sum_k x_{s,k}$ is the frequency of state $s$ implied by the lottery. Every coefficient in \eqref{eq:markov-bcce-finite-lp} is known once the profiles have been evaluated. An optimum of zero is equivalent to the existence of an MBCCE supported on the candidate list. Otherwise the optimum gives the smallest sum of these weighted deviation gains on that list. Finding a suitable list of profiles in a large game remains a separate problem. For the two protocols in Proposition~\ref{prop:empirical-markov-bcce}, the empirical weights give exactly the AMCR incentive gap at each horizon and satisfy the corresponding state-frequency equalities in the limit.

MBCCE also preserves the welfare implications of smoothness developed in
\cite{Roughgarden2009}: a pointwise deviation inequality transfers to
equilibrium outcomes through coarse obedience. Here the deviation is chosen
after observing the public state, and its payoff is evaluated by the
corresponding $Q$-value. The difference from the static setting is that a
current action can also change the next-state distribution and therefore
continuation payoffs. The nonnegative term $\xi$ in the following definition bounds the possible loss in total deviation gains caused by the actions' effects on future states.

For every state $s$, fix
$a^\star(s)\in\arg\max_{a'\in\cA}\sum_{i=1}^N u_i(s,a')$.

\begin{definition}[Markov smoothness]
\label{def:markov-smoothness}
Fix $\lambda>0$, $\kappa\ge0$, and a nonnegative bounded measurable function
$\xi:\cS\times\mathcal C\times\cA\to\R_+$. The stochastic game is
$(\lambda,\kappa,\xi)$-Markov smooth if, for every
$s\in\cS$, $\Pi\in\mathcal C$, and $a\in\cA$,
\begin{equation}
\label{eq:markov-q-smoothness}
\begin{aligned}
\sum_{i=1}^N
\left[
Q_i^\Pi(s,(a_i^\star(s),a_{-i}))-Q_i^\Pi(s,a)
\right]
\ge{}&
\lambda\sum_{i=1}^N u_i(s,a^\star(s))\\
&-(1+\kappa)\sum_{i=1}^N u_i(s,a)
-\xi(s,\Pi,a).
\end{aligned}
\end{equation}
\end{definition}

When $\gamma=0$ and $\xi=0$, condition
\eqref{eq:markov-q-smoothness} is equivalent to the usual utility-smoothness
condition
\[
\sum_{i=1}^N u_i(s,(a_i^\star(s),a_{-i}))
\ge
\lambda\sum_{i=1}^N u_i(s,a^\star(s))
-\kappa\sum_{i=1}^N u_i(s,a).
\]
Thus Proposition~\ref{prop:mbcce-smoothness} below recovers the familiar
welfare factor $\lambda/(1+\kappa)$. Although
\eqref{eq:markov-q-smoothness} involves continuation values,
Section~\ref{subsec:transition-aggregative-mcce} shows that it follows
from ordinary statewise smoothness and a bound on the effect of unilateral
actions on transitions.

\begin{proposition}
\label{prop:mbcce-smoothness}
Suppose \eqref{eq:markov-q-smoothness} holds. Then every MBCCE $\mu$
satisfies, for each state in the support of its state marginal,
\begin{equation}
\label{eq:mbcce-smoothness-bound}
\begin{aligned}
\E_{\substack{\Pi\sim\mu(\cdot\mid s) \\
a\sim\Pi(\cdot\mid s)}}
\left[\sum_{i=1}^N u_i(s,a)\right]
\ge
\frac{1}{1+\kappa}
\left(
\lambda\sum_{i=1}^N u_i(s,a^\star(s))
-
\E_{\substack{\Pi\sim\mu(\cdot\mid s)\\
a\sim\Pi(\cdot\mid s)}}
[\xi(s,\Pi,a)]
\right).
\end{aligned}
\end{equation}
\end{proposition}

Averaging \eqref{eq:mbcce-smoothness-bound} with the state frequencies of the equilibrium compares the welfare of the lottery with the maximum stage welfare at the states its policies actually visit. In episodic play, averaging its left-hand side over the public state gives $(1-\gamma)\int\sum_{i=1}^N\sum_s\rho_0(s)V_i^\Pi(s)\,\mu(\cS,d\Pi)$, the normalized discounted welfare from the reset distribution. In continuing play the same averaging gives mean welfare at a sampled decision date.

\section{Converse}
\label{sec:converse}

Theorem~\ref{thm:AMCR-regret-tracking} and
Proposition~\ref{prop:empirical-markov-bcce} show that statewise
no regret and accurate current-policy evaluation lead to MBCCE
under the updating protocols considered here. We now prove the
converse under the conditions below: every MBCCE can arise from
algorithms satisfying these requirements. In this sense, the MBCCE
set is tight: it consists exactly of the possible empirical limits
of reinforcement-learning algorithms satisfying these two familiar
requirements statewise no regret and policy evaluation.\footnote{Note that  a particular
algorithm such as the policy gradient method or actor-critic algorithm we studied may generally approach a smaller set.}

Use $R_{i,s}^{\mathrm{tar}}$ and $e_t$ from Section~\ref{sec:AMCR},
evaluated with $\Pi_t=p_t$. In episodic play, a state is reachable
if some sequence of joint actions starting in the support of the
reset law $\rho_0$ reaches it with positive probability.
Restrict the norm defining $e_t$ to all coordinates at these states.

\begin{theorem}[Converse]
\label{thm:mbcce-converse}
Consider either of the following protocols:
\begin{enumerate}[label=(\roman*)]
\item \emph{Episodic play:} fix an initial-state distribution $\rho_0\in\Delta(\cS)$, choose the policy profile before each episode reset, and hold it fixed throughout the episode, which continues after each stage independently with probability $\gamma$;
\item \emph{Continuing play:} Assumption~\ref{ass:communication} holds, play evolves along a single trajectory, and the executed policy path satisfies \eqref{eq:mbcce-average-increments} almost surely.
\end{enumerate}
Under either protocol, for every MBCCE $\mu$, there exist algorithms maintaining online value tables $(Q_t,V_t)$ whose joint play satisfies, almost surely,
\begin{equation*}
 \widehat\mu_T\Rightarrow\mu,
 \qquad
 \frac1T\sum_{s,i}R_{i,s}^{\mathrm{tar}}(N_s(T))
 \longrightarrow0,
 \qquad
 \frac1T\sum_{t=1}^T e_t\longrightarrow0.
\end{equation*}
\end{theorem}

The algorithms may depend on $\mu$ and use prescriptions agreed
in advance. During play, each player observes only public states,
her own actions and payoffs. The regret and evaluation guarantees
hold in joint play of the constructed algorithms. At every state
with positive limiting frequency, regret per visit also vanishes.

\section{Markov coarse correlated equilibrium}
\label{sec:mcce}

In this section we study the connection between MBCCE and MCCE. 
 MCCE allows a player to replace her entire
stationary policy and evaluates the deviation against one stationary
Markov policy.

For a stationary statewise joint-action law
$\sigma=(\sigma^s)_{s\in\cS}\in\prod_s\Delta(\cA)$, let $V_i^\sigma$ denote
the canonical value defined by \eqref{eq:canonical-bellman-system}. For a
 deterministic Markov policy $\beta_i:\cS\to A_i$, let
$V_i^{\beta_i,\sigma_{-i}}(s)$ be player $i$'s value from state $s$ when she
uses $\beta_i$ and the opponents use the marginals $\sigma_{-i}$. By standard dynamic programming arguments a 
deterministic stationary best response to the opponents exists in our setting.

\begin{definition}[Stationary MCCE]
\label{def:mcce}
For a stationary joint-action law $\sigma$, define
\begin{equation}
\label{eq:normalized-mcce-gap}
\GapMCCE(\sigma)
:=(1-\gamma)
\max_{\substack{i,\,s\in\cS\\ \beta_i:\cS\to A_i}}
\left(V_i^{\beta_i,\sigma_{-i}}(s)-V_i^\sigma(s)\right)_+.
\end{equation}
The law $\sigma$ is a $\delta$-MCCE if $\GapMCCE(\sigma)\le\delta$.
\end{definition}

Definition~\ref{def:mcce} uses the perfect criterion, which tests every
initial state and agrees with Definition~1 in
\cite{DaskalakisGolowichZhang2023}. 
We now give the exact connection of MCCE with MBCCE. Let
$\mu\in\Delta(\cS\times\mathcal C)$ assign positive probability to every
state and define its statewise average action law by
\begin{equation}
\label{eq:mbcce-induced-action-law}
\sigma_\mu^s(a)
:=\E_{\Pi\sim\mu(\cdot\mid s)}[\Pi(a\mid s)],
\qquad s\in\cS,\ a\in\cA.
\end{equation}
Let
$\omega_{i,s_0}^{\beta_i,\sigma_{-i}}(s)$ denote the normalized discounted
probability of visiting state $s$ when play starts from $s_0$, player $i$ uses
$\beta_i$, and the opponents use $\sigma_{-i}$:
\begin{equation}
\label{eq:normalized-deviation-occupancy}
\omega_{i,s_0}^{\beta_i,\sigma_{-i}}(s)
:=(1-\gamma)\sum_{m=0}^{\infty}\gamma^m
\Pr_{\beta_i,\sigma_{-i}}(S_m=s\mid S_0=s_0).
\end{equation}
Thus $\omega_{i,s_0}^{\beta_i,\sigma_{-i}}$ is a probability distribution on
$\cS$ and let the one-step deviation gain under $\sigma_\mu$ be 
\begin{equation}
\label{eq:averaged-one-step-gain}
g_i^\mu(s,b_i)
:=\sum_{a_{-i}}\sigma_{\mu,-i}^s(a_{-i})
Q_i^{\sigma_\mu}(s,(b_i,a_{-i}))-V_i^{\sigma_\mu}(s).
\end{equation}
Note that the following decomposition holds
\begin{align}
g_i^\mu(s,b_i)
=\E_{\substack{\Pi\sim\mu(\cdot\mid s)\\ a\sim\Pi(\cdot\mid s)}}
\Big[&Q_i^\Pi(s,(b_i,a_{-i}))-Q_i^\Pi(s,a)
\notag\\[-0.2em]
&+\gamma\langle
P(\cdot\mid s,(b_i,a_{-i}))-P(\cdot\mid s,a),
V_i^{\sigma_\mu}-V_i^\Pi\rangle\Big].
\label{eq:exact-mbcce-mcce-local}
\end{align}
The expected $Q$-difference in \eqref{eq:exact-mbcce-mcce-local} is exactly the
deviation gain in \eqref{eq:markov-bcce-obedience}. The second term is the change in the deviation gain caused by replacing $V_i^\Pi$ with $V_i^{\sigma_\mu}$.

\begin{lemma}
\label{lem:exact-mbcce-mcce}
For every $\mu$ with full-support state marginal,
\begin{equation}
\label{eq:exact-mbcce-mcce}
\GapMCCE(\sigma_\mu)
=
\max_{\substack{i,\,s_0\in\cS\\ \beta_i:\cS\to A_i}}
\left(
\sum_{s\in\cS}
\omega_{i,s_0}^{\beta_i,\sigma_{\mu,-i}}(s)
g_i^\mu(s,\beta_i(s))
\right)_+.
\end{equation}
\end{lemma}

Lemma~\ref{lem:exact-mbcce-mcce} writes the gain from a stationary policy
deviation as an occupancy-weighted sum of one-stage gains. MBCCE
controls the $Q$-difference in \eqref{eq:exact-mbcce-mcce-local} at every
state. A stationary MCCE deviation selects one action at every state, and
\eqref{eq:exact-mbcce-mcce} weights these gains by the states reached under
that deviation. If $\mu$ is an MBCCE, only the continuation term in
\eqref{eq:exact-mbcce-mcce-local} can increase the MCCE gap.
To isolate this remaining  term, define
\begin{align}
\mathsf C(\mu)
:=\max_{\substack{i,\,s_0\in\cS\\ \beta_i:\cS\to A_i}}
\Bigg(&
\sum_{s\in\cS}
\omega_{i,s_0}^{\beta_i,\sigma_{\mu,-i}}(s)
\E_{\substack{\Pi\sim\mu(\cdot\mid s)\\ a\sim\Pi(\cdot\mid s)}}
\Big[
\notag\\[-0.3em]
&\left\langle
P(\cdot\mid s,(\beta_i(s),a_{-i}))-P(\cdot\mid s,a),
V_i^{\sigma_\mu}-V_i^\Pi
\right\rangle
\Big]
\Bigg)_+.
\label{eq:continuation-discrepancy}
\end{align}
Because every discounted occupancy is a probability distribution,
\eqref{eq:exact-mbcce-mcce-local} and
Lemma~\ref{lem:exact-mbcce-mcce} imply
\begin{align}
\GapMCCE(\sigma_\mu)
\le{}&
\max_{i,s,b_i}
\left(
\E_{\substack{\Pi\sim\mu(\cdot\mid s)\\a\sim\Pi(\cdot\mid s)}}
[Q_i^\Pi(s,(b_i,a_{-i}))-Q_i^\Pi(s,a)]
\right)_+
+\gamma\mathsf C(\mu).
\label{eq:mbcce-mcce-upper-bound}
\end{align}
Thus, an MBCCE whose state marginal has full support satisfies
$\GapMCCE(\sigma_\mu)\le\gamma\mathsf C(\mu)$. For the empirical
distributions generated by our algorithms, this compares the continuation value
under the policy used at each visit to a state with the continuation value under
the statewise average of the policies used at visits to that state.

Take $\Gamma_t=p_t$ in Proposition~\ref{prop:empirical-markov-bcce} and let
$\widehat\mu_T:=T^{-1}\sum_{t=1}^T\delta_{(s_t,p_t)}$. At horizons for which
every state has been visited, define
\begin{align*}
\sigma_T^s(a)
&:=\frac1{N_s(T)}\sum_{t\in T_s(T)}p_t(a\mid s), \quad
\widehat\sigma_T^s(a)
=\frac1{N_s(T)}\sum_{t\in T_s(T)}\one\{a_t=a\}.
\end{align*}
By construction, $\sigma_T=\sigma_{\widehat\mu_T}$. Thus $\sigma_T$ is the
stationary action law induced by the empirical state--policy distribution,
whereas $\widehat\sigma_T$ is constructed from realized actions.

\begin{corollary}
\label{cor:algorithm-empirical-mcce}
Consider either of the following cases.
\begin{enumerate}[label=(\roman*)]
\item Assumption~\ref{ass:communication} holds and Algorithm~\ref{alg:bao-ac} uses the inverse-log schedule of Theorem~\ref{thm:algorithm-AMCR} or one of the polynomial schedules for which Theorem~\ref{thm:finite-time-AMCR} proves almost-sure pointwise critic tracking.
\item Algorithm~\ref{alg:reinforce-AMCR} is run under the schedules in
Theorem~\ref{thm:reinforce-AMCR}, and the episode-start distribution has full
support: $\rho_0(s)>0$ for every $s\in\cS$.
\end{enumerate}

Then, almost surely,
\begin{equation}
\label{eq:algorithm-empirical-mcce}
\limsup_{T\to\infty}\GapMCCE(\widehat\sigma_T)
\le\gamma\limsup_{T\to\infty}\mathsf C(\widehat\mu_T).
\end{equation}

\end{corollary}

In continuing play a state can be visited infinitely often even when its share of dates vanishes. We use the actor's regret and critic accuracy along visits to that state to control the MCCE deviation there. In episodic play the positive reset probability bounds each state's share of dates away from zero, so vanishing AMCR suffices for the same conclusion.

\subsection{Applications: transition-aggregative games}
\label{subsec:transition-aggregative-mcce}

Many economic and operations models have transitions that depend on a weighted
aggregate of individual actions. In large games, a single player's action can
then have only a small effect on the next-state distribution, even when stage
payoffs exhibit unrestricted strategic interactions. The previous section shows
that this is precisely the type of setting in which our algorithms can achieve
approximate MCCE. We now formalize this idea through
transition-aggregative stochastic games.

\begin{definition}[Transition-aggregative stochastic game]
\label{def:transition-aggregative-game}
The game is transition-aggregative if there are weights
$w_1,\ldots,w_N\ge0$ with $\sum_jw_j=1$ and the following property. For every
state $s$, there exist $d_s\in\mathbb N$, maps
$\phi_{j,s}:A_j\to\R^{d_s}$, a map $\mathcal P_s$ from feasible aggregate
vectors into $\Delta(\cS)$, and $L_s<\infty$ such that
\begin{align}
P(\cdot\mid s,a)
&=\mathcal P_s\left(\sum_{j=1}^Nw_j\phi_{j,s}(a_j)\right),
&&\forall a\in\cA,
\label{eq:transition-aggregative-kernel}\\
\TV(\mathcal P_s(x),\mathcal P_s(y))
&\le L_s\|x-y\|,
&&\text{for all feasible }x,y.
\label{eq:transition-aggregative-lipschitz}
\end{align}
\end{definition}

For each $i$ and $s$, let
$\Delta_{i,s}:=\max_{a_i,b_i\in A_i}
\|\phi_{i,s}(b_i)-\phi_{i,s}(a_i)\|$.
Then $L_sw_i\Delta_{i,s}$ bounds the change in the next-state law caused by a
unilateral action at $s$. For the sharper bound below, we also use the common
minorization condition
\begin{equation}
\label{eq:transition-common-minorization}
P(\cdot\mid s,a)\ge\eta q(\cdot)
\qquad\forall s,a,
\end{equation}
componentwise, for some $q\in\Delta(\cS)$ and $\eta\in(0,1]$.

\begin{proposition}[Transition influence and approximate MCCE]
\label{prop:transition-aggregative-mcce}
Under either case of Corollary~\ref{cor:algorithm-empirical-mcce}, suppose the
game is transition-aggregative. 
Almost surely,
\[
\limsup_{T\to\infty}\GapMCCE(\widehat\sigma_T)
\le
\frac{2\gamma}{1-\gamma}
\max_{i,s} L_s w_i\Delta_{i,s}.
\]

Under \eqref{eq:transition-common-minorization}, the denominator $1-\gamma$
can be replaced by $1-\gamma(1-\eta)$.
\end{proposition}

Thus, for a sequence of games with a common discount factor and $w_i=1/N$, if
$\max_{i,s}L_s\Delta_{i,s}$ is bounded independently of $N$, 
both algorithms yield an $O(1/N)$-approximate MCCE under the conditions of the proposition. The following examples show how individual influence on transitions determines this bound.

\begin{example}[Dynamic routing]
Let $E_s$ be the set of network edges available at state $s$. For a route
$a_i\in A_i$, define $\phi_{i,s}(a_i)\in\{0,1\}^{|E_s|}$ by
$[\phi_{i,s}(a_i)]_e=\one\{e\text{ is used by route }a_i\}$. If $w_i$ is
player $i$'s share of total demand, then the $e$th coordinate of
$\sum_iw_i\phi_{i,s}(a_i)$ is the fraction of demand using edge $e$. The
transition law can therefore describe how current traffic affects future
congestion.

Use the $\ell_1$ norm and suppose every feasible route uses at most $\ell$
edges. Changing routes removes at most $\ell$ edges and adds at most $\ell$
edges, so $\Delta_{i,s}\le2\ell$. Proposition~\ref{prop:transition-aggregative-mcce}
gives
\[
\limsup_{T\to\infty}\GapMCCE(\widehat\sigma_T)
\le
\frac{4\gamma\ell}{1-\gamma}
\max_{i,s}L_sw_i.
\]
For games with a common discount factor, if the maximum route length $\ell$ and $\max_s L_s$ are bounded independently of $N$ and $\max_i w_i=O(1/N)$, the bound is $O(1/N)$.
We note that stage payoffs may include nonlinear congestion costs, heterogeneous tolls,
and other player-specific parameters.
\end{example}

\begin{example}[Dynamic common-pool resource]
Let $\cS=\{\mathrm g,\mathrm d\}$ denote good and degraded resource
conditions. Each $A_i$ is a finite menu of resource-use decisions, and action
$a_i$ induces a depletion intensity $x_i(a_i)\in[0,1]$. Let
$m(a):=N^{-1}\sum_jx_j(a_j)$ and suppose
\[
P(\mathrm d\mid\mathrm g,a)=\alpha_{\mathrm g}+b_{\mathrm g}m(a),
\qquad
P(\mathrm g\mid\mathrm d,a)=\alpha_{\mathrm d}-b_{\mathrm d}m(a),
\]
where $b_{\mathrm g},b_{\mathrm d}\ge0$. Greater aggregate depletion makes
degradation more likely and recovery less likely. Assume every transition
probability belongs to $[\underline p,1-\underline p]$ for every feasible
aggregate, where $\underline p\in(0,1/2]$.

Take $w_i=1/N$, $\phi_{i,s}(a_i)=x_i(a_i)$,
$L=\max\{b_{\mathrm g},b_{\mathrm d}\}$, and $\Delta_{i,s}\le1$. Moreover,
$P(\cdot\mid s,a)\ge2\underline p\,q$ for $q=(1/2,1/2)$, so
\eqref{eq:transition-common-minorization} holds with
$\eta=2\underline p$. Proposition~\ref{prop:transition-aggregative-mcce}
therefore gives
\[
\limsup_{T\to\infty}\GapMCCE(\widehat\sigma_T)
\le
\frac{2\gamma L}{N[1-\gamma(1-2\underline p)]}
\qquad\text{almost surely}.
\]

For fixed discount and transition parameters the bound is $O(1/N)$.
 Rewards may include nonlinear revenue, state-dependent operating
costs, and asymmetric strategic interactions. The same argument applies to
any finite collection of resource-quality states when
$P(\cdot\mid s,a)=\mathcal P_s(m(a))$ is Lipschitz.
\end{example}

The final result in this section shows that, when each statewise stage game is
smooth, the transition-aggregative condition also implies the Markov smoothness
condition in Proposition~\ref{prop:mbcce-smoothness}, and therefore yields
welfare guarantees for MBCCE.  Routing games of the type considered above are examples in
which the stage game is smooth under  suitable utility
specifications.

\begin{corollary}[Welfare under transition aggregation]
\label{cor:transition-aggregative-mbcce-welfare}
Suppose the game is transition-aggregative and, for some $\lambda>0$ and
$\kappa\ge0$, every statewise stage game satisfies
\begin{equation}
\label{eq:stage-utility-smoothness}
\sum_{i=1}^N u_i(s,(a_i^\star(s),a_{-i}))
\ge
\lambda\sum_{i=1}^N u_i(s,a^\star(s))
-\kappa\sum_{i=1}^N u_i(s,a)
\qquad\forall s,a.
\end{equation}
Then every MBCCE $\mu$ satisfies, at each state in the support of its
state marginal,
\begin{equation}
\label{eq:transition-aggregative-mbcce-welfare-uniform}
\E_{\substack{\Pi\sim\mu(\cdot\mid s)\\ a\sim\Pi(\cdot\mid s)}}
\left[\sum_{i=1}^N u_i(s,a)\right]
\ge
\frac{\lambda}{1+\kappa}\sum_{i=1}^N u_i(s,a^\star(s))
-
\frac{\gamma H L_s}{1+\kappa}\sum_{i=1}^N w_i\Delta_{i,s}.
\end{equation}
Under \eqref{eq:transition-common-minorization}, $H$ can be replaced by
$[1-\gamma(1-\eta)]^{-1}$.
\end{corollary}

\section{Conclusions}

In this paper, we introduce adaptive Markov coarse regret (AMCR), a pathwise regret criterion for sequences of stationary policy profiles, and Markov Bayes coarse correlated equilibrium (MBCCE), a static distributional solution concept over states and stationary policy profiles. For the continuing and episodic protocols studied here, we show that vanishing AMCR implies that every accumulation point of the empirical state--policy distribution is an MBCCE, and that achieving AMCR reduces to statewise external regret and accurate evaluation of the current joint policy. We establish vanishing AMCR for a decentralized asynchronous actor--critic algorithm with vanishing exploration and an episodic multi-agent projected policy-gradient method. We also prove that every MBCCE can arise under the regret and evaluation criteria: in episodic play under the episode protocol, and in continuing play under the communication condition and the same control of average policy changes used by the actor--critic. Finally, we relate MBCCE to stationary MCCE and use this connection to obtain approximate MCCE guarantees for classes of stochastic games in which individual actions have limited effects on state transitions.

Several research questions remain open; we discuss two here. First, when a vanishing-AMCR path does converge, its limit is an MPE under state-space coverage, raising the question of when standard multi-agent reinforcement-learning algorithms possess such last-iterate guarantees. Existing results obtain convergence guarantees under additional game structure such as for policy-gradient methods in Markov potential games \cite{LeonardosEtAl2022}, but  last-iterate results remain limited. Second, our analysis is tabular. Extending the AMCR--MBCCE framework to large state spaces, where reinforcement-learning algorithms rely on function approximation for policy evaluation and optimization, is an important direction for bringing these equilibrium guarantees closer to modern reinforcement learning.

\clearpage
\appendix

\section{Example~\ref{ex:stochastic-investment-MBCCE} calculations}

\begin{proof}[Proof of the claims in Example~\ref{ex:stochastic-investment-MBCCE}]
We first derive the CCE region of the static game. Let $\lambda_{ab}$ denote the
probability of action profile $(a,b)\in\{0,1\}^2$, where the first coordinate
is player $1$'s action. Following the draw gives each player expected payoff
$\lambda_{01}+\lambda_{10}$. If player $1$ deviates to action $1$, her expected
payoff is $\lambda_{00}+\lambda_{10}$; if she deviates to action $0$, it is
$\lambda_{01}+\lambda_{11}$. Hence her CCE inequalities are
$\lambda_{01}\ge\lambda_{00}$ and $\lambda_{10}\ge\lambda_{11}$. The
corresponding inequalities for player $2$ are
$\lambda_{10}\ge\lambda_{00}$ and $\lambda_{01}\ge\lambda_{11}$. Thus, the
distribution is a CCE if and only if
\[
\lambda_{01}\ge\lambda_{00},
\qquad
\lambda_{10}\ge\lambda_{00},
\qquad
\lambda_{10}\ge\lambda_{11},
\qquad
\lambda_{01}\ge\lambda_{11}.
\]

Let $q_1:=\lambda_{10}+\lambda_{11}$ and
$q_2:=\lambda_{01}+\lambda_{11}$ be the induced probabilities of action $1$.
Using $\sum_{a,b}\lambda_{ab}=1$,
\begin{align*}
2-(2q_1+q_2)
&=2\lambda_{00}+\lambda_{01}-\lambda_{11},&
2-(q_1+2q_2)
&=2\lambda_{00}+\lambda_{10}-\lambda_{11},\\
(2q_1+q_2)-1
&=\lambda_{10}+2\lambda_{11}-\lambda_{00},&
(q_1+2q_2)-1
&=\lambda_{01}+2\lambda_{11}-\lambda_{00}.
\end{align*}
The CCE inequalities therefore imply
\begin{equation}
\label{eq:example-cce-projection}
2q_1+q_2\le2,
\qquad
q_1+2q_2\le2,
\qquad
2q_1+q_2\ge1,
\qquad
q_1+2q_2\ge1.
\end{equation}

Conversely, suppose $(q_1,q_2)\in[0,1]^2$ satisfies
\eqref{eq:example-cce-projection}. If $q_1+q_2\le1$, take
$(\lambda_{00},\lambda_{10},\lambda_{01},\lambda_{11})
=(1-q_1-q_2,q_1,q_2,0)$. The two lower inequalities in
\eqref{eq:example-cce-projection} are exactly
$\lambda_{10}\ge\lambda_{00}$ and $\lambda_{01}\ge\lambda_{00}$; the other two
CCE inequalities are automatic. If $q_1+q_2\ge1$, take
$(\lambda_{00},\lambda_{10},\lambda_{01},\lambda_{11})
=(0,1-q_2,1-q_1,q_1+q_2-1)$. The two upper inequalities in
\eqref{eq:example-cce-projection} are exactly
$\lambda_{10}\ge\lambda_{11}$ and $\lambda_{01}\ge\lambda_{11}$; the other two
CCE inequalities are automatic. Hence \eqref{eq:example-cce-projection}
characterizes the projected CCE set. Its vertices are $(1,0)$,
$(2/3,2/3)$, $(0,1)$, and $(1/3,1/3)$. The Nash marginal profiles are
$(1,0)$, $(0,1)$, and $(1/2,1/2)$ by the best-response threshold $1/2$.

We next consider the stochastic investment game. The payoffs in the example are in economic units. To apply the learning results with rewards in $[0,1]$, use $(u_i+1/2)/(3/2)$. The affine change multiplies every deviation gain by $2/3$, preserving all the incentive inequalities. Fix a stationary product
profile $\Pi$. Its Bellman equations are
\begin{align*}
V_i^\Pi(0)
&=
-\frac12\pi_i(1\mid0)
+\frac45\left[
\frac{\pi_1(1\mid0)+\pi_2(1\mid0)}{3}V_i^\Pi(1)
+\left(
1-\frac{\pi_1(1\mid0)+\pi_2(1\mid0)}{3}
\right)V_i^\Pi(0)
\right],\\
V_i^\Pi(1)
&=
1-\frac12\pi_i(1\mid1)
+\frac45\left[
\frac34V_i^\Pi(1)+\frac14V_i^\Pi(0)
\right].
\end{align*}
Solving these two equations for the value difference gives
\[
V_i^\Pi(1)-V_i^\Pi(0)
=
\frac{
1+\frac12\left[\pi_i(1\mid0)-\pi_i(1\mid1)\right]
}{
\frac25+\frac4{15}
\left[\pi_1(1\mid0)+\pi_2(1\mid0)\right]
}.
\]

For $j\ne i$, changing the current state-$0$ action from $0$ to $1$ lowers
the current payoff by $1/2$ and raises the probability of state $1$ by $1/3$,
independently of player $j$'s current action. Hence, for every
$a_j\in\{0,1\}$,
\begin{equation}
\label{eq:example-investment-action-difference}
Q_i^\Pi(0,(1,a_j))-Q_i^\Pi(0,(0,a_j))
=
\frac{
1-2\pi_j(1\mid0)-2\pi_i(1\mid1)
}{
6+4\pi_1(1\mid0)+4\pi_2(1\mid0)
}.
\end{equation}
At state $1$, changing action $0$ to action $1$ has gain $-1/2$, because the
action does not affect the transition law.

\medskip
\noindent
\emph{MPE.}
The one-deviation principle implies that every MPE satisfies
$\pi_i(1\mid1)=0$ for both players. The denominator in
\eqref{eq:example-investment-action-difference} is positive, so at state $0$
\[
\pi_i(1\mid0)\in
\begin{cases}
\{1\},&\pi_j(1\mid0)<\frac12,\\
\{0\},&\pi_j(1\mid0)>\frac12,\\
[0,1],&\pi_j(1\mid0)=\frac12.
\end{cases}
\]
Applying these conditions to both players yields exactly the state-$0$
marginal profiles $(1,0)$, $(0,1)$, and $(1/2,1/2)$. Conversely, each of these
profiles, together with $\pi_i(1\mid1)=0$ for both players, satisfies the
statewise one-deviation conditions and is therefore an MPE.

\medskip
\noindent
\emph{MCCE.}
Let $\sigma$ be a stationary statewise joint-action law, and let
$\bar\Pi^\sigma$ be the stationary product profile with the same statewise
marginals, so
$\bar\pi_i^\sigma(1\mid s)=\sigma_i^s(1)$. Under $\sigma$, player $i$'s
expected current payoff at state $s$ is
$s-\frac12\sigma_i^s(1)$. Moreover,
$P_\sigma(1\mid0)
=[\sigma_1^0(1)+\sigma_2^0(1)]/3$ and
$P_\sigma(1\mid1)=3/4$. These are exactly the expected current payoffs and
transition probabilities induced by $\bar\Pi^\sigma$. Hence
$V_i^\sigma=V_i^{\bar\Pi^\sigma}$.

Now fix a deterministic stationary deviation
$\beta_i:\cS\to\{0,1\}$ and let $j\ne i$. Against either $\sigma_{-i}$ or
$\bar\Pi_{-i}^\sigma$, the deviating player's current payoff at state $s$ is
$s-\frac12\beta_i(s)$. The induced transition probabilities are also the
same: at state $0$ they equal
$[\beta_i(0)+\sigma_j^0(1)]/3$, and at state $1$ they equal $3/4$. Therefore
$V_i^{\beta_i,\sigma_{-i}}
=V_i^{\beta_i,\bar\Pi_{-i}^\sigma}$ for every deterministic stationary
deviation $\beta_i$.

It follows that, for every player $i$, initial state $s$, and deterministic
stationary deviation $\beta_i$,
\[
V_i^{\beta_i,\sigma_{-i}}(s)-V_i^\sigma(s)
=
V_i^{\beta_i,\bar\Pi_{-i}^\sigma}(s)
-
V_i^{\bar\Pi^\sigma}(s).
\]
Thus the MCCE inequalities for $\sigma$ are exactly the MPE inequalities for
$\bar\Pi^\sigma$. A stationary joint-action law is therefore an MCCE if and
only if the product profile with the same statewise marginals is an MPE.
Consequently, its state-$0$ marginals are $(1,0)$, $(0,1)$, or
$(1/2,1/2)$, and its state-$1$ marginals are $(0,0)$.

\medskip
\noindent
\emph{MBCCE.}
Let $\rho_0$ be the initial state law, with full support, and write $d_\Pi(s)$ for the normalized discounted visit frequency of state $s$. Every profile has $d_\Pi(1)\ge(1-\gamma)\rho_0(1)>0$. At state~1 the gain from deviating to zero is $\pi_i(1\mid1)/2\ge0$. Its expectation can be nonpositive only if $\pi_i(1\mid1)=0$ almost surely under the policy prior. Hence a profile is described by its two low-state investment probabilities $x=\pi_1(1\mid0)$ and $y=\pi_2(1\mid0)$.

Write $\nu_0$ for the conditional law of $(x,y)$ at state~0. The discounted balance equation for state~0 gives
$d_\Pi(0)=3(1+\rho_0(0))/(6+4(x+y))$.
It is positive throughout the unit square. Every probability law $\nu_0$ on that square can therefore arise as the state-$0$ posterior under the visitation rule. To see this, give $\Pi$ weight proportional to $1/d_\Pi(0)$ under $\nu_0$, and normalize. The reciprocal is bounded because $x+y\le2$. Conditioning on state~0 multiplies these weights by $d_\Pi(0)$, recovering $\nu_0$. It remains to determine which such posteriors obey the incentive tests.

Put $D=6+4(x+y)$. By \eqref{eq:example-investment-action-difference}, the gain from action one relative to action zero for player~1 is $(1-2y)/D$. Thus her gains from committing to one or zero relative to her prescribed probability $x$ are $(1-x)(1-2y)/D$ and $-x(1-2y)/D$, respectively. Player~2 has the same expressions with $x,y$ interchanged. With the finite measure $\tau(d\Pi)=\nu_0(d\Pi)/D$, the four obedience conditions are
\begin{equation}
\label{eq:example-MBCCE-posterior}
\begin{aligned}
 \int(1-x)(1-2y)\,d\tau&\le0, &\int -x(1-2y)\,d\tau&\le0,\\
 \int(1-y)(1-2x)\,d\tau&\le0, &\int -y(1-2x)\,d\tau&\le0 .
\end{aligned}
\end{equation}
Up to a positive common factor, these are precisely the static CCE inequalities for the joint action distribution obtained by normalizing $\tau$ to a probability, sampling $(x,y)$ from it, and drawing two independent Bernoulli actions at those probabilities. In particular, the equal-weight two-profile example in the main text has constant $D$ and induces the static CCE assigning probability one third to each of $(0,0)$, $(1,0)$ and $(0,1)$. The change of measure in \eqref{eq:example-MBCCE-posterior} accounts for the size of the payoff gain under each continuation. Conditioning the prior on a low-state visit is a separate step determined by $d_\Pi(0)$.

The full projected set in Figure~\ref{fig:anti-coordination-investment-MBCCE1}(b) consists exactly of the pairs
\begin{equation}
\label{eq:example-MBCCE-projection}
 (q_1,q_2)=\left(\int xD\,d\tau,\ \int yD\,d\tau\right),
 \qquad \int D\,d\tau=1 ,
\end{equation}
where $\tau$ ranges over nonnegative measures on the unit square satisfying \eqref{eq:example-MBCCE-posterior}. Necessity follows from the preceding calculation. For sufficiency, set $d\nu_0=D\,d\tau$ and construct the prior from $\nu_0$ as above. This also shows why the region in the figure is the same for every full-support reset distribution.

To find the smallest and largest common investment probabilities, first normalize $\tau$ to a probability. Its mean total investment is at least $2/3$ by \eqref{eq:example-cce-projection}. Weighting this distribution by $D$ gives the posterior $\nu_0$ after normalization. This increases the mean of $x+y$ by $4\operatorname{Var}(x+y)/(6+4\E[x+y])$, with the variance and expectation taken under the normalized $\tau$. Thus $q_1+q_2\ge2/3$, with equality for the two profiles in the main text. For the other extreme, round $x$ and $y$ to conditionally independent Bernoulli outcomes with the same means under $\tau$. This preserves the four bilinear tests in \eqref{eq:example-MBCCE-posterior} and the normalization in \eqref{eq:example-MBCCE-projection}, and can only increase each marginal by increasing its second moment. For the rounded measure, multiply each of the four masses by $D$ to recover its low-state posterior mass $w_{ab}$. The tests against action zero for the two firms then require $w_{10},w_{01}\ge 5w_{11}/7$, since the relevant gains are $-1/10$ and $1/14$. It follows that $w_{11}\le7/17$ and $q_1+q_2\le1+w_{11}-w_{00}\le24/17$. The posterior with $w_{00}=0$, $w_{10}=w_{01}=5/17$, and $w_{11}=7/17$ obeys all four tests and attains equality. Convexity of \eqref{eq:example-MBCCE-projection} gives every symmetric pair between the two endpoints.

To draw \eqref{eq:example-MBCCE-projection}, take a uniform mesh with $n$ subintervals per axis. The feasible mesh laws form a polytope, and maximizing each supporting linear objective in $(q_1,q_2)$ is the finite-support LP. We bound its error for the full square. Given any feasible $\tau$, round each coordinate independently to its two neighboring mesh points, with probabilities chosen so its conditional mean is the original coordinate. This operation preserves the integral of each bilinear function of $(x,y)$ and therefore preserves \eqref{eq:example-MBCCE-posterior} and $\int D\,d\tau=1$. The first marginal in \eqref{eq:example-MBCCE-projection} can increase only through the term $4x^2$. The increase in its conditional second moment is at most $1/(4n^2)$, the largest variance of a two-point draw over a mesh interval. Since $\tau([0,1]^2)\le1/6$, the first marginal increases by at most $1/(6n^2)$. The same bound holds for the second marginal. Every point in the unrestricted region thus has a feasible mesh point within $1/(6n^2)$ in each coordinate, and the mesh region is contained in the unrestricted region. The figure uses $n=48$. Its positive area also follows directly from the feasible, noncollinear pairs $(1,0)$, $(0,1)$ and $(1/3,1/3)$.

For clarity, the sampling convention matters in this example. Under continuing sampling, $d_\Pi(0)=3/(3+4(x+y))$ and $d_\Pi(1)>0$ whenever $x+y>0$. Thus the same argument forces zero high-state investment on every profile with $x+y>0$, up to a null set. If $x=y=0$, the high state is never visited under that profile. Its prescribed high-state action can affect a low-state deviation. For the purpose of describing the marginals permitted by the continuing MBCCE definition, both firms may be assigned high-state action one at this single point. This weakly decreases both gains from switching to one at the low state, leaves their gains from switching to zero unchanged, and has no effect on high-state obedience since the profile has zero high-state frequency. Hence the continuing projection is given by \eqref{eq:example-MBCCE-projection} with the two integrands for a switch to one in \eqref{eq:example-MBCCE-posterior} set to $-1$ at $(0,0)$. The point mass at $(0,0)$ then obeys. Figure~\ref{fig:anti-coordination-investment-MBCCE1}(b) uses the episodic rule throughout.
\end{proof}

\section{Proofs for Section~\ref{sec:AMCR}}

This appendix establishes the two ingredients introduced in Section~\ref{sec:AMCR}. We first solve the frozen-policy Bellman system. We then compare each true deviation gain with the gain computed from the current critic; summing this pointwise comparison proves the regret--tracking principle.

\begin{proof}[Proof of Lemma~\ref{lem:Bellman-tying-unique}]
Fix a player $i$. For a statewise joint-action law $\rho$, define
\[
r_i^\rho(s):=\sum_a\rho^s(a)u_i(s,a),
\qquad
P_\rho(s'\mid s):=\sum_a\rho^s(a)P(s'\mid s,a).
\]
Averaging the first equation in \eqref{eq:canonical-bellman-system} under $\rho^s$ and using the second gives
\begin{equation}
\label{eq:canonical-value-fixed-point}
V_i^\rho=r_i^\rho+\gamma P_\rho V_i^\rho.
\end{equation}
Since $P_\rho$ is stochastic, the right-hand side of \eqref{eq:canonical-value-fixed-point} is a $\gamma$-contraction in the sup norm. Hence
\[
V_i^\rho
=(I-\gamma P_\rho)^{-1}r_i^\rho
=\sum_{k=0}^\infty\gamma^kP_\rho^k r_i^\rho
\]
is the unique solution of \eqref{eq:canonical-value-fixed-point}. The series representation and $0\le r_i^\rho\le1$ imply
$0\le V_i^\rho\le H$ componentwise.

Given $V_i^\rho$, the first equation in \eqref{eq:canonical-bellman-system} uniquely determines $Q_i^\rho$. Averaging that equation under $\rho^s$ recovers \eqref{eq:canonical-value-fixed-point}, so the tying equation also holds. Moreover,
$0\le Q_i^\rho(s,a)\le1+\gamma H=H$. The construction is independent across players, which proves existence, uniqueness, and the componentwise bounds.
\end{proof}

\begin{proof}[Proof of Proposition \ref{prop:AMCR-mpe}]
If $\mathcal A_i(\Pi^\star;s,b_i)>0$ for some $i,s,b_i$, continuity gives a strictly positive lower bound for $\mathcal A_i(\Gamma_t;s,b_i)$ at all sufficiently late dates. Since $\liminf_T N_s(T)/T>0$, the contribution of these visits to AMCR has a strictly positive lower limit, a contradiction. Thus the Bellman value of $\Pi^\star$ satisfies all statewise one-deviation inequalities against $\Pi^\star_{-i}$ for each player $i$. Its value consequently solves each player's optimality equation against those opponents, giving $\GapMPE(\Pi^\star)=0$.
\end{proof}

\begin{proof}[Proof of Theorem~\ref{thm:AMCR-regret-tracking}]
Fix a state $s$, player $i$, visit $t\in T_s(T)$, and deviation $a_i'$. The tying equation in \eqref{eq:canonical-bellman-system} gives
$V_i^{\Pi_t}(s)=\E_{a\sim\Pi_t(\cdot\mid s)}Q_i^{\Pi_t}(s,a)$. Therefore,
\begin{align}
&\mathcal A_i(\Pi_t;s,a_i')
-\left[
 g_{i,t}^{\mathrm{tar}}(a_i')
 -\left\langle\pi_{i,t}(\cdot\mid s),g_{i,t}^{\mathrm{tar}}\right\rangle
 \right]
\notag\\
&\quad=
\E_{a_{-i}\sim\Pi_{-i,t}(\cdot\mid s)}
\left[
 Q_i^{\Pi_t}(s,(a_i',a_{-i}))
 -Q_{i,t}(s,(a_i',a_{-i}))
\right]
\notag\\
&\qquad+
\E_{a\sim\Pi_t(\cdot\mid s)}
\left[Q_{i,t}(s,a)-Q_i^{\Pi_t}(s,a)\right].
\label{eq:AMCR-pointwise-identity}
\end{align}
Each expectation on the right-hand side of \eqref{eq:AMCR-pointwise-identity} is at most $e_t$ in absolute value. Thus
\begin{equation}
\label{eq:appendix-AMCR-pointwise}
\mathcal A_i(\Pi_t;s,a_i')
\le
 g_{i,t}^{\mathrm{tar}}(a_i')
 -\left\langle\pi_{i,t}(\cdot\mid s),g_{i,t}^{\mathrm{tar}}\right\rangle
 +2e_t.
\end{equation}
Sum \eqref{eq:appendix-AMCR-pointwise} over $t\in T_s(T)$, maximize over $a_i'$, and use $(x+y)_+\le x_++y$ for $y\ge0$. By \eqref{eq:target-law-regret},
\[
\left(
\max_{a_i'}\sum_{t\in T_s(T)}
\mathcal A_i(\Pi_t;s,a_i')
\right)_+
\le
R_{i,s}^{\mathrm{tar}}(N_s(T))
+2\sum_{t\in T_s(T)}e_t.
\]
Summing over $(i,s)$ and using that the sets $T_s(T)$ partition $\{1,\ldots,T\}$ gives the upper bound on AMCR. Interchanging the two gains in \eqref{eq:AMCR-pointwise-identity} gives the same bound on aggregate regret divided by $T$. These two inequalities prove \eqref{eq:AMCR-regret-tracking}.
\end{proof}

\section{Proofs for Section \ref{sec:markov-bcce}}

\begin{proof}[Proof of Proposition \ref{prop:empirical-markov-bcce}]
Fix $s$ with $N_s(T)>0$. Conditional on $s$ under $\widehat\mu_T$, the
posterior over policies is
$\widehat\mu_T(d\Pi\mid s)
=N_s(T)^{-1}\sum_{t\in T_s(T)}\delta_{\Gamma_t}(d\Pi)$.
Hence, for every $i$ and $b_i$,
\begin{align*}
&\frac{N_s(T)}{T}
\E_{\substack{\Pi\sim\widehat\mu_T(\cdot\mid s)\\
a\sim\Pi(\cdot\mid s)}}
\left[
Q_i^\Pi(s,(b_i,a_{-i}))-Q_i^\Pi(s,a)
\right]
\\
&\qquad=
\frac1T\sum_{t\in T_s(T)}
\mathcal A_i(\Gamma_t;s,b_i),
\end{align*}
where
$\E_{a\sim\Gamma_t(\cdot\mid s)}
Q_i^{\Gamma_t}(s,a)=V_i^{\Gamma_t}(s)$.
Maximizing over $b_i$, taking positive parts, and summing over players and
states proves the identity in the proposition.

For any weakly convergent subsequence, the right-hand side of
\eqref{eq:AMCR-empirical-mbcce} converges to
\begin{equation*}
 \sum_{s,i}\left(\max_{b_i}\int_{\mathcal C}
 \mathcal A_i(\Pi;s,b_i)\,\mu(\{s\},d\Pi)\right)_+,
\end{equation*}
where $\mu$ is the limit of the empirical laws along that subsequence. Indeed, the state space is finite, the advantage is continuous in $\Pi$ by Lemma~\ref{lem:AMCR-advantage-lipschitz}, and the maximum runs over finitely many actions. This expression vanishes exactly when the limit satisfies \eqref{eq:markov-bcce-obedience} at all states of positive probability. The laws on $\cS\times\mathcal C$ form a weakly compact space. It follows that AMCR vanishes exactly when all weak accumulation points satisfy obedience. For each of the two protocols in the proposition, the set defined by obedience and the state-frequency equations is closed, since the latter equations can be tested against continuous functions of the profile, as shown below. Thus membership of all accumulation points is equivalent to convergence of the distance to this set for any metric inducing weak convergence. It remains to verify the state-frequency equations in the two protocols.

\emph{Continuing play.} Fix $r\in\cS$ and a bounded Lipschitz function $f$ on $\mathcal C$. Conditional on the history $\mathcal H_t$, the probability of the next state $r$ is $P_{p_t}(r\mid s_t)$. Thus
$f(p_t)[\one\{s_{t+1}=r\}-P_{p_t}(r\mid s_t)]$ is a bounded martingale difference. Its time average tends to zero almost surely. For the arrival indicators, the difference between the two sums satisfies
\begin{equation*}
\left|\sum_{t=1}^T f(p_t)\one\{s_{t+1}=r\}
-\sum_{t=1}^T f(p_t)\one\{s_t=r\}\right|
\le 2\|f\|_\infty+\operatorname{Lip}(f)
\sum_{t=1}^{T-1}\|p_{t+1}-p_t\|_1 .
\end{equation*}
Indeed, at each interior date $t$, the state $s_t$ occurs once as the arrival from date $t-1$ and once as the state observed at date $t$. Only the endpoint observations and the change of policy at each paired date remain. The bound divided by $T$ tends to zero by \eqref{eq:mbcce-average-increments}. For any convergent subsequence of $\widehat\mu_T$, continuity of $P_\Pi(r\mid s)$ in $\Pi$ and the martingale law now give
\begin{equation*}
 \int f(\Pi)\,\mu(\{r\},d\Pi)
 =
 \sum_s\int f(\Pi) P_\Pi(r\mid s)\,\mu(\{s\},d\Pi).
\end{equation*}
Take a countable family of Lipschitz functions uniformly dense in the continuous functions on $\mathcal C$, and intersect their probability-one events over all states. The resulting equalities identify the measures on the two sides for every empirical limit on the same event. Let $\nu$ be the policy marginal of such a limit $\mu$, and write $d_\Pi(s)$ for the conditional probability of state $s$ given the profile, so $\mu(\{s\},d\Pi)=d_\Pi(s)\nu(d\Pi)$. For $\nu$-almost every $\Pi$ these conditional probabilities form a probability vector and satisfy $d_\Pi(r)=\sum_s d_\Pi(s)P_\Pi(r\mid s)$ for every $r$. This proves \eqref{eq:mbcce-continuing-states}.

\emph{Episodic play.} Write $p_e$ for the profile chosen before episode $e$, $Y_e$ for its number of stages, and $N_{e,r}$ for its number of visits to $r$. Given the past and $p_e$, $Y_e$ has the geometric law on the positive integers with mean $H$, and
$\E[N_{e,r}\mid\text{past},p_e]=H d_{p_e}(r)$ by \eqref{eq:mbcce-episodic-states}. Thus for each bounded Lipschitz $f$ the episode variables
$f(p_e)[N_{e,r}-Y_ed_{p_e}(r)]$ have conditional mean zero and uniformly bounded second moments. Their partial sums divided by the number of completed episodes tend to zero by the martingale strong law. Since $\sum_{e\le E}Y_e/E\to H$ almost surely, the same holds after division by the total number of stages. Also $Y_{E+1}/E\to0$ almost surely by the geometric tail and Borel--Cantelli. An unfinished episode changes the normalized sums by at most a constant times its length divided by the number of preceding stages, which tends to zero. For every weak limit $\mu$ and its policy marginal $\nu$ we obtain
\begin{equation*}
 \int f(\Pi)\,\mu(\{r\},d\Pi)=\int f(\Pi) d_\Pi(r)\,\nu(d\Pi).
\end{equation*}
Here $d_\Pi$ is the continuous function of $\Pi$ in \eqref{eq:mbcce-episodic-states}. As in continuing play, a countable family of tests proves the required equality of measures simultaneously at every state and for every weak limit.
\end{proof}

\section{Proofs for Section~\ref{sec:algorithm}}
\label{app:algorithm}

\subsection{Critic recursion and frozen-policy regularity}

For the critic analysis, freeze a target profile $\Pi$, let
$X=(Q,V)$ and define
\begin{align}
G_\eps^Q(\Pi,X)(i,s,a)
&:=u_i(s,a)+\gamma\sum_{s'}P(s'\mid s,a)V_i(s'),
\label{eq:critic-drift-q}\\
G_\eps^V(\Pi,X)(i,s)
&:=\E_{a_{-i}\sim p_{-i}^\eps(\Pi)(\cdot\mid s)}
\left[
\left\langle
\pi_i(\cdot\mid s),Q_i(s,(\cdot,a_{-i}))
\right\rangle
\right],
\label{eq:critic-drift-v}
\end{align}
where
$p_i^\eps(\Pi)(\cdot\mid s)
:=(1-\eps)\pi_i(\cdot\mid s)+\eps\mathsf U_i(\cdot\mid s)$ and
$p_{-i}^\eps(\Pi):=\bigotimes_{j\ne i}p_j^\eps(\Pi)$. Write
$G_\eps:=(G_\eps^Q,G_\eps^V)$. At date $t$, $G_{\eps_t}^Q(\Pi_t,X_t)$ gives the mean target for the updated action value conditional on the history and realized joint action. For the state value, $G_{\eps_t}^V(\Pi_t,X_t)$ gives the mean target before the action draw, since the opponents use their behavior policies.

Let $X_t:=(Q_t,V_t)\in\mathcal X:=[0,H]^m$, where $m$ is the number of critic
coordinates. At zero mixing the fixed-point equations of $G_0(\Pi,\cdot)$ are the Bellman equations \eqref{eq:canonical-bellman-system}. Write $\Lambda(\Pi)=(Q^\Pi,V^\Pi)$ for this evaluation. Proposition~\ref{prop:critic-regularity} shows that the map contracts in a weighted sup norm. When the actors use exploration at date $t$, its effect on the mean update is bounded by a constant times $\eps_t$. As exploration vanishes, this allows us to track $\Lambda(\Pi_t)$ itself.

The stochastic-approximation representation used in the proof is
\begin{equation*}
X_{t+1}
=X_t+\mathsf A_t\bigl(G_{\eps_t}(\Pi_t,X_t)-X_t+\xi_{t+1}\bigr).
\end{equation*}
The diagonal matrix $\mathsf A_t$ has entry
$\alpha_{s_t,a_t,n_t(s_t,a_t)}$ on each coordinate $Q_i(s_t,a_t)$,
entry $\eta_{s_t,n_t(s_t)}$ on each coordinate $V_i(s_t)$, and zero elsewhere. The nonzero innovation coordinates are
\begin{align*}
\xi_{t+1}^Q(i,s_t,a_t)
&:=\gamma\left(
V_{i,t}(s_{t+1})-
\sum_rP(r\mid s_t,a_t)V_{i,t}(r)
\right),\\
\xi_{t+1}^V(i,s_t)
&:=\left\langle\pi_{i,t}(\cdot\mid s_t),
Q_{i,t}(s_t,(\cdot,a_{-i,t}))\right\rangle
-G_{\eps_t}^V(\Pi_t,X_t)(i,s_t).
\end{align*}
and all remaining coordinates are zero.

For $\omega>0$, write
\[
\|(Q,V)\|_\omega
:=\max\{\|Q\|_\infty,\omega\|V\|_\infty\}.
\]
We use $\omega=\sqrt\gamma$ for the critic estimates below.

\begin{proposition}[Frozen critic fixed point]
\label{prop:critic-regularity}
With $\omega=\sqrt\gamma$, $G_\eps$ maps $\mathcal C\times\mathcal X$ into $\mathcal X$ and is jointly Lipschitz there, uniformly in $\eps\in[0,1]$. For each $\Pi$, the map $G_\eps(\Pi,\cdot)$ is a $\sqrt\gamma$-contraction in $\|\cdot\|_\omega$ on $\R^m$. The fixed point of $G_0(\Pi,\cdot)$ is $\Lambda(\Pi)$, which is Lipschitz in $\Pi$. Moreover, for $X\in\mathcal X$,
\begin{equation}
\label{eq:critic-mixing-bound}
 \|G_\eps(\Pi,X)-G_0(\Pi,X)\|_\omega
 \le 2\sqrt\gamma H(N-1)\eps .
\end{equation}
\end{proposition}

\begin{proof}
For $X=(Q,V)\in\mathcal X$, each action-value update target lies between zero and $1+\gamma H=H$. The state-value target averages entries of $Q$ and also lies in $[0,H]$.

Only the opponents' distributions in \eqref{eq:critic-drift-v} change with exploration. Each of them differs from the target by at most $2\eps$ in $\ell_1$. Expanding the difference of the two product distributions one factor at a time therefore gives an $\ell_1$ difference at most $2(N-1)\eps$. Since $\|Q\|_\infty\le H$, this proves \eqref{eq:critic-mixing-bound}. With the mixing probability held fixed, the same expansion shows that changing $\Pi$ to $\Pi'$ changes the joint distribution used in the state-value update by at most $\|\Pi-\Pi'\|_1$.

For arbitrary $X=(Q,V)$ and $X'=(Q',V')$ in $\R^m$, the two formulas for $G_\eps$ give
\begin{align*}
 \|G_\eps^Q(\Pi,X)-G_\eps^Q(\Pi,X')\|_\infty
 &\le\gamma\|V-V'\|_\infty
 \le\sqrt\gamma\|X-X'\|_\omega,\\
 \omega\|G_\eps^V(\Pi,X)-G_\eps^V(\Pi,X')\|_\infty
 &\le\omega\|Q-Q'\|_\infty
 \le\sqrt\gamma\|X-X'\|_\omega .
\end{align*}
These inequalities prove contraction in the tables and, with the preceding bound for the policy, joint Lipschitz continuity on $\mathcal C\times\mathcal X$. The fixed-point equations at $\eps=0$ are \eqref{eq:canonical-bellman-system}. Apply the contraction and the policy bound to the fixed-point identities at $\Pi$ and $\Pi'$. This gives
$(1-\sqrt\gamma)\|\Lambda(\Pi)-\Lambda(\Pi')\|_\omega\le \omega H\|\Pi-\Pi'\|_1$, proving the last regularity claim.
\end{proof}

\begin{lemma}[Regularity of current-policy advantages]
\label{lem:AMCR-advantage-lipschitz}
There is $L_{\mathcal A}<\infty$ such that
\[
|\mathcal A_i(\Pi;s,a_i')-\mathcal A_i(\Pi';s,a_i')|
\le L_{\mathcal A}\|\Pi-\Pi'\|_1
\]
for every $i,s,a_i'$ and $\Pi,\Pi'\in\mathcal C$.
\end{lemma}

\begin{proof}
By Proposition~\ref{prop:critic-regularity} at $\eps=0$ and norm
equivalence, there is $L_{\mathrm{ev}}<\infty$ such that
\[
\max\left\{
\|Q^\Pi-Q^{\Pi'}\|_\infty,
\|V^\Pi-V^{\Pi'}\|_\infty
\right\}
\le
L_{\mathrm{ev}}\|\Pi-\Pi'\|_1.
\]
Fix $(i,s,a_i')$. Adding and subtracting
$
\sum_{a_{-i}}
\Pi_{-i}(a_{-i}\mid s)
Q_i^{\Pi'}(s,(a_i',a_{-i}))
$ 
gives
\[
\begin{aligned}
&
|\mathcal A_i(\Pi;s,a_i')
-\mathcal A_i(\Pi';s,a_i')|
\\
&\quad\le
\|Q^\Pi-Q^{\Pi'}\|_\infty
+
H\left\|
\Pi_{-i}(\cdot\mid s)
-\Pi_{-i}'(\cdot\mid s)
\right\|_1
+
\|V^\Pi-V^{\Pi'}\|_\infty.
\end{aligned}
\]
The product-measure telescoping inequality implies
\[
\left\|
\Pi_{-i}(\cdot\mid s)
-\Pi_{-i}'(\cdot\mid s)
\right\|_1
\le
\sum_{j\ne i}
\left\|
\pi_j(\cdot\mid s)
-\pi_j'(\cdot\mid s)
\right\|_1
\le
\|\Pi-\Pi'\|_1.
\]
Hence the claim holds with
$L_{\mathcal A}=2L_{\mathrm{ev}}+H$.
\end{proof}

\subsection{Admissible actor families}

\begin{proof}[Proof of Proposition~\ref{prop:standard-admissible-actors}]
Fix a finite action set $A$. The cited regret results are stated for losses;
we apply them to the negatives of the gains. Let $\|\cdot\|$ be a norm on
$\R^A$, let $\|\cdot\|_*$ be its dual, and choose $K_A<\infty$ such that
$\|x\|_1\le K_A\|x\|$. Since $g\in[0,H]^A$,
$G_A:=\max_{g\in[0,H]^A}\|g\|_*<\infty$.

First consider FTRL. Let $h:\Delta(A)\to\R$ be lower
semicontinuous and $\mu$-strongly convex with respect to $\|\cdot\|$, and
suppose its oscillation
$\Omega_h:=\sup_{\pi}h(\pi)-\inf_{\pi}h(\pi)$ is finite. Subtracting the
constant $\inf h$, assume $\inf h=0$. Define
\[
Q_h(y)
:=
\arg\max_{\pi\in\Delta(A)}
\{\langle y,\pi\rangle-h(\pi)\}.
\]
Compactness and lower semicontinuity give existence, and strong convexity
gives uniqueness. The two optimality inequalities imply
\begin{equation}
\label{eq:choice-map-lipschitz}
\|Q_h(y)-Q_h(y')\|
\le
\frac1\mu\|y-y'\|_*.
\end{equation}
Indeed,
$\mu\|Q_h(y)-Q_h(y')\|^2
\le\langle y-y',Q_h(y)-Q_h(y')\rangle$.

Let $S_{n-1}:=\sum_{q<n}g_q$ and set
$\pi_n=Q_h(\beta_nS_{n-1})$. For the constrained loss problem on
$\Delta(A)$, let $\ell_q(\pi):=-\langle g_q,\pi\rangle$ and define the
centered incremental regularizers by $r_0=h/\beta_1$ and
$r_q=(\beta_{q+1}^{-1}-\beta_q^{-1})h$ for $q\ge1$, with the simplex
constraint encoded by its indicator. Their cumulative sum before round $n$
is $h/\beta_n$. Because $(\beta_n)$ is nonincreasing and $h\ge0$, every
$r_q$ is nonnegative. The resulting adaptive FTRL update is exactly
$\pi_n$. Moreover, $h/\beta_q$ is $\mu/\beta_q$-strongly convex, so the
dual-norm term in Theorem~1 in \cite{McMahan2017Adaptive} is
$\beta_q\|g_q\|_*^2/\mu$. The theorem therefore gives, for every $a\in A$,
\[
\sum_{q=1}^n
\bigl[g_q(a)-\langle\pi_q,g_q\rangle\bigr]
\le
\frac{h(e_a)}{\beta_n}
+
\frac1{2\mu}\sum_{q=1}^n\beta_q\|g_q\|_*^2.
\]
This is the simplex specialization of the regularized dual-averaging update
in Algorithm~1 in \cite{Xiao2010DualAveraging}. Hence
\begin{equation}
\label{eq:ftrl-admissible-regret}
\max_{a\in A}\sum_{q=1}^n
\bigl[g_q(a)-\langle\pi_q,g_q\rangle\bigr]
\le
\frac{\Omega_h}{\beta_n}
+
\frac{G_A^2}{2\mu}\sum_{q=1}^n\beta_q.
\end{equation}
Moreover, $n(\beta_n-\beta_{n+1})\le c\beta_n$ and
\eqref{eq:choice-map-lipschitz} give
\begin{align*}
\|\pi_{n+1}-\pi_n\|_1
&\le
\frac{K_A}{\mu}
\|\beta_{n+1}S_n-\beta_nS_{n-1}\|_*\\
&\le
\frac{K_A(1+c)G_A}{\mu}\beta_n.
\end{align*}
Thus FTRL is admissible.

A bounded predictable forecast may be inserted by taking
$\pi_n=Q_h(\beta_n(S_{n-1}+m_n))$, where
$\|m_n\|_*\le M_A$. Let
$\bar\pi_n:=Q_h(\beta_nS_{n-1})$. By
\eqref{eq:choice-map-lipschitz},
$\|\pi_n-\bar\pi_n\|\le\beta_nM_A/\mu$. Comparing the regret of
$(\pi_n)$ with that of $(\bar\pi_n)$ in
\eqref{eq:ftrl-admissible-regret} gives
\[
\max_{a\in A}\sum_{q=1}^n
\bigl[g_q(a)-\langle\pi_q,g_q\rangle\bigr]
\le
\frac{\Omega_h}{\beta_n}
+
\left(
\frac{G_A^2}{2\mu}+\frac{G_AM_A}{\mu}
\right)
\sum_{q=1}^n\beta_q.
\]
The same Lipschitz estimate yields
\[
\|\pi_{n+1}-\pi_n\|_1
\le
\frac{K_A[(1+c)G_A+2M_A]}{\mu}\beta_n.
\]
Hence the bounded-forecast FTRL actor is admissible.

We next consider online mirror ascent. Let $h$ be differentiable and
$\mu$-strongly convex, write $D_h$ for its Bregman divergence, and suppose
\[
B_h
:=
\max_{a\in A}\sup_{\pi\in\Delta(A)}D_h(e_a,\pi)
<\infty.
\]
Starting from any $z_1\in\Delta(A)$, let $(m_n)$ be predictable with
$\|m_n\|_*\le M_A$, and consider the optimistic mirror update
\begin{align*}
\pi_n
&\in
\arg\max_{\pi\in\Delta(A)}
\{\beta_n\langle\pi,m_n\rangle-D_h(\pi,z_n)\},\\
z_{n+1}
&\in
\arg\max_{z\in\Delta(A)}
\{\beta_n\langle z,g_n\rangle-D_h(z,z_n)\}.
\end{align*}
The ordinary mirror-ascent update is obtained by setting $m_n=0$, in which
case $\pi_n=z_n$. The one-step inequality in the proof of Lemma~3 in
\cite{RakhlinSridharan2013} is roundwise and therefore applies with the
learning rate $\beta_n$:
\[
\langle e_a-\pi_n,g_n\rangle
\le
\frac{D_h(e_a,z_n)-D_h(e_a,z_{n+1})}{\beta_n}
+
\frac{\beta_n}{2\mu}\|g_n-m_n\|_*^2.
\]
For $d_q:=D_h(e_a,z_q)$, monotonicity of $1/\beta_q$ gives
\[
\sum_{q=1}^n\frac{d_q-d_{q+1}}{\beta_q}
=
\frac{d_1}{\beta_1}
+
\sum_{q=2}^n
\left(\frac1{\beta_q}-\frac1{\beta_{q-1}}\right)d_q
-
\frac{d_{n+1}}{\beta_n}
\le
\frac{B_h}{\beta_n}.
\]
Summing the one-step inequality therefore gives
\[
\max_{a\in A}\sum_{q=1}^n
\bigl[g_q(a)-\langle\pi_q,g_q\rangle\bigr]
\le
\frac{B_h}{\beta_n}
+
\frac{(G_A+M_A)^2}{2\mu}
\sum_{q=1}^n\beta_q.
\]
Strong convexity of the two proximal problems also gives
$\|\pi_n-z_n\|\le\beta_nM_A/\mu$ and
$\|z_{n+1}-z_n\|\le\beta_nG_A/\mu$. Hence
\[
\|\pi_{n+1}-\pi_n\|_1
\le
\frac{K_A(G_A+2M_A)}{\mu}\beta_n.
\]
Thus optimistic mirror ascent, and ordinary mirror ascent as the case
$M_A=0$, are admissible.

For $h(\pi)=\sum_a\pi(a)\log\pi(a)$, FTRL is Hedge and
$\Omega_h=\log|A|$. For $h(\pi)=\|\pi\|_2^2/2$, mirror ascent is Euclidean
projected gradient ascent and $B_h\le1$, while FTRL is Euclidean
dual averaging. The bounded-forecast versions give optimistic Hedge and the
Euclidean optimistic mirror update, often called optimistic projected gradient.
\end{proof}

For the actor $\mathfrak A_{i,s}$ selected in Algorithm~\ref{alg:bao-ac},
fix admissibility constants
$D_{i,s}^{\mathrm A}$, $C_{i,s}^{\mathrm A}$, and
$L_{i,s}^{\mathrm A}$.

\subsection{Moving-target tracking}

We first show how long a block must be so each coordinate receives enough updates. The same blocks will serve for both exploration schedules. We use pointwise bounds on block movement for the slower schedule and sum the actual movement over blocks to obtain the finite-time bound for the faster schedule.

\begin{lemma}[One-step movement of admissible actors]
\label{lem:actor-one-step}
There is $L_\Pi<\infty$ such that
\begin{equation}
\label{eq:appendix-one-step-actor}
\|\Pi_{t+1}-\Pi_t\|_1
\le L_\Pi\beta_{s_t,n_t(s_t)}
\qquad\forall t.
\end{equation}
\end{lemma}

\begin{proof}
At stage $t$, only the player blocks at state $s_t$ receive a new gain. By
\eqref{eq:actor-movement-interface}, block $(i,s_t)$ moves by at most
$L_{i,s_t}^{\mathrm A}\beta_{s_t,n_t(s_t)}$. Summing over players proves
\eqref{eq:appendix-one-step-actor} with
$L_\Pi:=\sum_i\max_sL_{i,s}^{\mathrm A}$.
\end{proof}

The next lemma converts a uniform probability of observing an event within a
bounded number of stages into an exponential lower-tail bound for the number
of successful blocks.

\begin{lemma}[Concentration from uniform block hitting]
\label{lem:uniform-block-hitting}
Let $(\mathcal F_t)$ be a filtration and let
$E_t\in\mathcal F_{t+1}$. Suppose there are $L\in\mathbb N$ and
$\theta>0$ such that
\begin{equation}
\label{eq:uniform-block-hitting}
\bbP\left(
\bigcup_{r=t}^{t+L-1}E_r
\,\middle|\,
\mathcal F_t
\right)
\ge\theta
\qquad\forall t\ge1.
\end{equation}
For deterministic integers $t,B\ge1$, define
\[
Y_r
:=
\one\left\{
\bigcup_{u=t+(r-1)L}^{t+rL-1}E_u
\right\},
\qquad r=1,\ldots,B.
\]
Then
\begin{equation}
\label{eq:block-hit-azuma}
\bbP\left(
\sum_{r=1}^B Y_r<\frac{\theta B}{2}
\,\middle|\,
\mathcal F_t
\right)
\le
\exp\left(-\frac{\theta^2B}{8}\right)
\qquad\text{almost surely}.
\end{equation}
\end{lemma}

\begin{proof}
For each $r=1,\ldots,B$, $Y_r$ is
$\mathcal F_{t+rL}$-measurable. Moreover,
\eqref{eq:uniform-block-hitting} applied at time $t+(r-1)L$ gives
$
\E[Y_r\mid\mathcal F_{t+(r-1)L}]\ge\theta.
$ 
Thus
\[
Y_r-\E[Y_r\mid\mathcal F_{t+(r-1)L}],
\qquad r=1,\ldots,B,
\]
is a martingale-difference sequence with respect to
$(\mathcal F_{t+rL})_{r=0}^B$, with increments bounded in absolute value
by one. If $\sum_{r=1}^B Y_r<\theta B/2$, then
\[
\sum_{r=1}^B
\left\{
Y_r-\E[Y_r\mid\mathcal F_{t+(r-1)L}]
\right\}
<-\frac{\theta B}{2}.
\]
Conditional Azuma--Hoeffding therefore gives
\eqref{eq:block-hit-azuma}.
\end{proof}

For the coverage estimates below, write $K=N(d+1)$. For a critic coordinate $j$, let $\chi_t(j)$ denote its actual step at stage $t$. On the visited pair it is $\alpha_0 n_t(s_t,a_t)^{-b}$ for each action value, and at the visited state it is $\eta_0 n_t(s_t)^{-b}$ for each state value. It is zero on every other coordinate. Thus $\mathsf A_t=\operatorname{diag}(\chi_t(j))_j$. In the half-step formulas below, terms for inactive coordinates are also zero. Let
$w_t=\eps_{2t}^{K}$ and take blocks of length $m_t=\lceil t^b/w_t\rceil$.
For the inverse-log schedule $m_t=t^{b+o(1)}$. For the polynomial schedule write $z=Kr$, so $w_t$ is of order $t^{-z}$ and $m_t$ of order $t^{b+z}$. In Theorem~\ref{thm:finite-time-AMCR} choose $b$ sufficiently close to $1/2$ that $b+z<c-r$, $b(1-2z)/2>r$, and $b>z$. These strict inequalities hold at $b=1/2$ for $r<1/(2K+4)$. In both cases $m_t=o(t)$, so $t+m_t\le2t$ for large $t$.

\begin{lemma}[Quantitative critic blocks]
\label{lem:uniform-critic-blocks}
There exist positive constants $a_0,a_1,C_0,\delta$ such that the following holds for all sufficiently large $t$. Write $\mathcal V_t$ for the event $n_{t-1}(s,a)\ge a_0t w_t$ for every pair, and $\mathcal E_t$ for the event that $\mathcal V_t$ holds and every pair occurs at least $a_0m_t w_t$ times on $[t,t+m_t)$. Then
\begin{equation}
\label{eq:critic-block-probabilities}
\bbP(\mathcal E_t^c)\le C_0\exp(-a_1tw_t^2)+C_0\exp(-a_1m_tw_t^2).
\end{equation}
On $\mathcal E_t$, $\sum_{u=t}^{t+m_t-1}\chi_u(j)\ge\delta$ for every $j$. On $\mathcal V_t\in\mathcal H_t$, the largest critic step is $O((tw_t)^{-b})$ and total target movement on the block is $O(m_t(tw_t)^{-c})$.
\end{lemma}

\begin{proof}
Put $L=d+1$, and let $\underline P_{\mathrm U}$ be the smallest positive entry of $P^{\mathrm U}$. Every history with current date $u\le 2t$ gives joint-action mass at least $\eps_{2t}^N/|\cA|$ to any specified action. It therefore gives each outgoing edge of the support graph transition probability at least $\eps_{2t}^N\underline P_{\mathrm U}$. Follow a shortest path from the current state to a prescribed state and choose the specified action there. In any $L$-date window ending by $2t$ the conditional probability of hitting the specified pair is at least $\underline P_{\mathrm U}^{d} w_t/|\cA|$. When $d=0$ this is simply the action probability bound.

Apply Lemma~\ref{lem:uniform-block-hitting} to disjoint complete $L$-date windows on $[1,t)$, and again on $[t,t+m_t)$. Its proof requires the hitting bound only within the windows used. Each successful window contributes at least one visit. A union bound over the finite set of pairs gives \eqref{eq:critic-block-probabilities}, after choosing $a_0$ small enough to account for integer parts. At every activation the local counter is at most $2t$, so each step is at least $\min(\alpha_0,\eta_0)(2t)^{-b}$. Every state coordinate is activated at least as often as any pair at that state. Since $m_tw_t\ge t^b$, all cumulative steps exceed a constant $\delta>0$ on $\mathcal E_t$. On $\mathcal V_t$ every active counter on the block is at least $a_0tw_t$. This event uses only pre-block observations. Applying Lemma~\ref{lem:actor-one-step} proves the movement bound.

The right side of \eqref{eq:critic-block-probabilities} is summable over all integer $t$. For inverse-log exploration its two exponential arguments are of orders $t^{1-o(1)}$ and $t^{b-o(1)}$. For polynomial exploration they are of orders $t^{1-2z}$ and $t^{b-z}$, both positive powers. Thus all sufficiently late blocks are good almost surely, and eventually
\begin{equation}
\label{eq:critic-global-frequency}
n_{t-1}(s,a)\ge a_0tw_t \qquad\text{for all }(s,a).
\end{equation}

More generally the proof uses only an $L$-date hitting probability bounded below by a constant times the minimum exploration probability over those dates raised to $K$. Any uniform bound of this form with a smaller power can replace $K$ throughout. In the investment game of Example~\ref{ex:stochastic-investment-MBCCE}, the chance of moving from state~0 to state~1 is at least $\eps_{2t}/3$, and from state~1 to state~0 it is $1/4$. Drawing any prescribed joint action at the desired state costs at most two additional powers of exploration. This gives power three in a window of two dates.
\end{proof}

Because $\alpha_{s,a,n},\eta_{s,n}\in(0,1]$, each critic update is a convex
combination of its previous value and a target in $[0,H]$. Indeed,
$0\le u_i(s,a)+\gamma V_{i,t}(s')\le1+\gamma H=H$ and
$0\le
\langle\pi_{i,t}(\cdot\mid s),
Q_{i,t}(s,(\cdot,a_{-i,t}))\rangle
\le H$.
The initialization therefore implies $X_t\in\mathcal X=[0,H]^m$ for every
$t$.

For $t\le k\le t+m_t$, define
$M_{t,k}:=\sum_{u=t}^{k-1}\mathsf A_u\xi_{u+1}$ and
$\mathcal N_t:=\max_{t\le k\le t+m_t}\|M_{t,k}\|_\omega$.

\begin{lemma}[Blockwise critic noise]
\label{lem:blockwise-critic-noise}
With $B_t=\sqrt{m_t}(tw_t)^{-b}$, there are constants $C_1,c_1>0$ such that
\begin{align}
\E[\one\{\mathcal V_t\}\mathcal N_t]&\le C_1 B_t,
\label{eq:blockwise-critic-noise-expectation}\\
\bbP(\mathcal V_t\cap\{\mathcal N_t>x\})&\le C_1\exp(-c_1 x^2/B_t^2)\quad(x>0).
\label{eq:blockwise-critic-noise-tail}
\end{align}
For any collection of disjoint blocks with starts in $[U,2U)$ and $w_t$ of order $t^{-z}$, one also has
\begin{equation}
\sum_{t\ \mathrm{in\ these\ starts}}
\E[\one\{\mathcal V_t\}\mathcal N_t^2]\le C U^{-(2b-1)(1-z)}.
\label{eq:critic-block-noise-group}
\end{equation}
\end{lemma}

\begin{proof}
The action both selects the active coordinates and generates the $V$-critic innovation, while the subsequent state transition generates the $Q$-critic innovation. We therefore separate these two sources of randomness by splitting every stage into two half-steps. For
$u\ge1$, set $\mathcal G_{2u}:=\mathcal H_u$,
$\mathcal G_{2u+1}:=\sigma(\mathcal H_u,a_u)$, and
$\mathcal G_{2u+2}:=\mathcal H_{u+1}$. Let $\Delta_{2u+1}$ have zero
$Q$ block and, for every $(i,s)$, set
\begin{align*}
\Delta_{2u+1}^{V}(i,s)
:=\one\{s_u=s\}\eta_{s,n_u(s)}
\bigl(&\langle\pi_{i,u}(\cdot\mid s),
Q_{i,u}(s,(\cdot,a_{-i,u}))\rangle
-G_{\eps_u}^V(\Pi_u,X_u)(i,s)\bigr).
\end{align*}
Let $\Delta_{2u+2}$ have zero $V$ block and, for every $(i,s,a)$, set
\begin{align*}
\Delta_{2u+2}^{Q}(i,s,a)
:=\one\{(s_u,a_u)=(s,a)\}\alpha_{s,a,n_u(s,a)}\gamma
\left(
V_{i,u}(s_{u+1})-
\sum_{s'}P(s'\mid s,a)V_{i,u}(s')
\right).
\end{align*}
Equations~\eqref{eq:behavior-sampling}, \eqref{eq:critic-drift-v}, and the
transition law imply
$\Delta_{2u+1}+\Delta_{2u+2}=\mathsf A_u\xi_{u+1}$,
$\E[\Delta_{2u+1}\mid\mathcal G_{2u}]=0$, and
$\E[\Delta_{2u+2}\mid\mathcal G_{2u+1}]=0$.

For fixed $t$, set $S_{t,2t}:=0$ and
$S_{t,\ell}:=\sum_{j=2t}^{\ell-1}\one\{\mathcal V_t\}\Delta_{j+1}$ for
$2t<\ell\le2(t+m_t)$. Since $\mathcal V_t\in\mathcal G_{2t}$,
$(S_{t,\ell})_{\ell=2t}^{2(t+m_t)}$ is a finite-dimensional martingale with
respect to $(\mathcal G_\ell)_{\ell=2t}^{2(t+m_t)}$. At every complete
stage, $S_{t,2k}=\one\{\mathcal V_t\}M_{t,k}$ for
$t\le k\le t+m_t$.

On $\mathcal V_t$, the half-step increments have norm at most $C(tw_t)^{-b}$ because all innovations are bounded. Their squared bounds sum to at most $C B_t^2$. Applying maximal Azuma to each coordinate and taking a finite union gives \eqref{eq:blockwise-critic-noise-tail}. Integration gives \eqref{eq:blockwise-critic-noise-expectation}.

For \eqref{eq:critic-block-noise-group}, apply Doob's $L^2$ inequality to each coordinate of this martingale. Orthogonality bounds each expected squared maximum by four times the sum of expected squared increments. Pathwise on blocks with $\mathcal V_t$, every local clock starts at least at $aU^{1-z}$ for a constant $a>0$. Each coordinate takes its step of order $n^{-b}$ at most once at each local index $n$ across all these disjoint blocks. Bounded innovations therefore bound the sum of squared increments over the blocks by $C\sum_{n\ge a U^{1-z}} n^{-2b}$. Since $2b>1$, this is the right side of \eqref{eq:critic-block-noise-group}.
\end{proof}

\begin{lemma}[Asynchronous contraction on one block]
\label{lem:asynchronous-block-contraction}
Let $T$ be a $\theta$-contraction in $\|\cdot\|_\omega$, with fixed point $x^\star$. Suppose $y_{n+1}=y_n+A_n(T(y_n)-y_n+r_n)$ for $p\le n<q$, where $A_n$ is diagonal with entries in $[0,1]$, and $\sum_{p\le n<q}[A_n]_{jj}\ge\delta$ for each $j$. With $\rho=\max_{p\le n<q}\|r_n\|_\omega$ and $E_n=\|y_n-x^\star\|_\omega$, one has
\begin{align}
 E_q&\le [\theta+(1-\theta)e^{-\delta}] E_p+\rho/(1-\theta), \label{eq:asynchronous-block-contraction}\\
 \max_{p\le n\le q} E_n&\le\max\{E_p,\rho/(1-\theta)\}.\label{eq:asynchronous-block-envelope}
\end{align}
\end{lemma}

\begin{proof}
Work in scaled coordinates so the weighted norm becomes the ordinary sup norm. If $a_{n,j}=[A_n]_{jj}$, the absolute error of coordinate $j$ after a step is bounded by a convex combination of its previous absolute error and $\theta E_n+\rho$, with weights $1-a_{n,j}$ and $a_{n,j}$. The bound in \eqref{eq:asynchronous-block-envelope} follows by induction. Iterating the coordinate inequality leaves a coefficient $P_j=\prod_{p\le n<q}(1-a_{n,j})\le e^{-\delta}$ on the initial error. For the remaining coefficient, use $\theta E_n+\rho\le\theta E_p+\rho/(1-\theta)$. Taking the maximum over coordinates gives \eqref{eq:asynchronous-block-contraction}.
\end{proof}

\begin{lemma}[One-block moving-target contraction]
\label{lem:one-block-tracking}
For this block calculation, write $e_k^\omega=\|X_k-\Lambda(\Pi_k)\|_\omega$, so $\omega e_k\le e_k^\omega\le e_k$ for the error $e_k$ in Theorem~\ref{thm:AMCR-regret-tracking}. Write $D_t=\sum_{u=t}^{t+m_t-1}\|\Pi_{u+1}-\Pi_u\|_1$. On $\mathcal E_t$, for constants $q<1$ and $C$ independent of $t$,
\begin{align}
 e_{t+m_t}^\omega&\le q e_t^\omega+C(D_t+\eps_t+\mathcal N_t),
 \label{eq:one-block-tracking-recursion}\\
 \max_{t\le k\le t+m_t}e_k^\omega&\le e_t^\omega+C(D_t+\eps_t+\mathcal N_t).
 \label{eq:one-block-tracking-envelope}
\end{align}
\end{lemma}
\begin{proof}
Freeze the map $G_0(\Pi_t,\cdot)$ at the target profile at block start. Subtract the accumulated innovations by writing $Y_k=X_k-M_{t,k}$. Then
\begin{align*}
Y_{k+1}&=Y_k+\mathsf A_k\big(G_0(\Pi_t,Y_k)-Y_k+r_k\big),\\
r_k&=G_{\eps_k}(\Pi_k,X_k)-G_0(\Pi_t,Y_k)-(X_k-Y_k).
\end{align*}
Equation~\eqref{eq:critic-mixing-bound} and the two Lipschitz bounds in Proposition~\ref{prop:critic-regularity} bound $\|r_k\|_\omega$ by $C(D_t+\eps_t)+(1+\sqrt\gamma)\mathcal N_t$. Indeed, at $X_k\in\mathcal X$, mixing changes the opponents' law by $O(\eps_k)$, and the change in the target profile is at most $D_t$. The change from $X_k$ to $Y_k$ is controlled by contraction on all of $\R^m$. Now apply Lemma~\ref{lem:asynchronous-block-contraction} with $\theta=\sqrt\gamma$ and the cumulative step from Lemma~\ref{lem:uniform-critic-blocks}. Finally, putting back $M_{t,k}$ costs at most $\mathcal N_t$, and moving $\Lambda(\Pi_t)$ to $\Lambda(\Pi_k)$ costs at most $C D_t$ by Proposition~\ref{prop:critic-regularity}. This proves both bounds.
\end{proof}

\begin{proposition}[Moving-target joint online policy evaluation]
\label{prop:tracking-algorithm}
For the inverse-log schedule in Theorem~\ref{thm:algorithm-AMCR}, $e_t\to0$ almost surely. Under the choices in Theorem~\ref{thm:finite-time-AMCR},
\begin{equation}
\E\left[\frac1T\sum_{u\le T}e_u\right]\le C_r T^{-r},
\qquad
\frac1T\sum_{u\le T}e_u\longrightarrow0\quad\text{almost surely}.
\label{eq:finite-time-critic-tracking}
\end{equation}
If $r<1/[4(K+1)]$ and $b$ is sufficiently close to $1/2$, also $\E e_t=O(t^{-r})$ and $e_t\to0$ almost surely.
\end{proposition}

\begin{proof}
Start at a sufficiently large deterministic date $t_1$ and partition subsequent dates by $t_{j+1}=t_j+m_{t_j}$. Almost surely, $\mathcal E_{t_j}$ holds for all but finitely many $j$. Under inverse-log exploration, Lemma~\ref{lem:uniform-critic-blocks} gives $D_t=O(t^{b-c+o(1)})$ on $\mathcal E_t$. The noise bound has $B_t=t^{-b/2+o(1)}$. Apply \eqref{eq:blockwise-critic-noise-tail} with $x=\sqrt{B_t}$ and use Borel--Cantelli. The perturbation on the right of \eqref{eq:one-block-tracking-recursion} tends to zero almost surely. Iterating the recursion with $q<1$ gives $e_{t_j}^\omega\to0$ and \eqref{eq:one-block-tracking-envelope} gives the same result at every date.

For polynomial exploration put $h=b+z$. On $\mathcal E_t$ the perturbation consists of $D_t+\eps_t+\mathcal N_t$. On any block the errors are bounded by a constant because the tables and their targets stay in $\mathcal X$. Thus \eqref{eq:one-block-tracking-recursion} and \eqref{eq:one-block-tracking-envelope} hold on every path if the perturbation is replaced by a constant times
$D_t+\eps_t+\one\{\mathcal V_t\}\mathcal N_t+\one\{\mathcal E_t^c\}$.
Consecutive block lengths have ratio tending to one. Multiply the endpoint recursion by $m_{t_{j+1}}$ and sum. Eventually $q\,m_{t_{j+1}}\le q' m_{t_j}$ for some $q'<1$, so the sum of length-weighted endpoint errors is bounded by a constant plus a constant times the length-weighted sum of the perturbations. Equation~\eqref{eq:one-block-tracking-envelope} then bounds the sum of the weighted errors at all dates. Since $e_k\le e_k^\omega/\omega$, the same order holds in the ordinary sup norm. We may complete the last block that intersects $[1,T]$, since its endpoint is at most $2T$ for large $T$.

By Lemma~\ref{lem:actor-one-step}, the total target movement through $2T$ is at most $C\sum_s\sum_{n\le N_s(2T)}n^{-c}=O(T^{1-c})$. The blocks have length $O(T^h)$, so their length-weighted movement divided by $T$ is $O(T^{h-c})$. For the noise, group block starts in $[U,2U)$ for $U=1,2,4,\ldots$. There are $O(U^{1-h})$ such disjoint blocks, of length $O(U^h)$. Cauchy--Schwarz and \eqref{eq:critic-block-noise-group} give
\[
 \frac1U\sum_{U\le t_j<2U}m_{t_j}\E[\one\{\mathcal V_{t_j}\}\mathcal N_{t_j}]
 \le C U^{h-1+(1-h)/2-(2b-1)(1-z)/2}
 =C U^{-b(1-2z)/2}.
\]
Multiplying by $U$, summing over these powers of two up to $T$, and dividing by $T$ gives $O(T^{-b(1-2z)/2})$. The length-weighted exploration term gives $O(T^{-r})$ since $t_j\le u<t_{j+1}\le2t_j$ on each sufficiently late block. The expected length-weighted failure terms have a finite sum by \eqref{eq:critic-block-probabilities}. The choices $c-h>r$ and $b(1-2z)/2>r$ prove the expectation bound in \eqref{eq:finite-time-critic-tracking}. Markov's inequality on dyadic horizons and Borel--Cantelli give almost-sure vanishing there. The cumulative sum of the nonnegative errors is nondecreasing, which covers all intervening horizons.

For the last assertion, on $\mathcal E_t$ we have $D_t=O(t^{h-c(1-z)})$, and $B_t=O(t^{h/2-b(1-z)})$. For $r<1/[4(K+1)]$, both exponents are strictly less than $-r$ at $b=1/2$. This follows by substituting $c=1-r$ and $z=Kr$, and remains true for $b$ slightly larger. The expected endpoint recursion now has the form $\E e_{t_{j+1}}^\omega\le q\E e_{t_j}^\omega+C t_j^{-r}$. Since $t_{j+1}/t_j\to1$, induction gives $\E e_{t_j}^\omega=O(t_j^{-r})$. Equation~\eqref{eq:one-block-tracking-envelope} extends the bound to all dates, also in the ordinary sup norm. The summable tail estimate at $x=\sqrt{B_t}$ gives almost-sure vanishing by the same endpoint argument as above.

We can obtain pointwise convergence for some further exponents by distinguishing a state visit from a state--action visit. Reaching a state uses at most $d$ exploratory transitions and requires no final action draw. The same concentration argument gives state visit counts of order at least $t^{1-Ndr}$, apart from a summable failure probability of order $C\exp(-a t^{1-2Ndr})$. Suppose $r<1/(2K+4)$ and $(1-r)(1-Ndr)>1/2+Kr$. Choose $b$ sufficiently close to $1/2$ that $c(1-Ndr)>h$. On $\mathcal E_t$ and this additional state-count event we have $D_t=O(t^{h-c(1-Ndr)})$. For the noise, on $\mathcal V_t$ Doob's inequality and the local sum of squared steps, as in \eqref{eq:critic-block-noise-group}, give $\E[\one\{\mathcal V_t\}\mathcal N_t]=O(t^{-(1-z)(b-1/2)})$. Applying \eqref{eq:one-block-tracking-recursion} and \eqref{eq:one-block-tracking-envelope} as before gives $\E e_t=O(t^{-a'})$ with $a'=\min\{r,c(1-Ndr)-h,(1-z)(b-1/2)\}>0$. For almost-sure convergence, the full martingale with half-step increments $\Delta_j$ converges: each coordinate takes each local step at most once, and $\sum_n n^{-2b}<\infty$ bounds the total expected squared increments. All late block fluctuations therefore vanish, and so does the critic error by the same contraction argument.
\end{proof}

\subsection{Actor regret and the AMCR bound}

Under Assumption~\ref{ass:communication},
\eqref{eq:critic-global-frequency} implies that $t_q(s)<\infty$ for every
state $s$ and every $q\ge1$ almost surely. Each $t_q(s)$ is an
$(\mathcal H_t)_{t\ge1}$-stopping time. For any almost surely finite,
integer-valued stopping time $\tau$, let
\begin{equation*}
\mathcal H_\tau
:=
\left\{
A\in\mathcal F:
A\cap\{\tau=t\}\in\mathcal H_t
\ \text{for every }t\ge1
\right\}
\end{equation*}
denote the corresponding stopped sigma-field.\footnote{We use the following
elementary identity. If $\tau$ is an almost surely finite
$(\mathcal H_t)_{t\ge1}$-stopping time, $Y_t$ is bounded, and $Y_t$ is
$\mathcal H_{t+1}$-measurable for every $t$, then
\begin{equation*}
\E[Y_\tau\mid\mathcal H_\tau]
=
\sum_{t=1}^{\infty}
\one\{\tau=t\}\E[Y_t\mid\mathcal H_t].
\end{equation*}}

For a fixed state $s$, write
$\widehat g_{i,q}^s:=
Q_{i,t_q(s)}\bigl(s,(\cdot,a_{-i,t_q(s)})\bigr)$ and
$\mathcal H_{s,q}:=\mathcal H_{t_q(s)}$. Define
\begin{equation}
\label{eq:observed-gain-md}
g_{i,q}^{p,s}
:=
\E[\widehat g_{i,q}^s\mid\mathcal H_{s,q}],
\qquad
\Delta_{i,q}^s
:=
\widehat g_{i,q}^s-g_{i,q}^{p,s}.
\end{equation}
Applying the stopped-time identity above coordinatewise and using
\eqref{eq:behavior-sampling} gives, for every $a_i\in A_i$,
\begin{equation}
\label{eq:behavior-law-gain-at-visits}
g_{i,q}^{p,s}(a_i)
=
\sum_{a_{-i}}
p_{-i,t_q(s)}(a_{-i}\mid s)
Q_{i,t_q(s)}(s,(a_i,a_{-i})).
\end{equation}

\begin{proposition}[Observed-feedback regret]
\label{prop:observed-feedback-regret}
For every player $i$ and state $s$, there is a nonnegative process
$\mathcal M_{i,s}(n)$, $n\ge0$, with $\mathcal M_{i,s}(0)=0$, such that
$\mathcal M_{i,s}(n)/n\to0$ almost surely,
$\E[\mathcal M_{i,s}(n)]\le K_iH\sqrt n$ for every $n\ge1$, and
\begin{equation}
\label{eq:observed-feedback-regret}
R_{i,s}^{\mathrm{tar}}(n)
\le
\frac{D_{i,s}^{\mathrm A}}{\beta_{s,n}}
+C_{i,s}^{\mathrm A}\sum_{q=1}^n\beta_{s,q}
+\mathcal M_{i,s}(n)
+4H(N-1)\sum_{q=1}^n\eps_{t_q(s)}
\end{equation}
for every $n\ge1$.
\end{proposition}

\begin{proof}
Apply \eqref{eq:actor-regret-interface} pathwise to the realized vectors
$\widehat g_{i,q}^s$. To analyze the sampling error on the local clock of state $s$, let $\mathcal K_{s,0}$ be the trivial sigma-field and
set $\mathcal K_{s,q}:=\mathcal H_{t_q(s)+1}$ for $q\ge1$. 
Thus, $\mathcal H_{s,q}=\mathcal H_{t_q(s)}$ is the information available immediately before the action at the $q$-th visit to $s$, whereas $\mathcal K_{s,q}$ contains the information available after that visit. In particular, $\Delta^s_{i,q}$ is $\mathcal K_{s,q}$-measurable. Moreover, for $q\ge2$,
$t_{q-1}(s)+1\le t_q(s)$, so monotonicity of stopped sigma-fields gives $\mathcal K_{s,q-1}\subseteq\mathcal H_{s,q}$.  
The same inclusion holds for $q=1$ because $\mathcal K_{s,0}$ is trivial.
Then, by definition of  $\Delta_{i,q}^s$, we have $\E[\Delta_{i,q}^s\mid\mathcal H_{s,q}]=0$. 
Therefore,
$\E[\Delta_{i,q}^s(a_i)\mid\mathcal K_{s,q-1}]=0$ for every action $a_i$.
Since $\pi_{i,t_q(s)}(\cdot\mid s)$ is
$\mathcal H_{s,q}$-measurable, the tower property also gives
\begin{equation*}
\E\!\left[
\left\langle
\pi_{i,t_q(s)}(\cdot\mid s),\Delta_{i,q}^s
\right\rangle
\,\middle|\,
\mathcal K_{s,q-1}
\right]
=0.
\end{equation*}
Consequently, for every action $a_i$,
\begin{equation*}
\left(\sum_{q=1}^n\Delta_{i,q}^s(a_i)\right)_{n\ge0}
\quad\text{and}\quad
\left(
\sum_{q=1}^n
\left\langle
\pi_{i,t_q(s)}(\cdot\mid s),\Delta_{i,q}^s
\right\rangle
\right)_{n\ge0}
\end{equation*}
are martingales with respect to $(\mathcal K_{s,n})_{n\ge0}$. Define
\begin{equation*}
\mathcal M_{i,s}(n)
:=
\max_{1\le m\le n}
\left\{
\max_{a_i}\left|\sum_{q=1}^m\Delta_{i,q}^s(a_i)\right|
+
\left|\sum_{q=1}^m
\left\langle
\pi_{i,t_q(s)}(\cdot\mid s),\Delta_{i,q}^s
\right\rangle\right|
\right\}.
\end{equation*}
Thus $\mathcal M_{i,s}(n)$ simultaneously controls, uniformly over all prefixes $m\le n$, both the coordinatewise sampling errors for every fixed action and the sampling error of the actions actually mixed by the actor.
All martingale increments are bounded. Hence, for each of the finitely
many martingales $Y$, the series
$\sum_{n\ge1}n^{-2}\E[(Y_n-Y_{n-1})^2]$ is finite, and the martingale
strong law gives $Y_n/n\to0$. It follows that
$\max_{m\le n}|Y_m|/n\to0$: after a fixed initial segment, use
$|Y_m|=o(m)$, while the initial segment divided by $n$ vanishes. Therefore,
$\mathcal M_{i,s}(n)/n\to0$ almost surely. The terminal second moment of
each martingale is at most $nH^2$. Applying Doob's $L^2$ maximal inequality
to these finitely many martingales and summing the resulting bounds gives
$\E[\mathcal M_{i,s}(n)]\le K_iH\sqrt n$.

It remains to compare the conditional mean in \eqref{eq:observed-gain-md}
with the target-law vector in \eqref{eq:target-law-gain}. By
\eqref{eq:behavior-law-gain-at-visits}, the bound on products obtained by changing one factor at a time, and
\eqref{eq:fixed-mixing},
$\TV(p_{-i,t_q(s)}(\cdot\mid s),
\Pi_{-i,t_q(s)}(\cdot\mid s))\le(N-1)\eps_{t_q(s)}$. Since every gain coordinate
lies in $[0,H]$,
\begin{equation}
\label{eq:behavior-target-gain-bias}
\|g_{i,q}^{p,s}-g_{i,t_q(s)}^{\mathrm{tar}}\|_\infty
\le2H(N-1)\eps_{t_q(s)}.
\end{equation}
For every action $a_i$, \eqref{eq:observed-gain-md} gives
\begin{align}
&\sum_{q=1}^n
\left[
 g_{i,t_q(s)}^{\mathrm{tar}}(a_i)
 -\left\langle\pi_{i,t_q(s)}(\cdot\mid s),
 g_{i,t_q(s)}^{\mathrm{tar}}\right\rangle
\right]
\notag\\
&=\sum_{q=1}^n
\left[
 \widehat g_{i,q}^s(a_i)
 -\left\langle\pi_{i,t_q(s)}(\cdot\mid s),
 \widehat g_{i,q}^s\right\rangle
\right]
\notag\\
&\quad-\sum_{q=1}^n
\left[
 \Delta_{i,q}^s(a_i)
 -\left\langle\pi_{i,t_q(s)}(\cdot\mid s),
 \Delta_{i,q}^s\right\rangle
\right]
\notag\\
&\quad+\sum_{q=1}^n
\left[
 (g_{i,t_q(s)}^{\mathrm{tar}}-g_{i,q}^{p,s})(a_i)
 -\left\langle\pi_{i,t_q(s)}(\cdot\mid s),
 g_{i,t_q(s)}^{\mathrm{tar}}-g_{i,q}^{p,s}
 \right\rangle
\right].
\label{eq:observed-to-target-regret-decomposition}
\end{align}
The second sum on the right-hand side of
\eqref{eq:observed-to-target-regret-decomposition} is bounded in absolute
value by $\mathcal M_{i,s}(n)$. By
\eqref{eq:behavior-target-gain-bias}, the third sum is at most
$4H(N-1)\sum_{q=1}^n\eps_{t_q(s)}$. Maximizing over $a_i$, taking the positive part, and
applying \eqref{eq:actor-regret-interface} to the first sum proves
\eqref{eq:observed-feedback-regret}.
\end{proof}

\begin{proof}[Proof of Theorem~\ref{thm:algorithm-AMCR}]
Proposition~\ref{prop:tracking-algorithm} and equivalence of the weighted and ordinary sup norms give \eqref{eq:main-critic-tracking}. In Proposition~\ref{prop:observed-feedback-regret}, the first two terms grow as $O(n^c+n^{1-c})$. The martingale term is $o(n)$ almost surely and $n^{-1}\sum_{q\le n}\eps_{t_q(s)}\to0$ since $t_q(s)\ge q$. This proves \eqref{eq:main-target-regret} at every state. Theorem~\ref{thm:AMCR-regret-tracking} now gives vanishing target-path AMCR. By Lemma~\ref{lem:AMCR-advantage-lipschitz}, the behavior-path gap differs from it by at most $C T^{-1}\sum_{t\le T}\eps_t\to0$. This proves \eqref{eq:algorithm-pathwise-AMCR-main}.

It remains to verify the state distribution. By \eqref{eq:actor-total-variation} and monotonicity of exploration, the total movement of the executed profiles is $O(T^{1-c})$. Hence \eqref{eq:mbcce-average-increments} holds and Proposition~\ref{prop:empirical-markov-bcce} gives the asserted MBCCE conclusion.

One can also bound the error in state balance at a fixed horizon. For a test $f$ bounded by one with Lipschitz constant at most one, this error at state $s$ is the absolute value of $T^{-1}\sum_{t\le T} f(p_t)[\one\{s_t=s\}-P_{p_t}(s\mid s_t)]$. Replacing the indicator by $\one\{s_{t+1}=s\}$ changes the average by at most $(2+\sum_{t<T}\|p_{t+1}-p_t\|_1)/T$ by the same index shift as in Proposition~\ref{prop:empirical-markov-bcce}. The new numerator is a sum of bounded martingale differences, with expected absolute value at most $\sqrt T$. Thus the expected absolute balance error is $O(T^{-1/2}+T^{-c})$ under either schedule.
\end{proof}

\subsection{Finite-time AMCR rate}
\label{app:finite-time-AMCR-rate}
\begin{proof}[Proof of Theorem~\ref{thm:finite-time-AMCR}]
The tracking and statewise regret proofs also apply with $\eps_t=(t+1)^{-r}$. Almost surely \eqref{eq:finite-time-critic-tracking} gives vanishing average error, the regret on each visit clock is $o(n)$ by Proposition~\ref{prop:observed-feedback-regret}, and the preceding movement argument gives \eqref{eq:mbcce-average-increments}. Theorem~\ref{thm:AMCR-regret-tracking} and Proposition~\ref{prop:empirical-markov-bcce} therefore prove almost-sure vanishing AMCR and the equilibrium claim.

For the finite-time bound, apply \eqref{eq:observed-feedback-regret} at $n=N_s(T)$, assigning zero regret to an empty visit clock. Since $\mathcal M_{i,s}$ is nondecreasing and $N_s(T)\le T$, its expected contribution is at most $K_iH\sqrt{T}$ per state and player. The exploration terms sum over their actual dates, giving $O(\sum_{t\le T}\eps_t)$. Summing the deterministic actor terms over finitely many states gives $O(T^c+T^{1-c})$. Consequently
\begin{equation}
\label{eq:finite-time-aggregate-regret}
\frac1T\sum_{s,i}\E[R_{i,s}^{\rm tar}(N_s(T))]
\le C\left(T^{c-1}+T^{-c}+T^{-1/2}+T^{-r}\right).
\end{equation}
Theorem~\ref{thm:AMCR-regret-tracking}, \eqref{eq:finite-time-critic-tracking}, and $c=1-r$ give the claimed order for the target path. Transferring to the behavior path costs $O(T^{-1}\sum_{t\le T}\eps_t)=O(T^{-r})$ by Lemma~\ref{lem:AMCR-advantage-lipschitz}.

To see the scope of this rate, let $a$ denote the exponent of a general power exploration schedule and ask when the preceding bounds give $O(T^{-r})$. The exploration term requires $a\ge r$, and the first term of \eqref{eq:finite-time-aggregate-regret} requires $1-c\ge r$. With the coverage bound used here, a block of order $t^{b+Ka}$ is required to guarantee a fixed cumulative update in every coordinate. The bound on its time-averaged movement then requires $c-b-Ka\ge r$. Adding these restrictions gives $(K+2)r\le1-b<1/2$. For every $r<1/(2K+4)$ the choices in the theorem satisfy these bounds and \eqref{eq:critic-block-noise-group}. This optimizes the exponents in these estimates.
\end{proof}

Shortening the blocks weakens the guaranteed critic update at the least frequently visited pair, while lengthening them increases the bound on policy movement. Whether one can bound the algorithm's incentive gap more sharply is open. If the transitions require fewer exploratory choices to reach a state, the same proof can improve it. For example, in the investment game, exploration by just one firm gives the transition from state~0 to state~1 probability at least $\eps_t/3$, and leaving state~1 has probability $1/4$. At the desired state the probability of drawing any specified pair of actions is at least $(\eps_t/2)^2$. The proof consequently allows $r<1/10$ in that example.

\begin{proof}[Proof of Corollary~\ref{cor:actor-generated-law}]
Fix a Lipschitz function $f$ on $\mathcal C$ and a state $s$. We first prove
\begin{equation}
 \E\left|\int f(\Pi)\,[\widehat\mu_T(\{s\},d\Pi)-\bar\mu_T(\{s\},d\Pi)]\right|
 =O_f(T^{-1/2+Ndr}+T^{-c+2Ndr}).
 \label{eq:actor-state-frequency-test}
\end{equation}
Fix $T$ and a reference state $s_0$. Consider a frozen behavior profile $p$ with mixing level at least $\eps_T=(T+1)^{-r}$. In the path argument of Lemma~\ref{lem:uniform-critic-blocks}, reaching $s_0$ takes at most $d$ transitions. The state test has no prescribed action upon arrival, so the probability of reaching $s_0$ within these transitions is at least $a\eps_T^{Nd}$ for a constant $a>0$. For more than one state, applying this bound in successive windows shows that the first hitting time of $s_0$ has mean at most $CT^{Ndr}$ uniformly in the initial state and in $p$. In a one-state game the two state laws are identical and the test error is zero. When the game admits the smaller joint hitting power described after \eqref{eq:critic-global-frequency}, a hit of any specified pair at $s_0$ is also a hit of $s_0$. The state estimates below can therefore use the smaller of that power and $Nd$.

For any bounded state function $g$, write $d_p g=\sum_x d_p(x)g(x)$. There is a unique solution $h$ of $h-P_p h=g-d_p g$ with $h(s_0)=0$. Indeed, until the first visit to $s_0$, sum $g-d_pg$ and take its expectation starting at each state. The first-step equation holds off $s_0$, and invariance with $d_p(s_0)>0$ gives the equation also at $s_0$. The same hitting representation proves uniqueness. It gives $\|h\|_\infty\le C T^{Ndr}\|g\|_\infty$.

Apply this to $g_p(x)=f(p)\one\{x=s\}$ and call the solution $h_p$. For profiles $p,q$ with the specified mixing bound, the difference $h_p-h_q$ is the normalized solution for the kernel $P_p$ and right-hand side $g_p-g_q+(P_p-P_q)h_q$ after subtraction of its $d_p$-mean. The map $p\mapsto P_p$ is Lipschitz by the product formula for joint actions. Applying the preceding bound again gives
\[
 \|h_p-h_q\|_\infty\le C_f T^{2Ndr}\|p-q\|_1.
\]
All the profiles $p_t$ up to $T$ meet the mixing bound. Summing their equations along the observed trajectory decomposes $\sum_{t\le T}[g_{p_t}(s_t)-d_{p_t}g_{p_t}]$ into the martingale sum of $h_{p_t}(s_{t+1})-(P_{p_t}h_{p_t})(s_t)$ and the sum of $h_{p_t}(s_t)-h_{p_t}(s_{t+1})$. The martingale has second moment $O_f(T^{1+2Ndr})$. On shifting the index, the other sum has two endpoint terms of order $O_f(T^{Ndr})$ and is otherwise bounded by $C_fT^{2Ndr}\sum_{t<T}\|p_{t+1}-p_t\|_1$. This sum of policy increments is $O(T^{1-c})$ by Lemma~\ref{lem:actor-one-step} on each local clock and monotonicity of exploration. Dividing by $T$ proves \eqref{eq:actor-state-frequency-test}.

Now apply \eqref{eq:actor-state-frequency-test} to the finitely many $f(\Pi)=\mathcal A_i(\Pi;s,b_i)$, which are Lipschitz by Lemma~\ref{lem:AMCR-advantage-lipschitz}. The maxima and positive parts change by at most the maximum absolute change in their arguments. Since $(Nd+1)r\le(K+1)r<1/2$ and $c=1-r$, both $T^{-1/2+Ndr}$ and $T^{-c+2Ndr}$ are $O(T^{-r})$. The empirical incentive gap has expectation $O(T^{-r})$ by \eqref{eq:AMCR-empirical-mbcce} and \eqref{eq:finite-time-AMCR-behavior}. This proves \eqref{eq:actor-generated-gap}.
\end{proof}

\section{Proofs for Section~\ref{sec:reinforce-AMCR}}
\label{app:reinforce-AMCR}

\begin{proof}[Proof of Theorem~\ref{thm:reinforce-AMCR}]
We first identify the fixed-action deviation gain estimated by one episode. We then use the projection update to control cumulative gains, taking account of the random number of stages in each episode. Finally we replace realized visits by each profile's discounted state frequencies and pass from completed episodes to arbitrary stage horizons.

Let $\mathscr F_e$ be the $\sigma$-field available immediately before episode $e$. In particular, the target and behavior profiles used during episode $e$ are $\mathscr F_e$-measurable. Conditional on $\mathscr F_e$, the initial state $s_{e,0}$ is drawn from $\rho_0$, and independently the terminal index $L_e$ is drawn according to \eqref{eq:reinforce-horizon}. The episode is then generated using the fixed behavior profile $p_e$ and the transition kernel $P$. The players observe termination when it occurs. Let $\mathscr F_{e+1}$ be generated by $\mathscr F_e$, the complete history of episode $e$, and the update in \eqref{eq:reinforce-update}. The draws $(s_{e,0},L_e)$ are independent across episodes, and the random draws used for initialization and termination are independent of those used for actions and state transitions.

For each deterministic stage index $k$, let $\mathscr H_{e,k}$ be generated by $\mathscr F_e$, the history through arrival at $s_{e,k}$, and the indicators of survival through dates $0,\ldots,k$. Assign a dummy state and action after termination and set the corresponding score to zero. This history is used to condition expectations in the proof. On $\{L_e\ge k\}$ it includes the current state but gives no further information about the remaining lifetime. Write $N_{e,s}=\sum_{k=0}^{L_e}\one\{s_{e,k}=s\}$. For fixed $e$, each score coordinate is bounded by $|A_i|/\varepsilon_e$ and both the return and the number of score vectors are bounded by $Y_e$, whose fourth moment is finite. Thus the sums over stages below are absolutely integrable, and $\widehat G_{i,e}$ has finite second moment.

Fix player $i$, state $s$, and action $b_i\in A_i$. In this proof, let
$d_{i,e}^{s,b_i}:=e_{b_i}-\pi_{i,e}(\cdot\mid s)$ and
\begin{equation*}
q_{i,e}^s(a_i)
:=
\sum_{a_{-i}\in A_{-i}}
 p_{-i,e}(a_{-i}\mid s)
 Q_i^{p_e}(s,(a_i,a_{-i})).
\end{equation*}
Let $\widehat G^s_{i,e}$ denote the state-$s$ block of the estimate in \eqref{eq:reinforce-estimator}. Write $Z_{e,k}$ for the inner product of $d_{i,e}^{s,b_i}$ with the state-$s$ block of the stage-$k$ log-policy gradient, setting $Z_{e,k}=0$ after termination. At a surviving visit to $s$, the direct parameterization gives
\[
 Z_{e,k}=(1-\varepsilon_e)
 \frac{d_{i,e}^{s,b_i}(a_{i,e,k})}{p_{i,e}(a_{i,e,k}\mid s)} .
\]
The score is zero at other states. Its conditional mean given $\mathscr H_{e,k}$ is zero because the coordinates of $d_{i,e}^{s,b_i}$ sum to zero. Thus rewards earned before stage $k$, which are measurable in this history, have zero expected contribution when multiplied by $Z_{e,k}$.

Conditional on $\mathscr H_{e,k}$ and $a_{e,k}=a$, the memoryless property
of \eqref{eq:reinforce-horizon} gives, on $\{L_e\ge k\}$,
\begin{equation*}
\E\!\left[
\sum_{\ell=k}^{L_e}u_i(s_{e,\ell},a_{e,\ell})
\,\middle|\,
\mathscr H_{e,k},\ a_{e,k}=a
\right]
=
Q_i^{p_e}(s_{e,k},a).
\end{equation*}
By \eqref{eq:reinforce-estimator}, $\langle d_{i,e}^{s,b_i},\widehat G^s_{i,e}\rangle$ is the sum over score dates $k$ of $Z_{e,k}$ times the episode return. The preceding zero-mean argument removes rewards before each score date when taking expectations. For the remaining rewards, conditioning
additionally on $a_{e,k}=a$ and using the memoryless identity above gives
\[
\E\left[
Z_{e,k}
\sum_{\ell=k}^{L_e}u_i(s_{e,\ell},a_{e,\ell})
\,\middle|\,
\mathscr H_{e,k}
\right]
=
(1-\varepsilon_e)
\one\{L_e\ge k,\ s_{e,k}=s\}
\left\langle
d^{s,b_i}_{i,e},
q^s_{i,e}
\right\rangle .
\]
Indeed, conditional on $\mathscr H_{e,k}$ and $s_{e,k}=s$, averaging the
stage-$k$ action under $p_e$ yields
\[
\sum_{a_i,a_{-i}}
p_{i,e}(a_i\mid s)
p_{-i,e}(a_{-i}\mid s)
\frac{(1-\varepsilon_e)d^{s,b_i}_{i,e}(a_i)}
     {p_{i,e}(a_i\mid s)}
Q_i^{p_e}(s,(a_i,a_{-i}))
=
(1-\varepsilon_e)
\left\langle
d^{s,b_i}_{i,e},
q^s_{i,e}
\right\rangle .
\]
Summing over $k$ and conditioning on $\mathscr F_e$, while noting that
$d^{s,b_i}_{i,e}$ and $q^s_{i,e}$ are $\mathscr F_e$-measurable, gives
\begin{equation}
\label{eq:reinforce-directional-identity}
\E\!\left[
\left\langle
 d_{i,e}^{s,b_i},
 \widehat G_{i,e}^{\,s}
\right\rangle
\,\middle|\,
\mathscr F_e
\right]
=
\E[N_{e,s}\mid\mathscr F_e]
(1-\varepsilon_e)
\left\langle d_{i,e}^{s,b_i},q_{i,e}^s\right\rangle.
\end{equation}

The second equation in \eqref{eq:canonical-bellman-system} gives
$V_i^{p_e}(s)=\langle p_{i,e}(\cdot\mid s),q_{i,e}^s\rangle$. Hence
\begin{align}
\mathcal A_i(p_e;s,b_i)
&=
q_{i,e}^s(b_i)
-
\left\langle p_{i,e}(\cdot\mid s),q_{i,e}^s\right\rangle
\notag\\
&\le
(1-\varepsilon_e)
\left\langle d_{i,e}^{s,b_i},q_{i,e}^s\right\rangle
+\varepsilon_e H.
\label{eq:reinforce-exploration-comparison}
\end{align}
Indeed, the difference between the two sides before the term
$\varepsilon_e H$ is
$\varepsilon_e[q_{i,e}^s(b_i)-\langle\mathsf U_i,q_{i,e}^s\rangle]\le
\varepsilon_e H$.

Projection in \eqref{eq:reinforce-update} separates across states. The
standard projected-gradient inequality with nonincreasing steps gives
\begin{equation*}
\sum_{e=1}^E
\left\langle
 d_{i,e}^{s,b_i},
 \widehat G_{i,e}^{\,s}
\right\rangle
\le
\frac1{\eta_E}
+
\frac12\sum_{e=1}^E
\eta_e\|\widehat G_{i,e}^{\,s}\|_2^2;
\end{equation*}
see \cite{Zinkevich2003}. The first term is at most $1/\eta_E$ because the
squared Euclidean diameter of a simplex is at most two.

Define
\begin{align*}
\Delta M_{i,s,b_i,e}
:={}&
N_{e,s}(1-\varepsilon_e)
\left\langle d_{i,e}^{s,b_i},q_{i,e}^s\right\rangle
-
\left\langle d_{i,e}^{s,b_i},\widehat G_{i,e}^{\,s}\right\rangle,
\\
M_{i,s,b_i}(E)
:={}&\sum_{e=1}^E\Delta M_{i,s,b_i,e}.
\end{align*}
Each increment is $\mathscr F_{e+1}$-measurable, and
equation~\eqref{eq:reinforce-directional-identity} gives
$\E[\Delta M_{i,s,b_i,e}\mid\mathscr F_e]=0$. Thus
$(M_{i,s,b_i}(E),\mathscr F_{E+1})_{E\ge0}$ is a martingale. Combining the
projected-gradient inequality with
\eqref{eq:reinforce-exploration-comparison} gives
\begin{align*}
\sum_{e=1}^E N_{e,s}\mathcal A_i(p_e;s,b_i)
\le{}&
\frac1{\eta_E}
+
\frac12\sum_{e=1}^E
\eta_e\|\widehat G_{i,e}^{\,s}\|_2^2
+
M_{i,s,b_i}(E)
+
H\sum_{e=1}^E\varepsilon_eN_{e,s}.
\end{align*}
Maximizing over $b_i$, taking positive parts, and summing over players and
states yields
\begin{align}
\mathcal G_{K_E}^{\mathrm{AMCR}}
\le{}&
\frac{N|\cS|}{K_E\eta_E}
+
\frac1{2K_E}
\sum_{e=1}^E\eta_e
\sum_{i=1}^N\|\widehat G_{i,e}\|_2^2
\notag\\
&+
\frac1{K_E}
\sum_{i=1}^N\sum_{s\in\cS}
\max_{b_i\in A_i}|M_{i,s,b_i}(E)|
+
\frac{NH}{K_E}
\sum_{e=1}^E\varepsilon_e(L_e+1).
\label{eq:reinforce-master-bound}
\end{align}

Conditional on $\mathscr F_e$, the profile and exploration level are fixed. We give the second-moment estimate used in Lemma~D.4 in \cite{LeonardosEtAl2022} for our episode filtration. Conditional also on $Y_e=m$, the $m$ stages follow the frozen policy, since termination is independent of the action and transition draws. At any state $s$ the conditional expected squared norm of a single ambient score is $(1-\varepsilon_e)^2\sum_{a_i}1/p_{i,e}(a_i\mid s)\le A_{\max}^2/\varepsilon_e$. The return is at most $m$. Cauchy--Schwarz applied to the sum of the $m$ score vectors thus bounds the conditional second moment of the estimator by $m^4 A_{\max}^2/\varepsilon_e$. Since $Y_e^4\le Y_e(Y_e+1)(Y_e+2)(Y_e+3)$ and the latter product has mean $24H^4$ by the geometric series, we obtain
\begin{equation}
\label{eq:reinforce-second-moment}
\E\!\left[
\|\widehat G_{i,e}\|_2^2
\,\middle|\,
\mathscr F_e
\right]
\le
\frac{24A_{\max}^2H^4}{\varepsilon_e}.
\end{equation}

We next bound the martingale increments. Since
$|(1-\varepsilon_e)\langle d_{i,e}^{s,b_i},q_{i,e}^s\rangle|\le H$,
\begin{equation*}
\sum_{s\in\cS}\sum_{b_i\in A_i}
N_{e,s}^2(1-\varepsilon_e)^2
\left\langle d_{i,e}^{s,b_i},q_{i,e}^s\right\rangle^2
\le
A_{\max}H^2(L_e+1)^2.
\end{equation*}
Also, Cauchy--Schwarz and
$\sum_{b_i\in A_i}\|e_{b_i}-\pi_{i,e}(\cdot\mid s)\|_2^2\le2|A_i|$
give, for every state block,
\begin{equation*}
\sum_{b_i\in A_i}
\left\langle
 e_{b_i}-\pi_{i,e}(\cdot\mid s),
 \widehat G_{i,e}^{\,s}
\right\rangle^2
\le
2|A_i|\|\widehat G_{i,e}^{\,s}\|_2^2.
\end{equation*}
Using $(x-y)^2\le2x^2+2y^2$,
$\E[(L_e+1)^2]=(1+\gamma)H^2\le2H^2$, and
\eqref{eq:reinforce-second-moment}, we obtain
\begin{equation}
\label{eq:reinforce-martingale-second-moment}
\sum_{s\in\cS}\sum_{b_i\in A_i}
\E\!\left[
(\Delta M_{i,s,b_i,e})^2
\,\middle|\,
\mathscr F_e
\right]
\le
\frac{100A_{\max}^3H^4}{\varepsilon_e}.
\end{equation}
Moreover,
\begin{equation*}
\left(
\sum_{i=1}^N\sum_{s\in\cS}
\max_{b_i\in A_i}|M_{i,s,b_i}(E)|
\right)^2
\le
N|\cS|
\sum_{i=1}^N\sum_{s\in\cS}\sum_{b_i\in A_i}
M_{i,s,b_i}(E)^2.
\end{equation*}
Since martingale increments from different episodes are orthogonal,
\eqref{eq:reinforce-martingale-second-moment} implies
\begin{equation}
\label{eq:reinforce-martingale-aggregate}
\E\!\left[
\left(
\sum_{i=1}^N\sum_{s\in\cS}
\max_{b_i\in A_i}|M_{i,s,b_i}(E)|
\right)^2
\right]
\le
100N^2|\cS|A_{\max}^3H^4
\sum_{e=1}^E\frac1{\varepsilon_e}.
\end{equation}

The variables $(Y_e)$ are independent geometric random
variables on $\{1,2,\ldots\}$ with mean $H$, and
$K_E=\sum_{e=1}^E Y_e$. The variable $K_E$ has probability mass
$\bbP(K_E=k)=\binom{k-1}{E-1}(1-\gamma)^E\gamma^{k-E}$ for $k\ge E$.
For $E\ge3$, the identities
\begin{equation*}
\frac1{k-1}\binom{k-1}{E-1}
=
\frac1{E-1}\binom{k-2}{E-2},
\qquad
\frac1{(k-1)(k-2)}\binom{k-1}{E-1}
=
\frac1{(E-1)(E-2)}\binom{k-3}{E-3}
\end{equation*}
and the binomial series give
$\E[(K_E-1)^{-1}]=[H(E-1)]^{-1}$ and
$\E[((K_E-1)(K_E-2))^{-1}]=[H^2(E-1)(E-2)]^{-1}$.
Therefore, for $E\ge3$,
\begin{equation}
\label{eq:reinforce-inverse-horizon}
\E\!\left[\frac1{K_E}\right]
\le
\frac1{H(E-1)},
\qquad
\E\!\left[\frac1{K_E^2}\right]
\le
\frac1{H^2(E-1)(E-2)}.
\end{equation}

The estimator from episode $e$ depends on $Y_e$, so it cannot be separated
directly from $K_E^{-1}$. Let $m:=\lfloor E/2\rfloor$. If $e\le m$, the
sum $\sum_{j=m+1}^E Y_j$ is independent of $\widehat G_{i,e}$. Thus
\begin{align*}
\E\!\left[
\frac{\|\widehat G_{i,e}\|_2^2}{K_E}
\right]
&\le
\E[\|\widehat G_{i,e}\|_2^2]
\E\!\left[\frac1{\sum_{j=m+1}^E Y_j}\right].
\end{align*}
If $e>m$, the sum $\sum_{j=1}^mY_j$ is $\mathscr F_e$-measurable, so
\begin{align*}
\E\!\left[
\frac{\|\widehat G_{i,e}\|_2^2}{K_E}
\right]
&\le
\E\!\left[
\frac1{\sum_{j=1}^mY_j}
\E[\|\widehat G_{i,e}\|_2^2\mid\mathscr F_e]
\right].
\end{align*}
Each half contains at least three episodes when $E\ge6$. Applying
\eqref{eq:reinforce-second-moment} and
\eqref{eq:reinforce-inverse-horizon} to the relevant half gives
\begin{equation}
\label{eq:reinforce-estimator-random-horizon}
\E\!\left[
\frac{\|\widehat G_{i,e}\|_2^2}{K_E}
\right]
\le
\frac{96A_{\max}^2H^3}{E\varepsilon_e},
\qquad E\ge6.
\end{equation}

For the martingale term, Cauchy--Schwarz,
\eqref{eq:reinforce-martingale-aggregate}, and
\eqref{eq:reinforce-inverse-horizon} give
\begin{align*}
&\E\!\left[
\frac1{K_E}
\sum_{i=1}^N\sum_{s\in\cS}
\max_{b_i\in A_i}|M_{i,s,b_i}(E)|
\right]
\\
&\qquad\le
\frac{20NA_{\max}^{3/2}H\sqrt{|\cS|}}{E}
\left(\sum_{e=1}^E\frac1{\varepsilon_e}\right)^{1/2}.
\end{align*}
Finally, exchangeability and $\sum_{e=1}^EY_e/K_E=1$ give
$\E[Y_e/K_E]=1/E$ for every $e$.

Taking expectations in \eqref{eq:reinforce-master-bound} and using
\eqref{eq:reinforce-estimator-random-horizon} gives, for $E\ge6$,
\begin{align}
\E\!\left[\mathcal G_{K_E}^{\mathrm{AMCR}}\right]
\le{}&
\frac{2N|\cS|}{HE\eta_E}
+
\frac{48NA_{\max}^2H^3}{E}
\sum_{e=1}^E\frac{\eta_e}{\varepsilon_e}
\notag\\
&+
\frac{20NA_{\max}^{3/2}H\sqrt{|\cS|}}{E}
\left(\sum_{e=1}^E\frac1{\varepsilon_e}\right)^{1/2}
+
\frac{NH}{E}\sum_{e=1}^E\varepsilon_e.
\label{eq:reinforce-expected-bound}
\end{align}
Under the schedules in the theorem,
\begin{align*}
\frac1{E\eta_E}
&=O\!\left(
A_{\max}H^2|\cS|^{-1/2}E^{-1/3}
\right),
&
\sum_{e=1}^E\frac{\eta_e}{\varepsilon_e}
&=O\!\left(
\frac{|\cS|^{1/2}}{A_{\max}H^2}E^{2/3}
\right),
\\
\sum_{e=1}^E\varepsilon_e
&=O(E^{2/3}),
&
\sum_{e=1}^E\frac1{\varepsilon_e}
&=O(E^{4/3}).
\end{align*}
Substitution into \eqref{eq:reinforce-expected-bound} proves
\eqref{eq:reinforce-expected-rate} for $E\ge6$. Increasing $C$ covers
$E<6$ because $\mathcal G_{K_E}^{\mathrm{AMCR}}\le NH$.

The preceding bound concerns actual state visits. To obtain \eqref{eq:reinforce-episodic-gap}, we control what changes when the same policy lottery is evaluated at its discounted state frequencies. Fix a state $s$ and a bounded Borel test function $f$ on profiles. The variables
$f(p_e)[N_{e,s}-Y_ed_{p_e}(s)]$ have zero conditional mean by \eqref{eq:mbcce-episodic-states}. Their absolute values are bounded by $\|f\|_\infty Y_e$. Orthogonality, $\E Y_e^2\le2H^2$, and Cauchy--Schwarz with \eqref{eq:reinforce-inverse-horizon} give
\begin{equation*}
 \E\left|\int f(\Pi)\big[\widehat\mu_{K_E}(\{s\},d\Pi)-
 \bar\mu_{K_E}(\{s\},d\Pi)\big]\right|
 \le \frac{3\|f\|_\infty}{\sqrt{E}} .
\end{equation*}
For $E\le2$ this follows directly from $|N_{e,s}-Y_ed_{p_e}(s)|\le Y_e$. For the incentive tests we sum the martingale bounds before taking square roots to retain the dependence on $|\cS|$ in the theorem. Put $Z_{i,s,b,e}=\mathcal A_i(p_e;s,b)[N_{e,s}-Y_ed_{p_e}(s)]$. For each episode,
$\sum_s[N_{e,s}-Y_ed_{p_e}(s)]^2\le2Y_e^2$ because the two normalized state vectors are probability distributions. Since $|\mathcal A_i|\le H$, orthogonality yields
\[
 \sum_{i,s,b}\E\left[\left(\sum_{e\le E}Z_{i,s,b,e}\right)^2\right]
 \le 4NA_{\max}H^4 E .
\]
The absolute difference between the incentive gaps in \eqref{eq:reinforce-expected-rate} and \eqref{eq:reinforce-episodic-gap}, before taking expectations, is at most $K_E^{-1}\sum_{i,s}\max_b|\sum_{e\le E}Z_{i,s,b,e}|$. The squared numerator of this upper bound is at most $N|\cS|\sum_{i,s,b}(\sum_{e\le E}Z_{i,s,b,e})^2$. Applying Cauchy--Schwarz with \eqref{eq:reinforce-inverse-horizon} shows that the expected difference in absolute value is at most $C N H\sqrt{A_{\max}|\cS|/E}$ for $E\ge3$. Combining this with \eqref{eq:reinforce-expected-rate} proves \eqref{eq:reinforce-episodic-gap}. Both gaps are at most $NH$ when $E\le2$.

We next prove almost-sure convergence. The strong law gives $K_E/E\to H$.
Equation~\eqref{eq:reinforce-second-moment} and the schedules imply
\begin{align*}
\E\!\left[
\frac1E\sum_{e=1}^E\eta_e
\sum_{i=1}^N\|\widehat G_{i,e}\|_2^2
\right]
&=O(E^{-1/3}),
\\
\E\!\left[
\frac1E\sum_{e=1}^E\varepsilon_eY_e
\right]
&=O(E^{-1/3}).
\end{align*}
The two numerators are nonnegative and nondecreasing. For either numerator,
write it as $S_E$. Markov's inequality at $E=2^m$ gives
$\bbP(S_{2^m}/2^m>2^{-m/6})=O(2^{-m/6})$. The Borel--Cantelli lemma gives
$S_{2^m}/2^m\to0$, and monotonicity gives $S_E/E\to0$ between consecutive
dyadic horizons. Therefore, almost surely,
\begin{equation*}
\frac1E\sum_{e=1}^E\eta_e
\sum_{i=1}^N\|\widehat G_{i,e}\|_2^2
\longrightarrow0,
\qquad
\frac1E\sum_{e=1}^E\varepsilon_eY_e
\longrightarrow0.
\end{equation*}

For each fixed $(i,s,b_i)$,
\eqref{eq:reinforce-martingale-second-moment} and
$\varepsilon_e^{-1}=O(e^{1/3})$ imply
\begin{equation*}
\sum_{e=1}^{\infty}
\frac{\E[(\Delta M_{i,s,b_i,e})^2]}{e^2}
<\infty.
\end{equation*}
Hence the martingale
$\sum_{e=1}^E\Delta M_{i,s,b_i,e}/e$ has uniformly bounded second moments
and converges almost surely. Kronecker's lemma then gives
$M_{i,s,b_i}(E)/E\to0$ almost surely. There are finitely many
player--state--action triples. Since $1/(E\eta_E)\to0$ and $K_E/E\to H$,
\eqref{eq:reinforce-master-bound} now gives
\begin{equation}
\label{eq:reinforce-completed-episodes}
\mathcal G_{K_E}^{\mathrm{AMCR}}\longrightarrow0
\qquad\text{almost surely}.
\end{equation}

It remains to consider a horizon inside an episode. For sufficiently large
$T$, let $E_T:=\max\{E:K_E\le T\}$. Since each one-stage advantage belongs
to $[-H,H]$,
\begin{align*}
T\mathcal G_T^{\mathrm{AMCR}}
&\le
K_{E_T}\mathcal G_{K_{E_T}}^{\mathrm{AMCR}}
+NH(T-K_{E_T})
\\
&\le
K_{E_T}\mathcal G_{K_{E_T}}^{\mathrm{AMCR}}
+NHY_{E_T+1}.
\end{align*}
Thus
\begin{equation*}
\mathcal G_T^{\mathrm{AMCR}}
\le
\mathcal G_{K_{E_T}}^{\mathrm{AMCR}}
+NH\frac{Y_{E_T+1}}{K_{E_T}}.
\end{equation*}
The geometric tail and the Borel--Cantelli lemma imply $Y_{E+1}/E\to0$
almost surely. Together with $K_E/E\to H$, this gives
$Y_{E+1}/K_E\to0$. Equation~\eqref{eq:reinforce-completed-episodes} therefore gives vanishing AMCR at all horizons. Applying the episodic part of Proposition~\ref{prop:empirical-markov-bcce} proves the assertion about accumulation points.
\end{proof}

\section{Proofs for Section~\ref{sec:properties}}

\begin{proof}[Proof of Proposition \ref{prop:mbcce-smoothness}]
Fix a state $s$ in the support of the state marginal. Applying
\eqref{eq:markov-bcce-obedience} with $b_i=a_i^\star(s)$ gives
\[
\E_{\substack{\Pi\sim\mu(\cdot\mid s)\\
a\sim\Pi(\cdot\mid s)}}
\left[
Q_i^\Pi(s,(a_i^\star(s),a_{-i}))-Q_i^\Pi(s,a)
\right]
\le0
\]
for every player $i$. Summing over players and taking expectations in
\eqref{eq:markov-q-smoothness} therefore yields
\begin{align*}
0\ge{}&
\lambda\sum_{i=1}^N u_i(s,a^\star(s))
-(1+\kappa)
\E_{\substack{\Pi\sim\mu(\cdot\mid s)\\
a\sim\Pi(\cdot\mid s)}}
\left[\sum_{i=1}^N u_i(s,a)\right]\\
&-
\E_{\substack{\Pi\sim\mu(\cdot\mid s)\\
a\sim\Pi(\cdot\mid s)}}
[\xi(s,\Pi,a)].
\end{align*}
Rearranging proves \eqref{eq:mbcce-smoothness-bound}.
\end{proof}

\section{Proofs for Section~\ref{sec:converse}}

\begin{proof}[Proof of Theorem~\ref{thm:mbcce-converse}]
For necessity in part~(i), apply Theorem~\ref{thm:AMCR-regret-tracking} to the executed profiles and to the game restricted to reachable states. Proposition~\ref{prop:empirical-markov-bcce} then gives the episodic conclusion. For part~(ii), apply the same theorem on the full state space and use the continuing conclusion of the proposition, whose proof applies under \eqref{eq:mbcce-average-increments}.

\emph{Sufficiency in part~(i).} Let $\nu$ be the policy marginal of $\mu$. By \eqref{eq:mbcce-state-law} and \eqref{eq:mbcce-episodic-states}, an episode played under $\Pi$ has expected visits to $s$ equal to $d_\Pi(s)/(1-\gamma)$, and $\mu(\{s\},d\Pi)=d_\Pi(s)\nu(d\Pi)$. This is why drawing profiles according to $\nu$ before the reset produces the desired state frequencies.

Number the episodes $e=1,2,\ldots$ and their stages $k=0,\ldots,L_e$. Reserve the square-indexed episodes for estimating the model. Cycle through all pure stationary joint profiles on these episodes. On the other episodes, draw a prescribed profile independently from $\nu$ before reset. We use these draws to select a deterministic sequence of prescriptions at the end of the proof. Every learning rule plays its prescribed action on a reserved episode, regardless of any earlier switch. Since these joint actions are prescribed in advance, model estimation will require only the state and each player's own observations. These episodes account for $o(E)$ stages through episode $E$, almost surely: the maximum of the first $E$ geometric lengths is $O(\log(E+2))$ almost surely by the geometric tail and Borel--Cantelli on dyadic horizons.

On a nonreserved episode, player $i$ uses the prescribed component until her score test calls for a switch, as specified below. Writing $\pi_{i,e}$ for her target in either case, she executes $p_{i,e}=(1-\varepsilon_e)\pi_{i,e}+\varepsilon_e\mathsf U_i$ with $\varepsilon_e=1/\log(e+2)$ and $\mathsf U_i$ uniform on $A_i$. From her own data she computes, for each $s,b$,
\begin{equation}
\label{eq:attainment-score}
 Z_{i,e}(s,b)=\sum_{k=0}^{L_e}
 \frac{\one\{s_{e,k}=s,\ a_{i,e,k}=b\}}{p_{i,e}(b\mid s)}
 \sum_{\ell=k}^{L_e}u_i(s_{e,\ell},a_{e,\ell}).
\end{equation}
Only the observed payoffs, states and the player's own actions enter \eqref{eq:attainment-score}. Let $q_{i,e}(s,b)=\sum_{a_{-i}}p_{-i,e}(a_{-i}\mid s)Q_i^{p_e}(s,(b,a_{-i}))$. At a visit to $s$, the denominator in \eqref{eq:attainment-score} cancels the probability of choosing $b$, and the geometric remaining lifetime gives the discounted continuation value. Conditional on the past and all policies chosen before reset, the score and $\sum_{k\le L_e}\one\{s_{e,k}=s\}q_{i,e}(s,b)$ consequently have the same mean, namely $H d_{p_e}(s)q_{i,e}(s,b)$. The conditional second moment of $Z_{i,e}(s,b)$ is $O(\log^2(e+2))$ since its absolute value is at most $|A_i|(L_e+1)^2/\varepsilon_e$.

Divide episodes into epochs $2^j\le e<2^{j+1}$. At the start of each epoch, return to the prescription and reset each test to zero. For each $s,b$, add the score difference $Z_{i,e}(s,b)-\langle\pi_{i,e}(\cdot\mid s),Z_{i,e}(s,\cdot)\rangle$ over nonreserved episodes. If any sum exceeds $2^j/\sqrt{j+1}$, switch for the rest of this epoch to projected gradient ascent on these scores at each state, starting afresh with step $2^{-j/2}$. All updates occur between episodes.

For joint play, consider a hypothetical epoch in which all players stick to the prescriptions even if a test crosses its threshold. This avoids conditioning on the outcome of the tests. Given the history before the next nonreserved profile is drawn, the conditional mean of each signed score is at most $C/(j+1)$. Indeed, the conditional-mean expression for \eqref{eq:attainment-score} extends continuously to zero mixing, where the signed mean equals $(1-\gamma)^{-1}\int\mathcal A_i(\Pi;s,b)\,\mu(\{s\},d\Pi)\le0$. Perturbing each policy by $\varepsilon_e$ changes both its discounted state frequencies and the action values by $O(\varepsilon_e)$ uniformly over profiles, by \eqref{eq:mbcce-episodic-states} and the discounted Bellman equations. For all sufficiently large $j$ the total conditional means are at most half the switching threshold. Doob's inequality and the second-moment bound give a probability at most $C(j+1)^3 2^{-j}$ that any test would cross its threshold in this epoch. The actual and prescribed plays can use identical draws up to the first player's switch. This probability bound therefore applies from every previous history, including histories after earlier switches. The bound is summable in $j$. Hence, almost surely, only finitely many epochs of actual joint play contain switches.

Under prescribed play, an episode sum of a bounded continuous test $f(s,p_e)$ has uniformly bounded second moment. Its mean on nonreserved episodes tends to $(1-\gamma)^{-1}\int f\,d\mu$ by continuity of $f$, \eqref{eq:mbcce-episodic-states}, and the vanishing probability of action exploration. Using independent reset and transition draws for each episode, the actual play differs from this prescribed experiment only in finitely many epochs. The martingale strong law for episode sums, the strong law for episode lengths, and a countable dense family of continuous tests give $\widehat\mu_T\Rightarrow\mu$ almost surely. An unfinished episode contributes $o(T)$, so this holds along every horizon.

We next analyze one player against arbitrary opponents who choose policies before reset and hold them fixed through each episode. Fix any sequence of prescriptions, without requiring that her opponents follow it. In epoch $j$, almost surely the largest score vector has norm $O((j+1)^3)$ by the bound on episode lengths. Hence each test's excess at a first crossing is at most that order. On the remainder of the epoch, the projection inequality bounds positive external regret at each state through any prefix by
$C2^{j/2}+C2^{-j/2}\sum_e\|Z_{i,e}(s,\cdot)\|_2^2=O((j+1)^2 2^{j/2})$ almost surely. The equality follows from $\sum_{e<2^{j+1}}(L_e+1)^4=O(2^j)$ almost surely. Through an arbitrary prefix of epoch $j$, the positive cumulative estimated gain against any fixed action is bounded by the threshold, the one-episode overshoot, and the fallback regret. Each is $o(2^j)$ as $j$ grows. Summing the bounds over epochs, including the current incomplete epoch, gives $o(E)$ through any episode $E$.

For each deviation, replacing its estimated score sum by the corresponding sum of $q_{i,e}$ on realized visits changes it by a sum of martingale differences with conditional second moments $O(\log^2(e+2))$. This remains true after multiplication by each weight in $\pi_{i,e}(\cdot\mid s)$, which is chosen before reset. The martingale strong law makes the difference $o(E)$ almost surely. Replacing $\pi_{i,e}$ by the executed $p_{i,e}$ in the obedient payoff changes the sum by at most $C\sum_{e\le E}\varepsilon_e(L_e+1)=o(E)$ almost surely, using the strong law for lengths and $\varepsilon_e\to0$. Reserved episodes and a final unfinished episode also cost $o(E)$. Since the number of stages per completed episode tends to $1/(1-\gamma)$, every player's own positive deviation gains divided by $T$ vanish against the specified opponents. The individual guarantee holds almost surely for each fixed prescription and each permitted opponent rule. For joint play, the preceding argument averaged over independent profile prescriptions. Fubini's theorem now selects a deterministic sequence for which joint play attains $\mu$ and only finitely many epochs contain switches, almost surely over actions, transitions and resets.

It remains to construct online estimates of the Bellman tables in joint play. Each reachable $(s,a)$ can be visited in a reserved episode with positive probability. Choose a simple positive-probability path from the reset support to $s$, prescribe the required joint actions along this path and action $a$ at its end, and extend these choices to a pure stationary profile. Episodes with this profile recur infinitely often and start independently. Therefore each reachable pair is visited infinitely often, also on stages when the episode continues. In reserved episodes the joint action is known from the prescription and the state. Each player can thus estimate its transition row and her own reward. The estimates converge almost surely by the strong law for bounded transition observations on the visit clock. Discounted Bellman contraction then gives convergence of the estimated action and state values, uniformly over profiles on all reachable rows. Before a row has been observed, use any probability distribution on the state space for that row and any reward in $[0,1]$. Once a row is observed, estimate it from the reserved episodes, using the transition observations on stages that continue. Before each episode each player evaluates the publicly prescribed behavior profile using her own estimates from earlier reserved episodes, including the prescribed exploration on nonreserved episodes. In joint play this is the actual profile except in finitely many epochs. The average evaluation error consequently vanishes. The exact-evaluation identity following Theorem~\ref{thm:AMCR-regret-tracking}, or its two-sided error bound for these estimated tables, converts vanishing AMCR into $o(T)$ aggregate external regret. If the limiting frequency of $s$ is positive then $N_s(T)/T$ tends to that frequency, giving the claimed local bound. With an exact evaluator the tables on unreachable rows may also be supplied. Trajectory data alone cannot identify their rewards or transition probabilities.

\emph{Sufficiency in part~(ii).} Fix a continuing-play MBCCE $\mu$, and write $\mu(\{s\},d\Pi)=d_\Pi(s)\nu(d\Pi)$, where $d_\Pi P_\Pi=d_\Pi$ for $\nu$-almost every profile. We construct one continuing path. Uniform joint actions will be used to reach selected states, without any reset of the state process.

Assumption~\ref{ass:communication} gives a uniform bound on the time needed for this purpose. Since the state space is finite and $P^{\mathrm U}$ is irreducible, there are an integer $h\ge1$ and $a>0$ such that, from any state, uniform joint actions reach any designated state within $h$ transitions with probability at least $a$. Applying this bound after each successive group of $h$ transitions gives a geometric tail bound for every such hitting time, uniformly in its initial and designated states and in the preceding history. In particular, their conditional first and second moments are uniformly bounded.

Number the cycles $m=1,2,\ldots$, and let $(r_m,a^m)$ cycle deterministically through $\cS\times\cA$. First, all players use uniform actions until the public state is $r_m$. They then execute the prescribed pure joint action $a^m$ once and observe the next state; each player also observes her own payoff. The profile at this observation date may be taken to prescribe $a^m$ at every state. Next, draw $(S_m,\Pi^m)$ independently from $\mu$, independently of all preceding data. All players again use uniform actions until the state is $S_m$, and then follow $\Pi^m$ for exactly $m$ stages. The draw will be replaced by a deterministic sequence of prescriptions at the end of the proof. Every stopping decision uses only the public state. At each date the full profile being executed is specified by the procedure.

Let $\mathscr F_m$ contain the history before cycle $m$, and let $Z_m(f)$ be the sum of a bounded continuous function $f(s,p)$ over the $m$ stages played under $\Pi^m$. Conditional on the selected profile, $S_m$ has distribution $d_{\Pi^m}$. Upon reaching $S_m$, the subsequent transitions under $\Pi^m$ have the law of that profile started at $S_m$, by the Markov property at the hitting time. Since $d_{\Pi^m}$ is invariant, averaging over $S_m$ gives that distribution at every stage of the prescribed block. Consequently,
\begin{equation}
 \E[Z_m(f)\mid\mathscr F_m]=m\int f\,d\mu.
 \label{eq:continuing-attainment-block}
\end{equation}
This argument does not require the chain under $\Pi^m$ to be irreducible or aperiodic. Drawing $S_m$ also selects any recurrent class with precisely the weight assigned to it by $d_{\Pi^m}$.

Put $B_m=\sum_{j=1}^m j=m(m+1)/2$. The differences $Z_m(f)-m\int f\,d\mu$ form martingale differences by \eqref{eq:continuing-attainment-block} and have conditional second moments at most $4m^2\|f\|_\infty^2$. Since $\sum_m m^2/B_m^2<\infty$, martingale convergence and Kronecker's lemma give
\begin{equation*}
 \frac{1}{B_M}\sum_{m=1}^M Z_m(f)\longrightarrow\int f\,d\mu
 \qquad\text{almost surely}.
\end{equation*}
Let $D_m$ be the number of other stages in cycle $m$, including the two periods of uniform play and the single prescribed observation. The hitting-time bounds give uniformly bounded conditional means and second moments for $D_m$ given $\mathscr F_m$. The martingale strong law applied to $D_m-\E[D_m\mid\mathscr F_m]$ therefore gives $\sum_{m\le M}D_m=O(M)$ almost surely. The geometric tail bound also gives $D_m=O(\log(m+2))$ almost surely, by Borel--Cantelli. Thus these stages occupy a vanishing share of time, and the length $m+D_m$ of an unfinished cycle is negligible relative to $B_{m-1}$. Applying the preceding convergence to a countable uniformly dense family of continuous functions on the compact space $\cS\times\mathcal C$ proves $\widehat\mu_T\Rightarrow\mu$ almost surely along every horizon.

The procedure changes the executed profile at most a fixed number of times in each cycle. The diameter of $\mathcal C$ is finite, so the total profile variation through cycle $M$ is $O(M)$. Since elapsed time is at least $B_M$, and the same estimate applies within a cycle, \eqref{eq:mbcce-average-increments} holds.

Every state--action pair is used in the prescribed observations infinitely often. At such an observation the full joint action is known, so each player can estimate the corresponding transition row and her own reward from her observations alone. The bounded martingale strong law on each observation sequence gives convergence of every estimated transition row and reward. Before a row is observed, use any probability distribution for its transition row and any reward in $[0,1]$. At each date each player evaluates the currently prescribed profile using only estimates obtained before that date. Because there are finitely many rows, their estimation errors converge uniformly to zero. Discounted Bellman contraction then implies that the estimated action and state values converge to their true values uniformly over all stationary profiles. In particular, $e_t\to0$ almost surely along the constructed path.

The advantages $\mathcal A_i(\Pi;s,b_i)$ are continuous in $\Pi$ and bounded for the fixed discount factor. Since $\widehat\mu_T\Rightarrow\mu$ and $\mu$ satisfies obedience, \eqref{eq:AMCR-empirical-mbcce} gives $\mathcal G_T^{\mathrm{AMCR}}\to0$. Theorem~\ref{thm:AMCR-regret-tracking} and the vanishing evaluation error now give $o(T)$ aggregate statewise external regret. At a state of positive limiting frequency, division by $N_s(T)$ gives the asserted local conclusion as well.

All these statements hold almost surely jointly over the independent draws $(S_m,\Pi^m)$ and the randomness of play. Fubini's theorem selects a deterministic sequence of pairs for which they hold almost surely over play. The players can agree on this sequence and the observation sequence in advance. They then need only the public states, their own actions and payoffs, and their past observations to implement the construction. This proves part~(ii).
\end{proof}

\section{Proofs for Section~\ref{sec:mcce}}
\label{app:mcce}

\begin{proof}[Proof of Lemma~\ref{lem:exact-mbcce-mcce}]
Fix $i,s$, and $b_i$. By \eqref{eq:mbcce-induced-action-law},
\eqref{eq:averaged-one-step-gain}, and the Bellman equations in \eqref{eq:canonical-bellman-system},
$g_i^\mu(s,b_i)$ is the expectation under
$\Pi\sim\mu(\cdot\mid s)$ and $a\sim\Pi(\cdot\mid s)$ of
$Q_i^{\sigma_\mu}(s,(b_i,a_{-i}))-Q_i^{\sigma_\mu}(s,a)$. For every
$(\Pi,a)$, the Bellman equations give
\begin{align*}
&Q_i^{\sigma_\mu}(s,(b_i,a_{-i}))-Q_i^{\sigma_\mu}(s,a)
\\
&\quad=
Q_i^\Pi(s,(b_i,a_{-i}))-Q_i^\Pi(s,a)
+
\gamma\left\langle
P(\cdot\mid s,(b_i,a_{-i}))-P(\cdot\mid s,a),
V_i^{\sigma_\mu}-V_i^\Pi
\right\rangle.
\end{align*}
Averaging proves \eqref{eq:exact-mbcce-mcce-local}.

Fix a deterministic stationary deviation $\beta_i$, and let
$P_i^{\beta_i,\sigma_{\mu,-i}}$ be its transition matrix against
$\sigma_{\mu,-i}$. Its Bellman equation and
\eqref{eq:averaged-one-step-gain} imply
\begin{align*}
V_i^{\beta_i,\sigma_{\mu,-i}}-V_i^{\sigma_\mu}
={}&
g_i^\mu(\cdot,\beta_i(\cdot))
+
\gamma P_i^{\beta_i,\sigma_{\mu,-i}}
\left(V_i^{\beta_i,\sigma_{\mu,-i}}-V_i^{\sigma_\mu}\right).
\end{align*}
Iterating this identity and using
\eqref{eq:normalized-deviation-occupancy} gives
\begin{equation*}
(1-\gamma)
\left[V_i^{\beta_i,\sigma_{\mu,-i}}(s_0)-V_i^{\sigma_\mu}(s_0)\right]
=
\sum_s\omega_{i,s_0}^{\beta_i,\sigma_{\mu,-i}}(s)
 g_i^\mu(s,\beta_i(s)).
\end{equation*}
Maximizing over $i,s_0,\beta_i$ proves
\eqref{eq:exact-mbcce-mcce}.
\end{proof}

\begin{proof}[Proof of Corollary~\ref{cor:algorithm-empirical-mcce}]
We first bound the MCCE gap of the average behavior law $\sigma_T$.

Consider case~(i). The coverage argument for Algorithm~\ref{alg:bao-ac} gives infinitely many visits to each state almost surely. Fix $i$ and $s$. Summing \eqref{eq:appendix-AMCR-pointwise} over the first $n$ visits gives
\begin{equation*}
\frac1n
\left(
\max_{b_i}\sum_{q=1}^n\mathcal A_i(\Pi_{t_q(s)};s,b_i)
\right)_+
\le
\frac{R_{i,s}^{\mathrm{tar}}(n)}n+
\frac2n\sum_{q=1}^n \|X_{t_q(s)}-\Lambda(\Pi_{t_q(s)})\|_\infty.
\end{equation*}
The statewise regret estimate makes the first term vanish almost surely. In particular, the exploration contribution is bounded by a constant times $n^{-1}\sum_{q=1}^n\varepsilon_{t_q(s)}$, which tends to zero. Pointwise critic tracking makes the second term vanish along the visits to each state. Finally, Lemma~\ref{lem:AMCR-advantage-lipschitz} bounds the difference between the behavior and target advantages at date $t$ by a constant times $\varepsilon_t$. Since there are finitely many player--state--action triples, we obtain
\begin{equation*}
\max_{i,s,b_i}
\left(
\frac1{N_s(T)}\sum_{t\in T_s(T)}\mathcal A_i(p_t;s,b_i)
\right)_+\longrightarrow0
\qquad\text{almost surely}.
\end{equation*}
Apply \eqref{eq:mbcce-mcce-upper-bound} to $\widehat\mu_T$ to obtain \eqref{eq:algorithm-empirical-mcce} in case~(i), initially for $\sigma_T$.

Consider case~(ii). For $s\in\cS$, let
$C_s(E):=\sum_{e=1}^E\one\{s_{e,0}=s\}$. The strong law gives
$C_s(E)/E\to\rho_0(s)$. The proof of
Theorem~\ref{thm:reinforce-AMCR} gives $K_E/E\to H$ and
$(L_{E+1}+1)/E\to0$. Hence, for $E_T:=\max\{E:K_E\le T\}$,
$T/E_T\to H$. Since the first $E_T$ episode starts occur among the first
$T$ played stages,
\begin{equation}
\label{eq:reinforce-state-frequency}
\liminf_{T\to\infty}\frac{N_s(T)}T
\ge
\frac{\rho_0(s)}H
>0
\qquad\forall s\in\cS
\end{equation}
almost surely. For every $i,s,b_i$,
\begin{equation*}
\left(
\frac1{N_s(T)}\sum_{t\in T_s(T)}
\mathcal A_i(p_t;s,b_i)
\right)_+
\le
\frac{T}{N_s(T)}\mathcal G_T^{\mathrm{AMCR}}.
\end{equation*}
Theorem~\ref{thm:reinforce-AMCR} and
\eqref{eq:reinforce-state-frequency} make the right-hand side converge to
zero uniformly over the finitely many triples $(i,s,b_i)$. Applying
\eqref{eq:mbcce-mcce-upper-bound} to $\widehat\mu_T$ proves the case~(ii) bound
in \eqref{eq:algorithm-empirical-mcce} with $\widehat\sigma_T$ replaced by
$\sigma_T$.

It remains to replace the average behavior law by realized empirical play.
Let $\mathcal J_t$ be the information available immediately before the
$t$th action in the applicable sampling protocol. For
Algorithm~\ref{alg:bao-ac}, $\mathcal J_t=\mathcal H_t$; for
Algorithm~\ref{alg:reinforce-AMCR}, it is the corresponding episodic
pre-action history. In either case, $p_t$ is $\mathcal J_t$-measurable and
$\bbP(a_t=a\mid\mathcal J_t)=p_t(a\mid s_t)$.

Fix $s$ and $a$, and let $\mathcal K_{s,q}$ be the history immediately after
the action at the $q$th visit to $s$, with $\mathcal K_{s,0}$ trivial. Then
$\mathcal K_{s,q-1}\subseteq\mathcal J_{t_q(s)}$, and
\begin{equation*}
\E\!\left[
\one\{a_{t_q(s)}=a\}-p_{t_q(s)}(a\mid s)
\,\middle|\,
\mathcal K_{s,q-1}
\right]
=0.
\end{equation*}
The martingale strong law and infinite visitation imply
$\max_s\|\widehat\sigma_T^s-\sigma_T^s\|_1\to0$ almost surely.

Let $\delta_T:=\max_s\|\widehat\sigma_T^s-\sigma_T^s\|_1$. For every
$v\in[0,H]^{|\cS|}$, the Bellman operators induced by
$\widehat\sigma_T$ and $\sigma_T$ differ by at most $H\delta_T$ in sup norm;
the same bound holds for the operators induced by any fixed deviation
$\beta_i$ against their opponents' marginals. Since these operators are
$\gamma$-contractions,
\begin{equation*}
\max_i\left\{
\|V_i^{\widehat\sigma_T}-V_i^{\sigma_T}\|_\infty,
\sup_{\beta_i}
\|V_i^{\beta_i,\widehat\sigma_{T,-i}}
-V_i^{\beta_i,\sigma_{T,-i}}\|_\infty
\right\}
\le
H^2\delta_T.
\end{equation*}
Therefore,
$|\GapMCCE(\widehat\sigma_T)-\GapMCCE(\sigma_T)|
\le2H\delta_T\to0$ almost surely. This proves
\eqref{eq:algorithm-empirical-mcce}.
\end{proof}

For $v\in\R^{|\cS|}$, write
$\operatorname{sp}(v):=\max_{r\in\cS}v(r)-\min_{r\in\cS}v(r)$.

\begin{proof}[Proof of Proposition~\ref{prop:transition-aggregative-mcce}]
Fix $i,s,b_i$, and $\Pi$, and draw $a\sim\Pi(\cdot\mid s)$. By
\eqref{eq:transition-aggregative-kernel} and
\eqref{eq:transition-aggregative-lipschitz},
\begin{equation*}
\TV\!\left(
P(\cdot\mid s,(b_i,a_{-i})),P(\cdot\mid s,a)
\right)
\le
L_sw_i\|\phi_{i,s}(b_i)-\phi_{i,s}(a_i)\|
\le
L_sw_i\Delta_{i,s}.
\end{equation*}
For probability vectors $\alpha,\alpha'$ and any vector $v$,
$|\langle\alpha-\alpha',v\rangle|
\le\TV(\alpha,\alpha')\operatorname{sp}(v)$. In
$\mathsf C(\widehat\mu_T)$, the averaged law is $\sigma_T$ and each sampled
policy is one of the behavior profiles $p_t$. Since their canonical values belong to $[0,H]^{|\cS|}$,
$\operatorname{sp}(V_i^{\sigma_T}-V_i^{p_t})\le2H$. Substitution into
\eqref{eq:continuation-discrepancy} gives
\[
\mathsf C(\widehat\mu_T)
\le
\frac{2}{1-\gamma}
\max_{\substack{i,\,s_0\in\cS\\ \beta_i:\cS\to A_i}}
\sum_s
\omega_{i,s_0}^{\beta_i,\sigma_{T,-i}}(s)
L_sw_i\Delta_{i,s}.
\]
Because $\omega_{i,s_0}^{\beta_i,\sigma_{T,-i}}$ is a probability distribution
over $\cS$,
\[
\mathsf C(\widehat\mu_T)
\le
\frac{2}{1-\gamma}
\max_{i,s}L_sw_i\Delta_{i,s}.
\]
Corollary~\ref{cor:algorithm-empirical-mcce} proves the claimed bounds.

Under \eqref{eq:transition-common-minorization}, any stationary transition
matrix satisfies
$P_\sigma=\eta\mathbf 1q^\top+(1-\eta)\widetilde P_\sigma$
for some stochastic matrix $\widetilde P_\sigma$. Hence
$\operatorname{sp}(P_\sigma v)\le(1-\eta)\operatorname{sp}(v)$ for every
vector $v$. Since
$V_i^\sigma=r_i^\sigma+\gamma P_\sigma V_i^\sigma$ and
$r_i^\sigma\in[0,1]^{|\cS|}$,
\[
\operatorname{sp}(V_i^\sigma)
\le
1+\gamma(1-\eta)\operatorname{sp}(V_i^\sigma),
\]
so
$\operatorname{sp}(V_i^\sigma)\le[1-\gamma(1-\eta)]^{-1}$.
Therefore,
$\operatorname{sp}(V_i^{\sigma_T}-V_i^{p_t})
\le2[1-\gamma(1-\eta)]^{-1}$, which gives the sharper bound.
\end{proof}

\begin{proof}[Proof of Corollary~\ref{cor:transition-aggregative-mbcce-welfare}]
Fix $s$, $\Pi$, and $a$. The Bellman equations give
\begin{align}
&\sum_{i=1}^N
\left[
Q_i^\Pi(s,(a_i^\star(s),a_{-i}))-Q_i^\Pi(s,a)
\right]
\notag\\
&\quad=
\sum_{i=1}^N
\left[
u_i(s,(a_i^\star(s),a_{-i}))-u_i(s,a)
\right]
\notag\\
&\qquad+
\gamma\sum_{i=1}^N
\left\langle
P(\cdot\mid s,(a_i^\star(s),a_{-i}))-P(\cdot\mid s,a),
V_i^\Pi
\right\rangle.
\label{eq:transition-aggregative-q-decomposition}
\end{align}
By \eqref{eq:stage-utility-smoothness}, the first term on the right-hand side
of \eqref{eq:transition-aggregative-q-decomposition} is at least
$\lambda\sum_i u_i(s,a^\star(s))-(1+\kappa)\sum_i u_i(s,a)$.
Moreover, Definition~\ref{def:transition-aggregative-game} gives
\[
\TV\!\left(
P(\cdot\mid s,(a_i^\star(s),a_{-i})),
P(\cdot\mid s,a)
\right)
\le L_sw_i\Delta_{i,s}.
\]
Since $V_i^\Pi\in[0,H]^{|\cS|}$, the continuation term in
\eqref{eq:transition-aggregative-q-decomposition} is therefore bounded below by
$-\gamma H L_s\sum_i w_i\Delta_{i,s}$.
Proposition~\ref{prop:mbcce-smoothness} now proves
\eqref{eq:transition-aggregative-mbcce-welfare-uniform}.

Under \eqref{eq:transition-common-minorization}, the range of every
$V_i^\Pi$ is at most $[1-\gamma(1-\eta)]^{-1}$ by the argument above, so the
same proof gives the sharper bound.
\end{proof}

\bibliographystyle{abbrvnat}
\bibliography{AMCR}

\end{document}